\documentclass[11pt]{article}
\usepackage{amssymb, amsthm}
\usepackage{algorithm}
\usepackage{algpseudocode}
\usepackage{amsmath, amssymb, amsthm}
\usepackage{float,multirow,subcaption, setspace}
\usepackage{graphicx}
\usepackage{comment}
\usepackage{enumitem}
\usepackage{tikz}
\usepackage{bm, bbm}
\usetikzlibrary{shapes.geometric, positioning}
\usetikzlibrary{quotes, angles}
\usepackage{rotating}
\usepackage{hyperref}
\usepackage{natbib}
\usepackage{multicol}
\usepackage{mathtools}
\usepackage{nicefrac}
\usepackage{mathrsfs}

\definecolor{SkyBlue}{RGB}{14, 118, 188}
\definecolor{BrightRed}{RGB}{223, 82, 78}
\definecolor{Green638}{RGB}{165,255,118} 
\definecolor{BurntOrange}{HTML}{BF5900}

\newcommand{\R}{\mathbb{R}} 
\newcommand{\E}{\mathbb{E}} 
\def\P{\mathbb{P}} 

\newcommand{\calP}{\mathcal{P}} 

\newcommand{\calC}{\mathcal{C}}
\newcommand{\calA}{\mathcal{A}}

\newcommand{\calE}{\mathcal{E}}
\newcommand{\calL}{\mathcal{L}}
\newcommand{\calM}{\mathcal{M}}
\newcommand{\calT}{\mathcal{T}}

\newcommand{\calI}{\mathcal{I}}

\newcommand{\calH}{\mathcal{H}}

\newcommand{\hmean}[1]{\textrm{hm}\left(#1\right)} 

\newcommand{\normaldist}[2]{\mathcal{N}\left(#1,#2\right)} 
\newcommand{\gammadist}[2]{\textrm{Gamma}\left(#1,#2\right)} 
\newcommand{\poisdist}[1]{\textrm{Poisson}\left(#1\right)} 
\newcommand{\unifdist}[2]{\textrm{Uniform}\left(#1,#2\right)} 

\newcommand{\by}{\bm{y}}
\newcommand{\bx}{\bm{x}}
\newcommand{\bz}{\bm{z}}

\newcommand{\br}{\bm{r}}

\newcommand{\bbeta}{\boldsymbol{\beta}}
\newcommand{\bmu}{\boldsymbol{\mu}}
\newcommand{\bomega}{\boldsymbol{\omega}}

\newcommand{\balpha}{\boldsymbol{\alpha}}

\theoremstyle{plain}
\newtheorem{theorem}{Theorem}

\newtheorem{lemma}{Lemma}

\theoremstyle{definition}
\newtheorem{definition}{Definition}

\newtheorem{ex}{Example}

\theoremstyle{plain}
\newtheorem{remark}[theorem]{Remark}

\usepackage{fullpage, parskip} 
\usepackage{cleveref}

\hypersetup{pdfborder = {0 0 0.5 [3 3]}, colorlinks = true, linkcolor = BrightRed, citecolor = SkyBlue}

\title{Scalable piecewise smoothing with BART}
\author{Ryan Yee\thanks{Department of Statistics, University of Wisconsin--Madison. \url{ryee2@wisc.edu}} \and Soham Ghosh\thanks{Department of Statistics, University of Wisconsin--Madison. \url{sghosh39@wisc.edu}} \and Sameer K.\ Deshpande\thanks{Department of Statistics, University of Wisconsin--Madison. \url{sameer.deshpande@wisc.edu}} }

\newcommand{\blind}{0}
\def\codelinkblind{\texttt{link to GitHub repository blinded for review}}
\def\codelinkunblind{\url{https://github.com/ryanyee3/ridgeBART}}
\def\codelink{
  \if\blind1
  \codelinkblind
  \else
  \codelinkunblind
  \fi
}

\newcommand{\jrssb}{0}

\newcommand{\switchref}[2]{%
  \if\jrssb0%
    #1%
  \else%
    #2%
\fi
}

\newcommand{\includesupp}{1}
\newcommand{\suppref}[2]{%
  \if\includesupp1%
    #1%
  \else%
    #2%
\fi
}

\begin{document}
\maketitle

\begin{abstract}
Although it is an extremely effective, easy-to-use, and increasingly popular tool for nonparametric regression, the Bayesian Additive Regression Trees (BART) model is limited by the fact that it can only produce discontinuous output.
Initial attempts to overcome this limitation were based on regression trees that output Gaussian Processes instead of constants. 
Unfortunately, implementations of these extensions cannot scale to large datasets. 
We propose ridgeBART, an extension of BART built with trees that output linear combinations of ridge functions (i.e., a composition of an affine transformation of the inputs and non-linearity); that is, we build a Bayesian ensemble of localized neural networks with a single hidden layer.
We develop a new MCMC sampler that updates trees in linear time and establish posterior contraction rates for estimating piecewise anisotropic H\"{o}lder functions and nearly minimax-optimal rates for estimating isotropic H\"{o}lder functions.
We demonstrate ridgeBART's effectiveness on synthetic data and use it to estimate the probability that a professional basketball player makes a shot from any location on the court in a spatially smooth fashion.

\end{abstract}

\section{Introduction}
\label{sec:intro}
\subsection{Motivation: Smooth NBA shot charts}
\label{sec:motivation}
In basketball, high-resolution player tracking data facilitates the creation of \emph{shot charts}, which visualize the probability with which a player makes a shot (hereafter field goal percentage or FG\%) from every point on the court.
Shot charts provide easy-to-interpret summaries of players' abilities and tendencies that can be used by coaches to design offensive and defensive strategies \citep{Hu2023}. 

\switchref{\Cref{fig:curry_shot_chart,fig:simmons_shot_chart}}{\Cref{fig:intro_shots}(A) and \Cref{fig:intro_shots}(B)} show the locations of made and missed shots for two players, Stephen Curry and Ben Simmons, from the 2023-24 National Basketball Association season.
One way to estimate shot charts is to divide the court into several small regions and to compute the proportion of shots each player makes in each region.
Unfortunately, such an approach can produce sharp discontinuities.
Consider region A in \switchref{\Cref{fig:curry_shot_chart}}{\Cref{fig:intro_shots}(A)}: Curry made three of his six attempts in this region but made all five of his attempts in the areas immediately to the right of region A.
The resulting sharp transition from 50\% to 100\% is highly undesirable as players' shooting abilities vary in a spatially smooth fashion \citep{Yin2022_shotchart,Yin2023_matrix}.
A further difficulty lies in predicting FG\% at locations where a player has not attempted a shot.
As an extreme example, Ben Simmons did not attempt a single shot within or near the region B in \switchref{\Cref{fig:simmons_shot_chart}}{\Cref{fig:intro_shots}(B)}.
Simple binning and averaging precludes principled prediction of Simmons' FG\% at these locations.

Bayesian hierarchical models are a conceptually simple way to overcome these difficulties.
For example, one can express each player's probit-transformed FG\% as a linear combination of fixed functions of locations (e.g., a tensor product of splines) and specify an exchangeable prior for the player-specific coefficients.
Such a model ``borrows strength'' across players so that every player's estimated FG\% at a particular location is informed not only by their own data but also by the estimated FG\%'s of all other players at that location.

Despite its merits, an exchangeable model can potentially smooth over important differences between players.
Specifically, players of different positions (i.e., guard, forward, center) often have vastly different tendencies and skillsets.
To account for these differences, it is tempting to first cluster players by position and fit separate exchangeable models within these clusters. 
This way, data about a tall center who only attempts shots near the basket is not used to make predictions about a smaller guard who predominantly shoots three-point shots.
Such an approach, which accounts for interactions between position and location, is unfortunately complicated by the extreme variation observed within positional groups. 
For instance, both Curry and Simmons are listed as point guards but are very different physically and in their style of play.
Simmons is listed at 6' 10", 240 lbs., which is above the league-average; uses his size to drive to the basket; and is considered a very poor shooter.
Curry, meanwhile, is listed at 6' 2", 185 lbs., which is below the league-average; and is considered one of the best shooters ever.

Ultimately, FG\% depends on complex interactions between location, player, and player characteristics like position, height, and weight.
The cases of Curry and Simmons underscore the difficulties faced in trying to specify these interactions correctly in a parametric model.
On this view, Bayesian Additive Regression Trees \citep[BART;][]{Chipman2010}, which approximates regression functions using an ensemble of binary regression trees, offers a compelling alternative.
With BART, users often obtain excellent predictive results and reasonably calibrated uncertainties ``off-the-shelf'' without needing to pre-specify the functional form of the regression function and with very little hyperparameter tuning.
{
\if\jrssb1%

\begin{figure}[ht]
\centering
\includegraphics[width = \textwidth]{figures/fig1}
\caption{Locations of shot attempts for Stephen Curry (A) and Ben Simmons (B). BART-estimated FG\% as a function for location for Curry (C) and Simmons (D).}
\label{fig:intro_shots}
\end{figure} %

\else%
\begin{figure}[ht]
\centering
\begin{subfigure}[b]{0.23\textwidth}
\centering
\includegraphics[width = \textwidth]{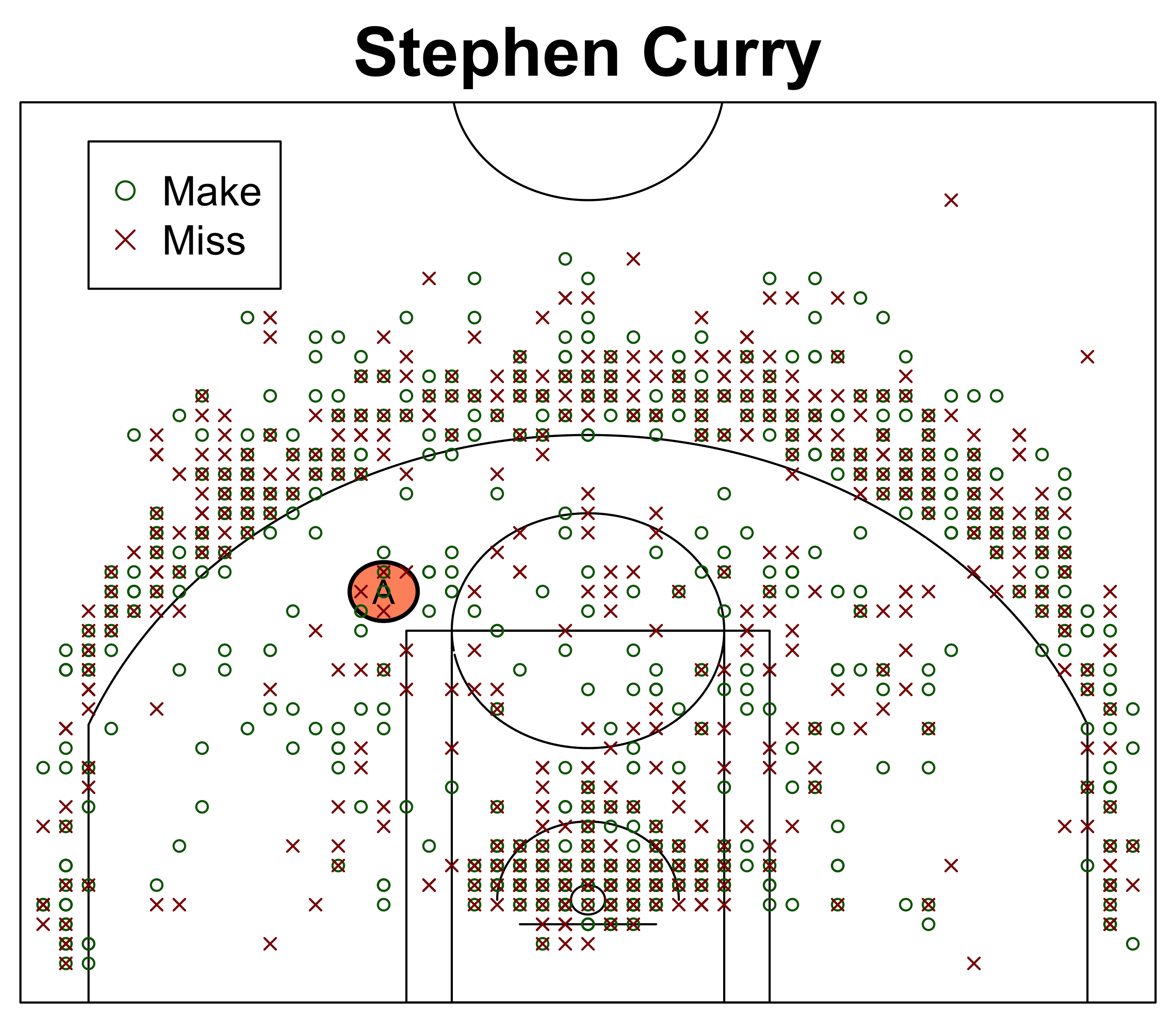}
\caption{}
\label{fig:curry_shot_chart}
\end{subfigure}
\begin{subfigure}[b]{0.23\textwidth}
\centering
\includegraphics[width = \textwidth]{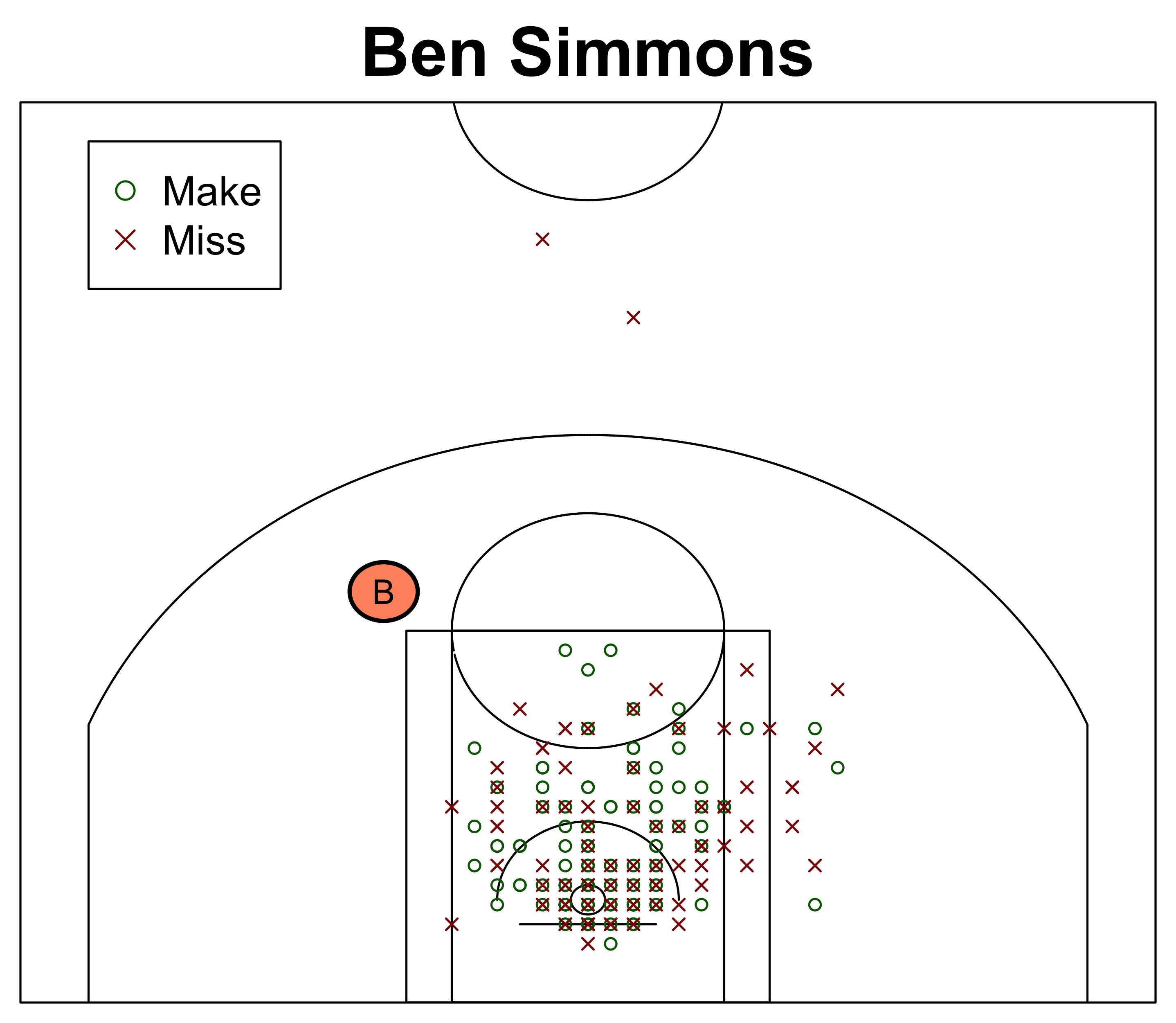}
\caption{}
\label{fig:simmons_shot_chart}
\end{subfigure}
\begin{subfigure}[b]{0.23\textwidth}
\centering
\includegraphics[width = \textwidth]{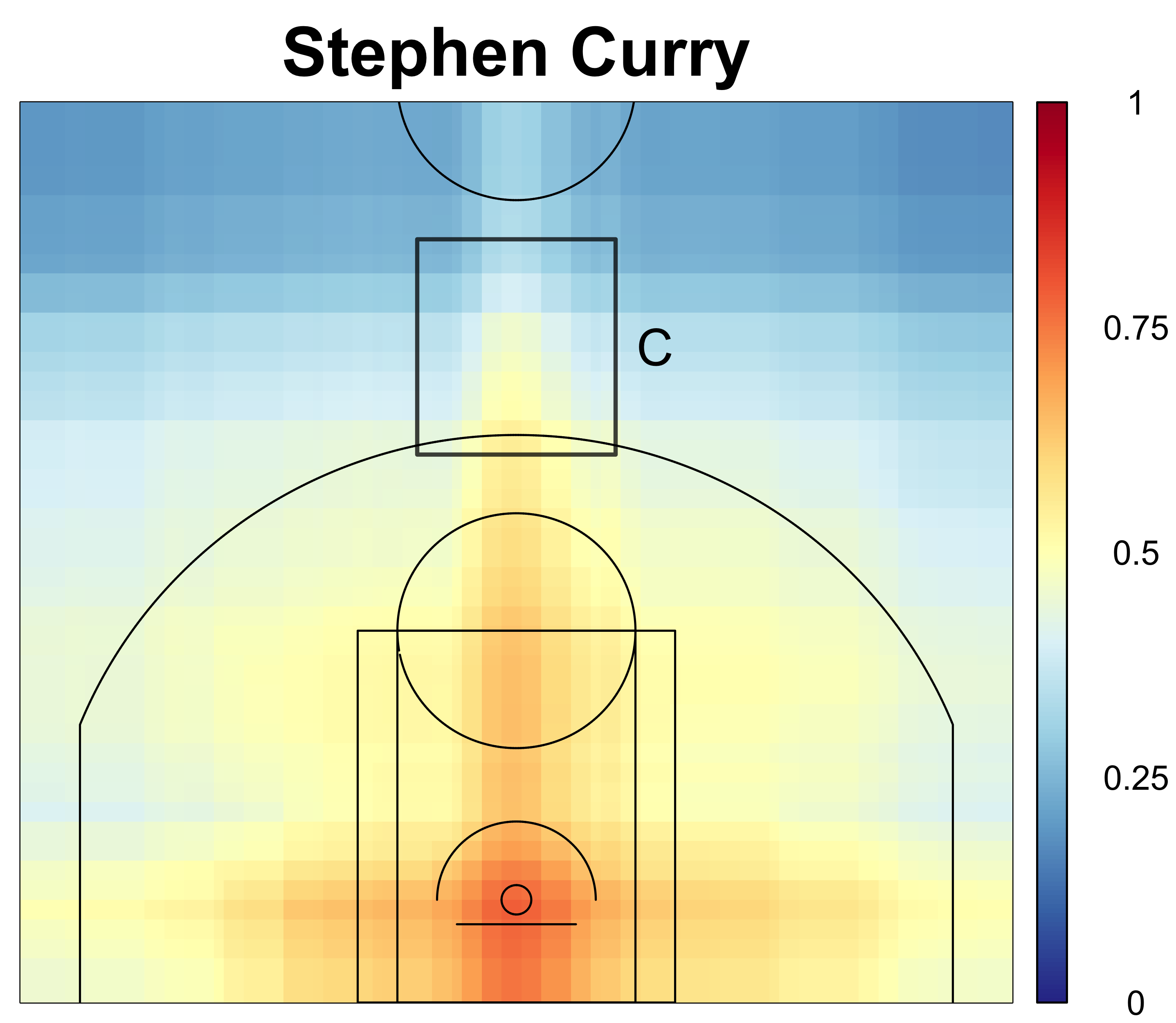}
\caption{}
\label{fig:curry_heatmap}
\end{subfigure}
\begin{subfigure}[b]{0.23\textwidth}
\centering
\includegraphics[width = \textwidth]{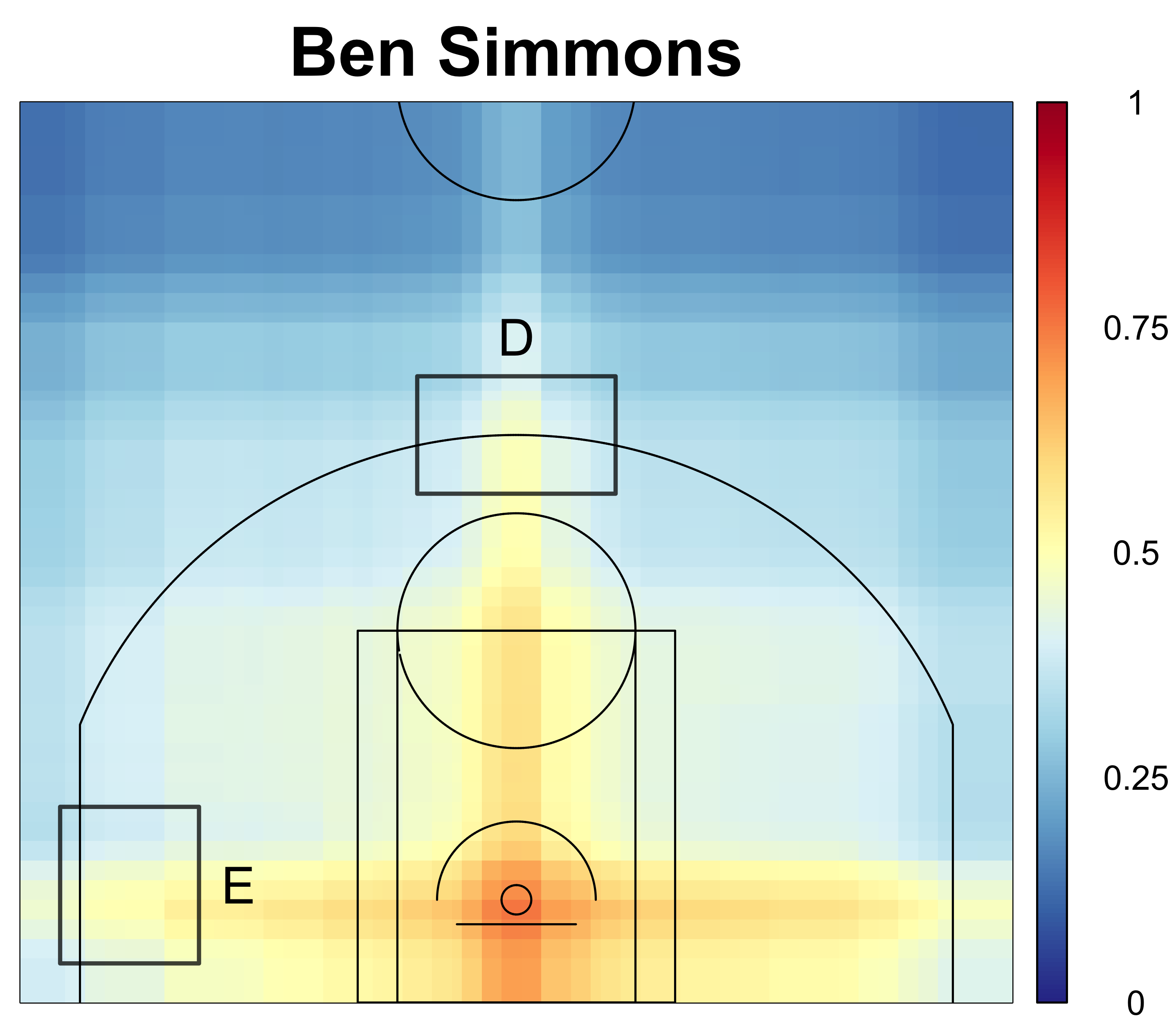}
\caption{}
\label{fig:simmons_heatmap}
\end{subfigure}
\caption{Locations of made and missed shots for Stephen Curry (a) and Ben Simmons (b). BART-estimated FG\% as a function for location for Curry (c) and Simmons (d).}
\label{fig:shot_chart}
\end{figure}%
\fi %
}

\switchref{\Cref{fig:curry_heatmap,fig:simmons_heatmap}}{\Cref{fig:intro_shots}(C) and \Cref{fig:intro_shots}(D)~} show the estimated shot charts for Curry and Simmons obtained from a BART model whose trees split on continuous features like shot location, player height, and player weight and categorical features like player name and player position. 
Since these trees are nothing more than piecewise constant step functions, these shot display contain sharp discontinuities in location. 
For instance, Curry's and Simmons' estimated FG\% jump from below 50\% (blue shading) to above 50\% (orange/red shading) and then back to below 50\% within regions C, D, and E in the figures.

Several authors have proposed BART extensions based on regression trees that output smooth functions including linear models \citep{Prado_2021}, splines \citep{Low-Kam2015, Cao2025_fbart}, and Gaussian Processes \citep[GPs;][]{Starling2020tsbart, maia-gp-bart}.
Given the ubiquity of splines and GPs in spatial statistics, the latter two types of extensions are attractive options for estimating shot charts.
Unfortunately, existing implementations either (i) only smooth in one dimension (\citet{Starling2020tsbart}'s tsBART and \citet{Low-Kam2015}'s and \citet{Cao2025_fbart}'s treed-spline models) or (ii) scale poorly and cannot be fit to a single's season of data consisting around 190,000 shots (tsBART and \citet{maia-gp-bart}'s GP-BART).

\subsection{Our Contribution}

We develop a novel extension of BART, which we call ridgeBART, that uses trees that output linear combinations of \emph{ridge functions}, which compose affine transformations of the inputs with a non-linearity.
Whereas the original BART model approximates functions with piecewise constant step functions, ridgeBART approximates functions with piecewise continuous functions. 
We introduce a new Gibbs sampler whose complexity scales linearly with the size of the data, offering substantial runtime improvements over GP-based BART extensions.
We further derive posterior contraction rates for the ridgeBART posterior in the setting where the true regression function lies in a piecewise anisotropic H\"{o}lder class.
Briefly, this class contains functions that display different degrees of smoothness with respect to each covariate in different rectangular sub-regions of $[0,1]^{p}.$
As a special case of our main result, we show that the ridgeBART posterior contracts at almost the minimax optimal rate when the true regression function belongs to the isotropic H\"{o}lder class $\calC^{\alpha}([0,1]^{p})$ for \emph{arbitrary} $\alpha > 0,$ losing only a $\log{(n)}$-factor.
An \textsf{R} implementation of ridgeBART is available at \codelink.

In \Cref{sec:review}, we review the basics of the original BART model and its GP-based extensions. 
Then, we formally introduce ridgeBART and derive our Gibbs sampling algorithm in \Cref{sec:methods} before studying its theoretical properties in \Cref{sec:theory}.
In \Cref{sec:experiments}, we assess ridgeBART's empirical performance using several synthetic datasets and use it to produce NBA shot charts.
We conclude in \Cref{sec:discussion} by discussing potential methodological and theoretical extensions.

\section{A review of BART and its GP-based extensions}
\label{sec:review}
Given $n$ observations from the model $y \sim \normaldist{f(\bx)}{\sigma^{2}}$ with $\bx \in [0,1]^{p},$ BART expresses $f$ as a sum of a large number of binary regression trees.

A \emph{regression tree} is a pair $(\calT, \calM)$ of (i) a finite, rooted binary decision tree $\calT$ consisting of $\calL(\calT)$ terminal or \emph{leaf} nodes and several non-terminal or \emph{decision} nodes; and (ii) a collection of \emph{jumps} $\calM = \{\mu_{1}, \ldots, \mu_{\calL(\calT)}\}$, one for each leaf node in $\calT.$
The decision tree partitions $[0,1]^{p}$ into several axis-parallel boxes, one for each leaf, and we can associate each $\bx \in [0,1]^{p}$ with a single leaf denoted $\ell(\bx; \calT).$
By further associating each leaf $\ell$ with a scalar jump $\mu_{\ell},$ the regression tree $(\calT, \calM)$ represents a piecewise constant step function.
We let $g(\bx; \calT, \calM) = \mu_{\ell(\bx; \calT)}$ denote the \emph{evaluation function}, which returns the jump associated to $\bx$'s leaf.

\subsection{The original BART model}
\label{sec:original_bart}

BART introduces an \emph{ensemble} $\calE = \{(\calT_{1},\calM_{1}), \ldots, (\calT_{M}, \calM_{M})\}$ of $M$ regression trees so that $f(\bx) = \sum_{m = 1}^{M}{g(\bx; \calT_{m}, \calM_{m})}.$
Fitting a BART model involves specifying a prior over $(\calE, \sigma)$ and computing summaries of the posterior distribution of $(\calE, \sigma) \vert \by$ using Markov chain Monte Carlo (MCMC). 

\textbf{The BART prior}. \citet{Chipman2010} specified a conjugate inverse gamma prior for $\sigma^{2}$ and independent and identical priors on the regression trees.
We describe that prior implicitly by explaining how to simulate a random prior sample.
First, we use a branching process to generate the graphical structure of $\calT$ (i.e., the arrangement of decision and leaf nodes).
In the branching process, the probability that a node at depth $d$ is non-terminal is $0.95(1 + d)^{-2},$ which ensures that $\calT$ is almost surely finite. 
Then, we draw decision rules for each non-terminal node conditionally given the rules at all of its ancestors in the tree, which ensures each leaf of $\calT$ corresponds to a non-empty subset of $[0,1]^{p}$ with prior probability one. 
Finally, conditionally on $\calT,$ we independently draw jumps $\mu_{\ell} \sim \normaldist{0}{\tau^{2}}.$
\citet{Chipman2010} suggested default values for $\tau$ and an inverse Gamma prior for $\sigma^{2}$ that  have proven highly effective across many applications.

\textbf{Sampling from the BART posterior}. \citet{Chipman2010} deployed a Gibbs sampler that iterated between conditionally sampling $\calE$ and $\sigma^{2}.$
Sampling $\sigma^{2} \vert \by, \calE$ involves a normal-inverse Gamma conjugate update.
For the more involved conditional update of $\calE,$ they used a Metropolis-within-Gibbs strategy that sequentially samples one regression tree at a time, keeping all others fixed, by first drawing a new decision tree marginally before conditionally sampling new jumps. 

For the decision tree update, most BART implementations use a Metropolis-Hastings step that involves randomly pruning or growing an existing decision tree.
Growing $\calT$ involves (i) selecting an existing leaf uniformly at random and connecting it to two new leaf nodes; (ii) drawing a new decision rule for the selected node; and (iii) leaving the rest of the tree unchanged.
Pruning $\calT,$ on the other hand, involves (i) deleting two randomly selected leaf nodes that share a common parent; (ii) deleting the decision rule of that parent; and (iii) leaving the rest of the tree unchanged.

The marginal decision tree density and the conditional jump distributions are available in closed-form and can be computed in $O(n)$ time.
To describe these calculations, suppose we are updating the $m$-th regression tree $(\calT_{m}, \calM_{m})$ while leaving the other $M-1$ regression trees, which we denote by $\calE^{(-m)},$ fixed. 
Let $r_{i}$ be the $i$-th \emph{partial residual} based on the fits of the other trees in $\calE^{(-m)}$ with $r_{i} = y_{i} - \sum_{m' \neq m}{g(\bx_{i}; \calT_{m'}, \calM_{m'})}.$
Additionally, for each leaf $\ell$ of an arbitrary decision tree $\calT,$ let $\calI(\ell; \calT) = \{i : \ell(\bx_{i}; T) = \ell\}$ be the set of indices for those observations associated with leaf $\ell$ and let $n_{\ell} = \lvert \calI(\ell;\calT)\rvert$ count the number of observations associated with leaf $\ell.$ 
Finally, let $\bm{r}_{\ell} = (r_{i}: i \in \calI(\ell;\calT))$ be the vector of partial residuals corresponding to observations in leaf $\ell.$
The conditional posterior distribution of $(\calT_{m}, \calM_{m})$ factorizes over the leafs of the tree:
\begin{equation}
\label{eq:vanillaBART_T_M_posterior}
p(\calT, \calM \vert \by, \calE^{(-m)}, \sigma^{2}) \propto p(\calT) \times \prod_{\ell}{\left[\tau^{-1}\exp\left\{-\frac{1}{2}\left[\sigma^{-2}\lVert \bm{r}_{\ell} - \mu_{\ell}\bm{1}_{n_{\ell}}\rVert_{2}^{2} + \tau^{-2}\mu_{\ell}^{2}\right]\right\}\right]}
\end{equation}
Marginalizing over the jumps $\mu_{\ell},$ we compute
\begin{equation}
\label{eq:vanillaBART_T_posterior}
p(\calT \vert \by, \calE^{(-m)}, \sigma^{2}) \propto p(\calT) \times \prod_{\ell}{\left[ \tau^{-1}P_{\ell}^{-\frac{1}{2}}\exp\left\{\frac{P_{\ell}^{-1}\Theta_{\ell}^{2}}{2}\right\} \right]}
\end{equation}
where $P_{\ell} = \sigma^{-2} n_{\ell} + \tau^{-2}$ and $\Theta_{\ell} = \sigma^{-2}\bm{r}^{\top}_{\ell}\bm{1}_{n_{\ell}}.$
\Cref{eq:vanillaBART_T_M_posterior} further reveals that the jumps $\mu_{\ell} \in \calM_{m}$ are conditionally independent with $\mu_{\ell} \vert \calT_{m}, \by, \calE^{(-m)}, \sigma^{2}  \sim \normaldist{P_{\ell}^{-1}\Theta_{\ell}}{P_{\ell}^{-1}}.$
Since summing the partial residuals in each leaf $\ell$ takes $O(n_{\ell})$ time, we conclude that each individual regression tree update requires $O(n)$ time.

\subsection{Updating smoother regression trees}
\label{sec:gp_bart_computation}
Although the original BART model has demonstrated great empirical success, it can yield somewhat unsatisfying results in \emph{targeted smoothing} problems as reflected in \switchref{\Cref{fig:curry_heatmap,fig:simmons_heatmap}}{\Cref{fig:intro_shots}}.
Formally, suppose we observe triplets $(\bx_{1}, \bz_{1}, y_{1}), \ldots, (\bx_{n}, \bz_{n}, y_{n})$ of covariates $\bx \in [0,1]^{p}$, \emph{smoothing variables}\footnote{Although we notationally distinguish between covariates and smoothing variables, $\bx$ and $\bz$ need not be disjoint; see \Cref{sec:gen_smooth}} $\bz \in [0,1]^{q},$ and outcomes $y \in \R$ where $y \sim \normaldist{f(\bx,\bz)}{\sigma^{2}}.$
Suppose further that for every $\bx,$ $f(\bx, \bz)$ is a smooth function of $\bz.$
Because BART uses regression trees with scalar jumps, the posterior over $f$ concentrates completely on discontinuous piecewise constant functions.

To overcome this limitation, \citet{Starling2020tsbart} and \citet{maia-gp-bart} replaced the scalar jumps $\mu_{\ell}$ with functional jumps $\mu_{\ell}(\bz),$ for which they specified independent mean-zero GP priors $\mu_{\ell}(\bz) \sim \textrm{GP}(0,k).$
They specified further hyperpriors for the kernel hyperparameters in each leaf.
To draw posterior samples, they followed the same basic strategy of \citet{Chipman2010}: update regression trees one-at-a-time by first updating the decision tree marginally and then conditionally updating the functional jumps.
In their works, the full conditional posterior density of the $m$-th regression tree is given by
\begin{equation}
\label{eq:GPbart_T_M_posterior}
p(\calT, \calM \vert \by, \calE^{(-m)}, \sigma^{2}) \propto p(\calT) \times \prod_{\ell}\left[\lvert K_{\ell}\rvert^{-\frac{1}{2}}\exp\left\{-\frac{1}{2}\left[\sigma^{-2}\lVert \br_{\ell} - \bmu_{\ell} \rVert_{2}^{2} + \bmu_{\ell}^{\top}K_{\ell}^{-1}\bmu_{\ell} \right]\right\}\right],
\end{equation}
where $\bmu_{\ell} = (\mu(\bz_{i}): i \in \calI(\ell;\calT))$ is the vector containing evaluations of the function $\mu_{\ell}$ at observations in leaf $\ell$ and $K_{\ell} = (k(\bz_{i}, \bz_{i'}))_{i,i' \in \calI(\ell, \calT)}$ is the kernel matrix evaluated at these observations. 
From \Cref{eq:GPbart_T_M_posterior}, we compute
$$
p(\calT \vert \by, \calE^{(-m)}, \sigma^{2}) \propto p(\calT) \times \prod_{\ell}\left[\lvert K_{\ell} \rvert^{-\frac{1}{2}} \lvert P_{\ell} \rvert^{-\frac{1}{2}}\exp\left\{\frac{\Theta_{\ell}^{\top}P^{-1}_{\ell}\Theta_{\ell}}{2}\right\} \right],
$$
and conclude that the vectors $\bmu_{\ell}$ are conditionally independent with $\bmu_{\ell} \sim \normaldist{P_{\ell}^{-1}\Theta_{\ell}}{P^{-1}_{\ell}}$ where $P_{\ell} = \sigma^{-2}I_{n_{\ell}} + K^{-1}_{\ell}$ and $\Theta_{\ell} = \sigma^{-2}\bm{r}_{\ell}.$

Unlike the original BART model, each regression tree update in \citet{Starling2020tsbart}'s and \citet{maia-gp-bart}'s extensions cannot be done in linear time.
Indeed, each update involves computing the Cholesky factorizations of $P_{\ell}$ and $K_{\ell}.$ 
Consequently, each regression tree update requires $O(\sum_{\ell}{n_{\ell}^{3}})$ time, which can be computationally prohibitive when $n$ is large. 

\section{Introducing ridgeBART}
\label{sec:methods}
The original BART model admits very fast regression tree updates but is limited to representing discontinuous functions.
GP-based extensions of BART, on the other hand, represent a much richer set of functions at the expense of much slower computation.
To achieve a middle ground maintaining the representational flexibility of treed-GP ensembles while admitting linear-time tree updates, we express the functional jump $\mu_{\ell}(\bz)$ as a linear combination of $D$ ridge functions
\begin{equation}
\label{eqn:ridge}
\mu_{\ell}(\bz) = \sum_{d=1}^{D}{\beta_{\ell,d} \times \varphi(\omega_{\ell, d}^{\top}\bz + b_{\ell,d})}
\end{equation}
where $\varphi(\cdot)$ is a user-specified non-linear activation function and $(\bomega_{\ell}, \bm{b}_{\ell}, \bbeta_{\ell})$ is a triplet of parameters associated with leaf $\ell.$

Formally, we introduce a new type of functional regression trees $(\calT, \calM),$ where $\calT$ is a decision tree just like in the original BART model and $\calM$ is the collection of triplets $(\bomega, \bm{b}, \bbeta)$ associated with each leaf of $\calT.$
We additionally introduce a new evaluation function
$$
g(\bx, \bz; \calT, \calM) = \sum_{d = 1}^{D}{\beta_{\ell(\bx; \calT),d} \times \varphi(\omega_{\ell(\bx; \calT),d}^{\top}\bz + b_{\ell(\bx; \calT), d}}).
$$
Given $n$ observations from the model $y \sim \normaldist{f(\bx, \bz)}{\sigma^{2}}$ with $\bx \in [0,1]^{p}$ and $\bz \in [0,1]^{q},$ ridgeBART expresses $f(\bx,\bz)$ as a sum of $M$ of these new regression trees $f(\bx, \bz) = \sum_{m = 1}^{M}{g(\bx, \bz; \calT_{m}, \calM_{m})}.$

\textbf{Motivation for \Cref{eqn:ridge}.} Before proceeding, we pause to motivate the form of our tree outputs.
Suppose that $k: [0,1]^{q} \times [0,1]^{q} \rightarrow \R$ is a real-valued stationary kernel and that $h \sim \textrm{GP}(0,k).$
Given observations $h(\bz_{1}), \ldots, h(\bz_{n}),$ sampling from the posterior distribution of $h(\bz)$ involves computing the Cholesky factorization of the $n \times n$ kernel matrix with entries $k(\bz_{i}, \bz_{i'})$ \citep[see, e.g.,][Equation (2.19) and Appendix A2]{RasmussenWilliams_gpbook}.
Exactly computing this factorization requires $O(n^{3})$ time, which becomes computationally prohibitive when $n$ is large.
\citet{Rahimi2007} proposed the approximation $k(\bz, \bz') = 2D^{-1}\sum_{d = 1}^{D}{\cos\left(\omega_{d}^{\top}\bz + b_{d}\right)\cos\left(\omega_{d}^{\top}\bz' + b_{d}\right)},$ where $b_{d} \sim \unifdist{0}{2\pi}$ and the $\omega_{d}$'s are drawn i.i.d.~from $k$'s \emph{spectral measure} $\calP_{\omega},$ which satisfies $k(\bz, \bz') = \E_{\omega \sim \calP_{\omega}}[e^{i\omega^{\top}(\bz-\bz')}]$\footnote{The existence of $\calP_{\omega}$ is guaranteed by Bochner's theorem.}.
Essentially, their result allows us to approximate $h$ with a linear combination of $D$ \emph{randomly} constructed cosine features.
Computing the posterior of the weights, which induces an approximate posterior over $h(\bz),$ requires only $O(nD^{2})$ time, a considerable speedup over the $O(n^{3})$ time needed for exact GP computation.

\citet{Li2023} built regression trees ensembles with random Fourier feature expansions in each leaf.
While similar to that work, our proposed model, theoretical results, and software implementation are considerably more general, allowing for a wide range of popular activation functions $\varphi(\cdot)$ like the hyperbolic tangent $\tanh(\cdot),$ the sigmoid $[1 + e^{-\cdot}]^{-1},$ and the rectified linear unit (ReLU) $\max\{0, \cdot\}.$
By associating each tree leaf with an expansion like \Cref{eqn:ridge}, we are effectively building an ensemble of treed-neural networks.

\subsection{The ridgeBART Prior}
\label{sec:ridgeBART_prior}

We specify a conjugate inverse gamma prior for $\sigma^{2}$ with the same default hyperparameters as \citet{Chipman2010}.
We specify the regression tree prior compositionally, by describing how to simulate a prior draw of a single regression tree.
First, we draw the decision tree $\calT$ from the branching process prior described in \Cref{sec:review}.
Then, independently across leafs $\ell$ of $\calT,$ we draw each $\beta_{d,\ell} \sim \normaldist{0}{\tau^{2}}$ and each offset $b_{d,\ell} \sim \unifdist{0}{2\pi}$ if $\varphi$ is the cosine function and $b_{d,\ell} \sim \normaldist{0}{1}$ otherwise.
Finally, we draw a leaf-specific scale parameter $\rho_{\ell} \sim \gammadist{\nu/2}{\nu\lambda/2}$ and then independently draw each $\omega_{d,\ell} \sim \normaldist{0}{\rho_{\ell}^{-1}}.$

\textbf{Hyperparameter Defaults.}
In addition to $M, D,$ and $\varphi,$ the ridgeBART priors depend on $\tau^{2},$ the prior variance of each $\beta_{d}$; and $\nu$ and $\lambda,$ the hyperparameters for the prior on $\rho_{\ell}.$ 
We recommend setting $M = 50, D = 1$ and $\tau = (y_\text{max} - y_\text{min}) / (4 \times \sqrt{MD}).$
We further recommend fixing $\nu = 3$ and setting $\lambda \approx 0.788$ so that $\P(\rho_{\ell} < 1) \approx 50\%.$
Although these settings are somewhat \textit{ad hoc}, we have found that they work well in practice; see \suppref{\Cref{sec:hyperparameters}}{Appendix A of the Supplementary Materials} for a sensitivity analysis.

\textbf{Optional settings for spatial domains.}
In our numerical studies, we noticed that when $\bz$ contained spatial information (e.g., location on a basketball court, latitude/longitude, etc.) and when the prior covariance matrix $V$ for the $\omega_{\ell}$'s was diagonal, the ridgeBART predictions displayed unwanted axis-aligned artifacts. 
We found that drawing $\omega_{\ell}$'s from a \emph{randomly} rotated multivariate Gaussian avoided such behavior. 
Specifically, when $\bz$ contains spatial information, we recommend drawing $\omega_{\ell} \sim \normaldist{0}{QVQ^{\top}},$ where $Q$ is computed from the QR-decomposition of a randomly sampled matrix with independent standard normal entries.

\subsection{Sampling from the ridgeBART posterior}
\label{sec:ridgeBART_computation}

To sample from the ridgeBART posterior, it is tempting to follow \citet{Chipman2010}, \citet{Starling2020tsbart}, and \citet{maia-gp-bart}: update the regression trees one-at-a-time by first updating the decision tree marginally with MH and then conditionally updating the jumps.
In ridgeBART, the functional jump $\mu_{\ell}$ in leaf $\ell$ is parameterized by the scalar $\rho$ and three vectors, $\bomega_{\ell} = (\omega_{\ell,1}, \ldots, \omega_{\ell, D})^{\top},$ $\bm{b}_{\ell} = (b_{\ell,1}, \ldots, b_{\ell,D})^{\top}$ and $\bbeta_{\ell} = (\beta_{\ell,1}, \ldots, \beta_{\ell,D})^{\top}.$
The full conditional posterior density of $m$-th regression tree $(\calT_{m}, \calM_{m})$ is 
\begin{equation}
\label{eq:ridgeBART_T_M_posterior}
p(\calT, \calM \vert \ldots) = p(\calT) \times \prod_{\ell}{\left[p(\rho_\ell, \bomega_{\ell}, \bm{b}_{\ell})\tau^{-D} \exp\left\{\sigma^{-2}\lVert \br_{\ell} - \Phi_{\ell}\bbeta_{\ell} \rVert_{2}^{2} + \tau^{-2}\bbeta_{\ell}^{\top}\bbeta_{\ell}\right\}\right]},
\end{equation}
where $\Phi_{\ell}$ is a $n_{\ell} \times D$ matrix whose $d$-th column contains evaluations of the function $\varphi(\omega_{\ell,d}^{\top}\bz_{i} + b_{\ell,d})$ for the observations $i \in \calI(\ell;\calT).$

Unfortunately, it is not generally possible to obtain $p(\calT \vert \by, \calE^{(-m)}, \sigma^{2})$ in closed form, rendering the marginal decision tree and conditional jump updates impractical.
Notice, however, that (i) we can easily marginalize out the \emph{outer weights} $\bbeta_{\ell}$'s from \Cref{eq:ridgeBART_T_M_posterior} and (ii) that the $\bbeta_{\ell}$'s are conditionally independent and follow multivariate normal distributions given $\calT$ and all the \emph{inner weights} $\Xi = \{(\rho_\ell, \bomega_{\ell}, \bm{b}_{\ell})\}.$
Letting $P_{\ell} = \sigma^{-2}\Phi^{\top}_{\ell}\Phi_{\ell} + \tau^{-2}I_{n_\ell}$ and $\Theta_{\ell} = \sigma^{-2}\Phi_{\ell}^{\top}\br_{\ell},$ we compute
\begin{align}
p(\calT, \Xi \vert \by, \calE^{(-m)}, \sigma^{2}) &\propto p(\calT) \times \prod_{\ell}{\left[p(\rho_\ell, \bomega_{\ell}, \bm{b}_{\ell})\tau^{-D}\lvert P_{\ell} \rvert^{-\frac{1}{2}}\exp\left\{\frac{\Theta_{\ell}^{\top}P_{\ell}^{-1}\Theta_{\ell}}{2}\right\} \right]} \label{eq:ridgeBART_T_Xi_posterior} \\
\bbeta_{\ell} \vert \calT, \Xi, \by, \calE^{(-m)}, \sigma^{2} &\sim \normaldist{P_{\ell}^{-1}\Theta_{\ell}}{P_{\ell}^{-1}}. \label{eq:ridgeBART_beta_posterior}
\end{align}

\Cref{eq:ridgeBART_T_Xi_posterior,eq:ridgeBART_beta_posterior} immediately suggests a natural sampling strategy. 
First, we marginally update the pair $(\calT, \Xi)$ consisting of the decision tree $\calT$ and all leaf inner weights $\Xi$.
Then, conditionally on $(\calT, \Xi),$ we update the leaf outer weights. 

\textbf{Sampling $(\calT, \Xi)$}. Because exactly sampling from the distribution in \Cref{eq:ridgeBART_T_Xi_posterior} is infeasible, we update $(\calT, \Xi)$ using a Metropolis-Hastings step with a proposal mechanism that randomly grows the tree, prunes the tree, or changes the inner weights at all leafs of a tree.
A grow move in ridgeBART involves (i) selecting an existing leaf uniformly at random and connecting it to two new leaf nodes; (ii) drawing a new decision rule for the selected node; (iii) drawing inner weights for the newly created leaf nodes; and (iv) leaving the rest of tree and remaining inner weights unchanged.
Similarly, a ridgeBART prune move involves (i) deleting two randomly selected leaf nodes that share a common parent; (ii) deleting the decision rule of that parent and the inner weights associated with the now-deleted leaves; (iv) drawing inner weights for the newly created leaf node; and (v) leaving the rest of the tree and remaining inner weights unchanged.
A ridgeBART change move\footnote{Our change proposals differ from \citet{Chipman1998}'s, which modify decision rules} leaves the tree structure unchanged but draws new inner weights at each leaf.

In grow, prune, and change proposals, we draw new inner weights from the prior.
While it is tempting to propose $\Xi$ in a more data-informed fashion (e.g., by searching for an optimal set of inner weights), doing so yields proposals with exceedingly small acceptance probabilities, which in turn prevents the Markov chain from moving \citep[see][Appendix B2 for further details]{Deshpande2024_flexBART}.
The local nature of grow, prune, and change moves and the fact that we propose new inner weights from the prior facilitates considerable cancellation in the Metropolis-Hastings ratio.

\textbf{Sampling $\sigma^{2}.$} After updating $\calE,$ we draw a new $\sigma^{2}$ from a conditionally conjugate inverse gamma distribution.

\textbf{Computational complexity}. In each regression tree update, we must (i) compute and invert the matrix $P_{\ell},$ which takes $O(n_{\ell}D^{2} + D^{3})$ time; and (ii) evaluate the matrix-vector products $\Phi^{\top}\bm{r}_{\ell}$ and $P_{\ell}^{-1}\Theta_{\ell},$ which requires $O(n_{\ell}D + D^{2})$ time.
Thus, updating a single regression tree with $L$ leaves takes $O(n(D^{2}+D) + L(D^{3}+D^{2}))$ time.
In practice, we recommend setting $D = 1$ and typically observe $L \ll n.$
Consequently, the typical ridgeBART tree update's complexity is effectively $O(n),$ which is exactly the same as the original BART. 

By updating the inner weights $\Xi$ alongside the decision tree structure, we can fit much more expressive piecewise smooth models using a sampler that is nearly identical to that of the original BART model, which could only represent piecewise constant functions. 
In other words, we have vastly expanded the representational flexibility of BART without introducing substantial computational burden.

\section{Theoretical Results}
\label{sec:theory}
In addition to its excellent empirical performance, the original BART model exhibits desirable theoretical properties: under a slight modification to the branching process prior for $\calT$ and other mild assumptions, \citet{RockovaSaha2019} showed that the BART posterior contracts at a nearly-minimax optimal rate when the regression function $f(\bx)$ belongs to the H\"{o}lder class $\calC^{\alpha}([0,1]^{p}),$ for $0 < \alpha < 1$ (i.e., $f$ is rougher than Lipschitz). 
\citet{Jeong2023} extended this result to the setting where $f(\bx)$ displays different levels of smoothness with respect to each covariate across different rectangular sub-regions of $[0,1]^{p}$ (i.e., when $f$ belongs to a \textit{piecewise anisotropic} H\"{o}lder class; see \Cref{def:piecewise_anisotropic_holder}).
Their more general result is nevertheless limited by the assumption that $f$ is rougher than Lipschitz in each sub-region.
In this section, we consider the setting where $\bx = \bz$ and show that for sufficiently smooth activation functions $\varphi,$ the ridgeBART posterior concentrates when $f(\bx)$ lies in a piecewise anisotropic H\"{o}lder smoothness class of \emph{arbitrary} order.

\textbf{Setting and notation}. Before formally defining our function space, we introduce some notation. 
Suppose that $\boldsymbol{\Psi} = [0,1]^{p}$ has been partitioned in $R$ non-overlapping rectangular boxes $\boldsymbol{\Psi} = \bigcup_{r = 1}^{R}{\Psi_{r}}.$
We will allow the true regression function to display different degrees of smoothness within each region $\Psi_{r}.$
We denote the space of all continuous functions on $[0,1]^{p}$ by $C([0,1]^{p})$ and for any $h \in C([0,1]^{p}),$ $\bx \in [0,1]^{p}$ and $1 \leq i \leq p,$ let $h_{i}(\cdot \vert \bx)$ be the univariate slice function $y \mapsto h(x_{1}, \ldots, x_{i-1}, y, x_{i+1}, \ldots, x_{p}).$
Additionally, for $\bm{n}=(n_1,\dots,n_p)$, let $D^{\bm{n}} h$ denote the mixed partial derivatives of order $(n_1,\dots,n_p)$ of $h$. 
Given covariate vectors $\bx_{1}, \ldots, \bx_{n},$ we denote the empirical $L^2 $ norm by $\lVert h \rVert_n^2 = n^{-1}\sum_{i=1}^{n} h(\bm{x}_i)^2 $.
Finally, given a vector $\balpha = (\alpha_{1}, \ldots, \alpha_{p})^{\top} \in \R^{p}_{+}$ let $\hmean{\balpha} =  p\left(\sum_{j=1}^{p} \alpha_j^{-1}\right)^{-1}$ be the harmonic mean of its entries.

\begin{definition}[Anisotropic H\"older space]
\label{def:anisotropic_holder}
Let $\balpha = (\alpha_{1}, \ldots, \alpha_{p})^{\top} \in \R^{p}_{+}$ and let $\lambda > 0.$
The \emph{anisotropic H\"{o}lder space} with exponent $\balpha$ and coefficient $\lambda,$ $\calH^{\balpha,p}_{\lambda}(\Psi)$ consists of all functions $f: \Psi \to \R$ such that 
$$\max_{1 \leq i \leq p} \sup_{\bm{x} \in [0,1]^p} \sum_{j=0}^{\lfloor \alpha_i \rfloor} \| D^j f_i(\cdot | \bm{x}) \|_{\infty} \leq \lambda$$ 
and for all $y \in [0,1]$ and $h \in [0,1-y]$,
$$ \sup_{\,\bm{x} \in \Psi\,}\lVert D^{\lfloor \alpha_{j} \rfloor} f_{j}(y+h\vert \bx) - D^{\lfloor \alpha_{j} \rfloor}f_{j}(y \vert \bx) \rVert_{\infty} \leq \lambda \lvert h \rvert^{\alpha_{j} - \lfloor \alpha_{j} \rfloor}.
$$
\end{definition}

\citet{BhattaAOS} showed that posteriors based on multi-bandwidth Gaussian Processes contract at near minimax-optimal rates when the true function lies in the anisotropic H\"{o}lder-class $\calH^{\balpha,p}_{\lambda}(\Psi).$
We will instead consider the class of functions that, when restricted to each subregion $\Psi_{r},$ belong to an anisotropic H\"{o}lder class over each $\Psi_{r},$ but display different degrees of smoothness across each subregion.
\Cref{def:piecewise_anisotropic_holder} formally introduces our \emph{piecewise} anisotropic H\"{o}lder class and is a natural extension of \citet{Jeong2023}'s Definition 2.
Although we allow $f \in \calH^{\mathcal{A},p}_{\lambda}$ to exhibit different degrees of smoothness in each subregion $\Psi_{r},$ we will assume that $f$ is continuous across the boundaries of adjacent subregions. 

\begin{definition}[Piecewise anisotropic H\"older]
\label{def:piecewise_anisotropic_holder}
Let $\lambda > 0$ and let $\balpha_{1}, \ldots, \balpha_{R} \in \R^{p}_{+}$ with the $\hmean{\balpha_{1}} = \cdots \hmean{\balpha_{R}}.$
Denoting $\mathcal{A} = \{\balpha_{1}, \ldots, \balpha_{R}\},$ the piecewise anisotropic H\"{o}lder class $\calH_{\lambda}^{\mathcal{A}, p}$ consists of all continuous functions $f: \boldsymbol{\Psi} \to \R$ such that $f\vert_{\Psi_{r}} \in \calH^{\balpha_{r},p}_{\lambda}(\Psi_{r})$ for each $r = 1, \ldots, R.$
\end{definition}

The first key step in deriving ridgeBART's posterior contraction rate over $\calH^{\calA,p}_{\lambda}(\boldsymbol{\Psi})$ is to show that every function in the class is well-approximated by a linear combination of ridge functions.

\begin{lemma}
\label{lem:globalridgeapprox}
Suppose $f_{0} \in \calH^{\calA,p}_{\lambda}(\Psi).$
Then there exists a ridge function approximation of the form $\hat{f}(\bx)  = \sum_{d=1}^{D^{\star}}{\beta_{d}\varphi(\bomega_d ^{\top}\bx + b_{d})}$ such that $\lVert f_{0} - \hat{f}\rVert_{\infty} \leq C D^{\star \ -\min_{r,j}{\alpha_{r,j}}/p},$ for some constant $C.$
\end{lemma}

The proof of \Cref{lem:globalridgeapprox} is in \suppref{\Cref{sec:globalapprox}}{Section C1 of the Supplementary Materials} and follows the proof techniques in \citet{Leshno1993}, \citet{Barron93}, and \citet[Chapters 7 \& 8]{Devore-Lorentz}. 
Briefly, we begin by approximating the target function $f_0$ with polynomials in each box $\Psi_{r}.$
Then, using \suppref{\Cref{lem:localridge}}{Lemma C1 in the Supplementary Materials} and standard H\"{o}lder bounds, we show that each local polynomial approximates each $f_0\vert \Psi_{r}$ to within $D^{-\min_{j}{\alpha_{r,j}}/p}$ on each $\Psi_{r},$ where $D$ is the number of ridge units per leaf, as described in \Cref{eqn:ridge}. The total count of ridge expansions across $R$ subdomains is given by $D^{\star} \approx RD$. In simpler terms, our per-subdomain budget of ridge functions is $D$ and the total budget is $D^{\star}$.
Using ideas from \citet{Leshno1993} and \citet{Pinkus1997ApproximatingBR} we can approximate each polynomial path with a sum of ridge functions, while preserving the same approximation error order.
Finally, we can ``glue together'' the local approximations into a single approximation by using a smooth partition of unity. 

The next key result (\suppref{\Cref{lem:approximation}}{Lemma C2 in the Supplementary Materials}) shows that for each $f \in \calH^{\calA,p}_{\lambda}(\boldsymbol{\Psi}),$ there is a sufficiently deep \emph{anisotropic k-d} tree \citep[Definition 9]{Jeong2023} outputting linear combinations of ridge functions that approximates $f$ to within some error rate $\varepsilon_{n}.$
This error rate, which depends on $n, p, \lambda,$ and the harmonic mean of the vectors in $\calA,$ turns out to be the posterior concentration rate.

To prove our main result, beyond assuming that the data arose from the model $y \sim \normaldist{f_{0}(\bx)}{\sigma_{0}^{2}}$ with $f_{0} \in \calH^{\calA, p}_{\lambda}(\boldsymbol{\Psi}),$ we make the following assumptions:
\begin{itemize}
\item[(A1)]{$\|f_0\|_\infty \lesssim \sqrt{\log n}$ and $p = o(\log n).$}
\item[(A2)]{The H\"older coefficient $\lambda$ has at most polynomial growth in $n$, i.e., $\lambda \lesssim n^{a_2}$ for some $a_2 > 0$.}
\item[(A3)] The number of ridge functions in $\hat{f}$, $D^{\star} \asymp R \cdot\epsilon_n ^{-p/\bar{\alpha}}$.
\end{itemize}

Assumption (A1) mildly restricts the supremum norm of $f_0$, which is standard in high-dimensional Bayesian regression analyses, to guarantee sufficient prior concentration under the normal prior \citep[c.f.\ Assumption A3]{Jeong2023}. 
Similar boundedness assumptions appear in \citet{RockovaSaha2019} and \citet{Jeong2023}.  
Assumption (A2) with the H\"older coefficient $\lambda \lesssim n^{a_2}$ is a standard mild growth condition: the function can get slightly steeper with $n$, but not too rapidly. 
Assumption (A3) is the standard relationship balancing approximation and estimation errors in a sum‐of‐ridge or a neural network–like model.

We additionally make a few modifications to the ridgeBART prior:
\begin{itemize}
\item[(P1)]{In the branching process prior for $\calT,$ the probability that a node at depth $d$ is non-terminal is $\gamma^{d}$, with $1/n < \gamma < 1/2.$
}
\item[(P2)]{The outer weight prior is $\bbeta \sim \normaldist{0}{\Sigma_{\bbeta}}$, and there exists a constant $C_{\beta}$ that does not depend on $n$ or $p$ such that $C_{\beta}^{-1}I \prec \Sigma_{\bbeta} \prec C_{\beta}I.$}
\item[(P3)]{The variance parameter $\sigma^2$ is assigned an inverse gamma prior.}
\item[(P4)]{The activation function $\varphi \in \calC^{k}(\mathbb{R})$ for some integer $k \ge 1;$ is not a polynomial; and its Lipschitz constant on any compact set is bounded by $L_{\varphi}$.}
\end{itemize}
Assumption (P1) is from \citet{RockovaSaha2019}, which is a standard way to control the distribution of tree shapes and sizes in the BART prior.
Assumption (P2) controls the outer weights prior scale, preventing unbounded or negligible weights while Assumption (P3) is standard in Bayesian regression, ensuring good noise scale adaptivity. 
Assumption (P4) is central to the model's expressive power. 
By requiring $\varphi$ to be sufficiently smooth, we ensure the ridge function building blocks can themselves represent smooth functions. 
The requirement that it is non-polynomial is also critical (stated in \citet[Theorem 1]{Leshno1993}), as a polynomial activation function would limit the model to only generating other polynomials, thereby losing the universal approximation capability needed to represent arbitrary functions. 
This is the key mechanism that allows ridgeBART to approximate target functions of arbitrary smoothness $(\alpha>0),$ overcoming the limitations of standard BART.

Under these and five additional technical assumptions ((A4) -- (A8) in \suppref{\Cref{sec:proofanisotropic}}{Appendix C2 of the Supplementary Materials)}, the ridgeBART posterior places essentially all its mass in neighborhoods of $f_0$ and $\sigma_0 ^2$ whose empirical $L^2$ radius is $O(\epsilon_n)$. 
These additional assumptions ensure that there is regular-enough anisotropic k-d tree that can capture the behavior of $f_0$ in each sub-domain $\Psi_r$. 

\begin{theorem}
\label{thm:anisotropic_concentration}
    Suppose that $f_{0} \in \calH^{\calA,p}_{\lambda}(\Psi)$ and that $\hmean{\balpha_{1}} = \cdots = \hmean{\balpha_{R}} = \bar{\alpha}$ for some $\bar{\alpha} > 0.$ Under assumptions (A1)--(A8) and (P1)--(P4) there exists a constant $M > 0$ such that for $\epsilon_n = (\lambda p)^{\frac{p}{(2 \bar{\alpha}+p)}}(\nicefrac{(R\log n)}{n})^{\nicefrac{\bar{\alpha}}{(2 \bar{\alpha}+p)}}$:
    $$\mathbb{E}_0 \Pi \left\{(f, \sigma^2): \lVert f - f_0 \rVert_n + \lvert \sigma^2 - \sigma_0^2 \rvert  > M\epsilon_n \vert \by \right\} \longrightarrow 0.$$
\end{theorem}
The proof \Cref{thm:anisotropic_concentration} is in \suppref{\Cref{sec:proofanisotropic}}{Appendix C2 of the Supplementary Materials} and follows the strategy in \citet{Jeong2023} by verifying a Kullback-Leibler (KL), prior mass, and metric entropy conditions of \citet[Theorem 4]{Ghosal_2007}. 
We note that the contraction rate $\epsilon_{n}$ can, in principle, be derived without the assuming $\hmean{\balpha_1} = \cdots = \hmean{\balpha_R}.$
In that more general scenario, the global posterior contraction rate is dictated by the function restricted to the sub-region over which it is the least smooth.

When $R=1$ and all $\alpha_{1} = \cdots = \alpha_{p} = \alpha > 0$, the space $\mathcal{H}^{\calA, p}_{\lambda}(\boldsymbol{\Psi})$ reduces to the usual isotropic H\"older class $\calC^{\alpha}(\boldsymbol{\Psi}).$  
Initial theoretical analyses of the BART posterior \citep[e.g.,][]{RockovaSaha2019, Rockova2019} showed that Bayesian ensembles of regression trees with constant leaf outputs could recover functions in $\calC^{\alpha}(\boldsymbol{\Psi})$ at a near-optimal rate only for $0 < \alpha \leq 1.$
They further noted that the use of piecewise constant weak learners in the ensemble precludes near-optimal recovery of smoother functions.
By using ensembles of soft decision trees, which output smooth combinations of all leaf outputs, \citet{LineroYang2018}'s SoftBART can recover $\calC^{\alpha}(\boldsymbol{\Psi})$ functions for arbitrary $\alpha \geq 0$ 
Our result, when specialized to the isotropic setting, matches their result: by replacing constant outputs of hard decision trees with linear combinations of sufficiently many sufficiently smooth ridge functions, ridgeBART can recover $\calC^{\alpha}(\boldsymbol{\Psi})$ functions at a near-optimal rate for any $\alpha \geq 0.$

\section{Experiments}
\label{sec:experiments}
We performed two synthetic data simulations to assess ridgeBART's ability to perform targeted smoothing (\Cref{sec:ts}) and to recover more general functions (\Cref{sec:gen_smooth}). 
We then applied ridgeBART to the NBA shot chart data (\Cref{sec:shot_charts}). 
We ran ridgeBART with $M = 50$ trees and $D = 1$ ridge functions in each leaf and report the results from using cosine, ReLU, and $\tanh$ activation functions.
We compared ridgeBART to the original BART model, Soft BART, GP-BART, and tsBART. 
We specifically used the implementations available in the \textbf{BART} \citep{BARTpkg}, \textbf{SoftBart} \citep{SoftBartpkg}, \textbf{gpbart} \citep{gpbartpkg}, and \textbf{tsbart} \citep{tsbartpkg} \textsf{R} packages.
For each method, we drew 10,000 posterior samples by sequentially running 10 Markov chains for 2,000 iterations discarding the first 1000 iterations of each chain as burn-in.
All experiments were run on a shared high-throughput cluster using \textsf{R} version 4.4.1.

\subsection{Targeted smoothing}
\label{sec:ts}

We assessed ridgeBART's ability to perform targeted smoothing using a data generating process that mimics patient-specific recovery curves.
Briefly, following a major injury or surgery, patients experience a sharp initial drop from a pre-event baseline value followed by a smooth increase to an asymptotic value below the baseline.
\citet{Wang2019} parametrized recovery curves as $f(\bx, z) = (1 - A(\bx))(1 - B(\bx) e^{-z \times C(\bx)}),$
where $z > 0$ is the time since the event (in months) and the functions $A(\bx), B(\bx),$ and $C(\bx)$ respectively capture the asymptote relative to the baseline, the initial drop relative to the final asymptote, and the recovery rate as a function of patient covariates.
In our first experiment, we set $A(\bx) = 0.5\lvert (1 + e^{-2 x_1})^{-1} - 0.5 \rvert - 0.5;$ $B(\bx) = 1+\min\{0, 0.15\cos{(5x_{2})}\},$ and $C(\bx) = 5e^{x_{3}}.$

For each patient $i = 1, \ldots, n,$ we drew covariates $\bx_{i} \in [0,1]^{6}$ uniformly; sampled $n_{i} \sim 1 + \poisdist{3}$ observation times $z_{it} \in [0, 24];$ and generated observations $y_{it} \sim \normaldist{f(\bx_{i}, z_{it})}{0.05^{2}}.$
To mimic real clinical data, we clustered the observation times around follow-up times of 1, 2, 4, 6, 8, 12, 18, and 24 months post-event; see \suppref{\Cref{sec:recovery_curve_details}}{Section B1 of the Supplementary Materials} for specifics.

\Cref{fig:rc_examples} shows the recovery curves for three different patients obtained after fitting ridgeBART with different choices of $\varphi,$ BART, and SoftBART to a simulated dataset containing 2419 observations from $n = 600$ patients.
Averaging over 10 data replications, it took ridgeBART 342 seconds (cosine activation); 306 seconds (tanh activation); and 93 seconds (ReLU activation) to simulate a single Markov chain for 2000 iterations.
By comparison, it took SoftBART 1087 seconds and BART 376 seconds on average.
We were unable to fit tsBART and GP-BART models to these data using the package defaults within our compute cluster's 72-hour limit.

Each panel of \Cref{fig:rc_examples} shows the patient's observed data (points), true curve $f(\bx, z)$ (solid black line), the estimated posterior mean (red line) and pointwise 95\% credible intervals (shaded) of $f(\bx, z).$
Every method produced excellent estimates of each individual curve: the posterior means closely follow the true curves' shapes and the pointwise 95\% credible intervals display higher-than-nominal coverage.
But on further inspection, we found that BART and SoftBART returned somewhat jagged posterior mean estimates and credible intervals whose widths oscillated from very narrow to very wide.
In sharp contrast, ridgeBART returned much smoother posterior means and uncertainty intervals with more consistent width.
Visually, the three ridgeBART estimates also appeared much more accurate than the BART and SoftBART estimates.

\begin{figure}[ht]
\centering
\includegraphics[width = .95\textwidth]{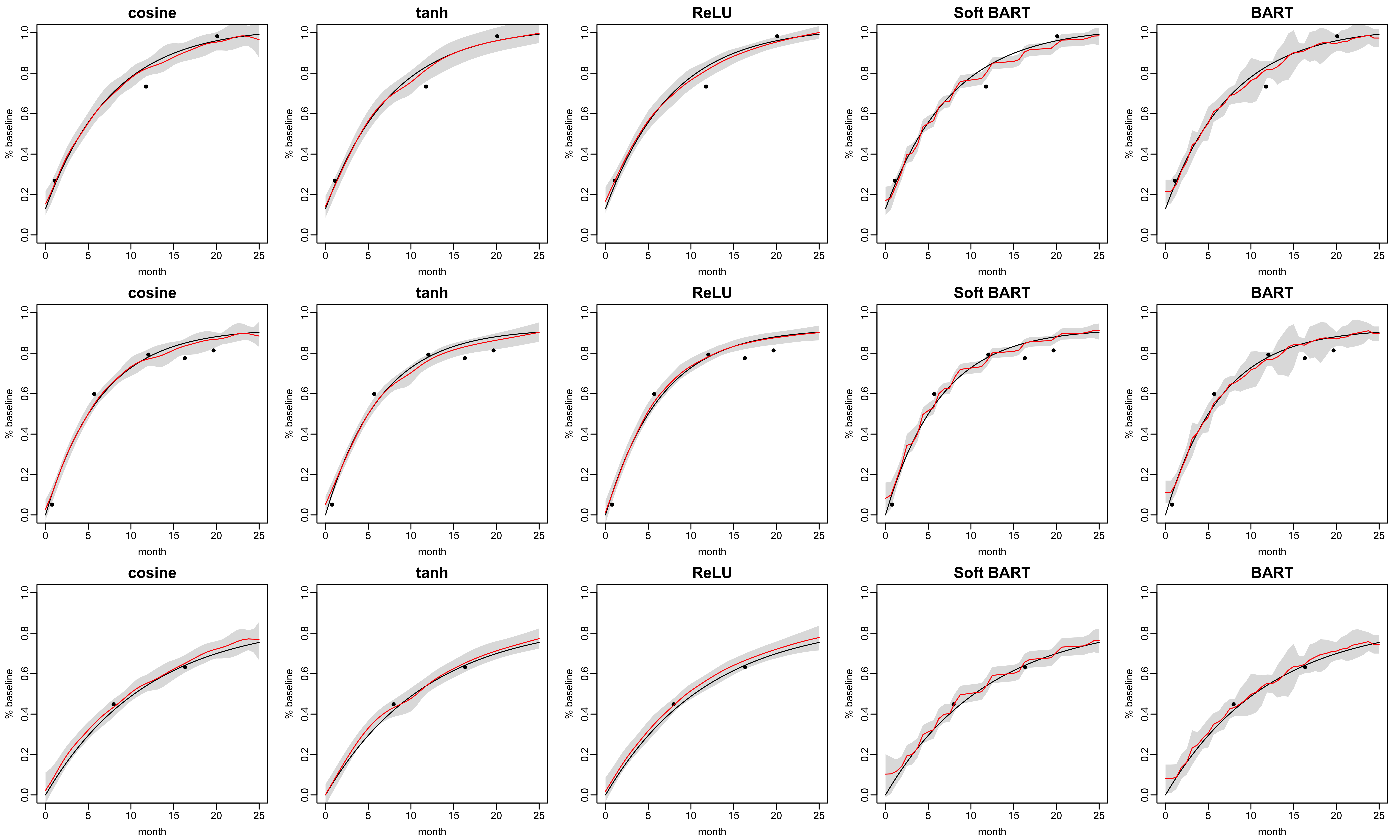}
\caption{Examples of three patient function recoveries for each model when $n = 600$.}
    \label{fig:rc_examples}
\end{figure}

To quantify the improvement offered by ridgeBART, we generated 10 synthetic datasets for each value of $n \in \{200, 400, 600, 800, 1000\}$ and computed two average out-of-sample root mean squared errors for each method.
The first error measures each method's ability to evaluate the recovery curve at a new observation time for individuals in the training data.
The second error measures each method's ability to evaluate the recovery curve at a new observation for individuals not in the training data.
To compute the first error, we used each method to predict the value of $f(\bx, z)$ along a uniformly-spaced grid of time points $z \in [0,24].$
The second error was computed by generating 1000 new covariate vectors that were not seen during training and evaluating their recovery function at a single randomly chosen point in $[0, 24]$.
Both errors are informative: the first tracks our ability to predict how well an already-observed patient will continue to recover while the second quantifies how well we can predict an entirely new patient's recovery.
{
\if\jrssb1%
\begin{figure}[ht]
\centering
\includegraphics[width = \textwidth]{figures/fig3}
\caption{Prediction error for recovery curve simulation study on (A) in-sample patients and (B) out-of-sample-patients and (C) Friedman function simulation study on out-of-sample folds. Panels (A) and (B) depict mean RMSE relative to BART, while panel (C) depicts actual RMSE.}%
\label{fig:experiments}%
\end{figure}%
\else%

\begin{figure}[ht]
    \centering
    \begin{subfigure}{0.32\textwidth}
    \centering
    \includegraphics[width = \textwidth]{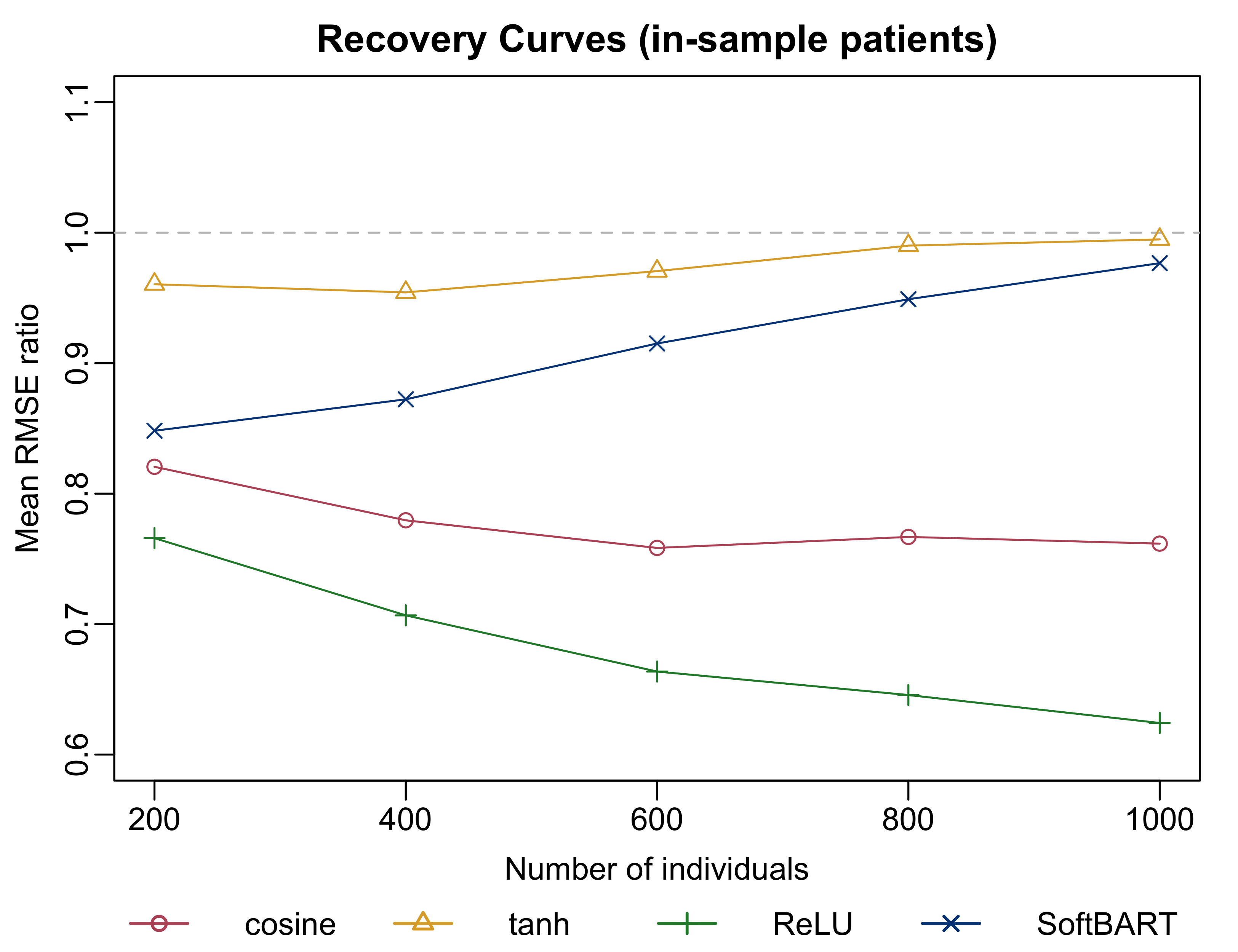}
    \caption{}
    \label{fig:rc_in_rmse}
    \end{subfigure}
    \begin{subfigure}{0.32\textwidth}
    \centering
    \includegraphics[width = \textwidth]{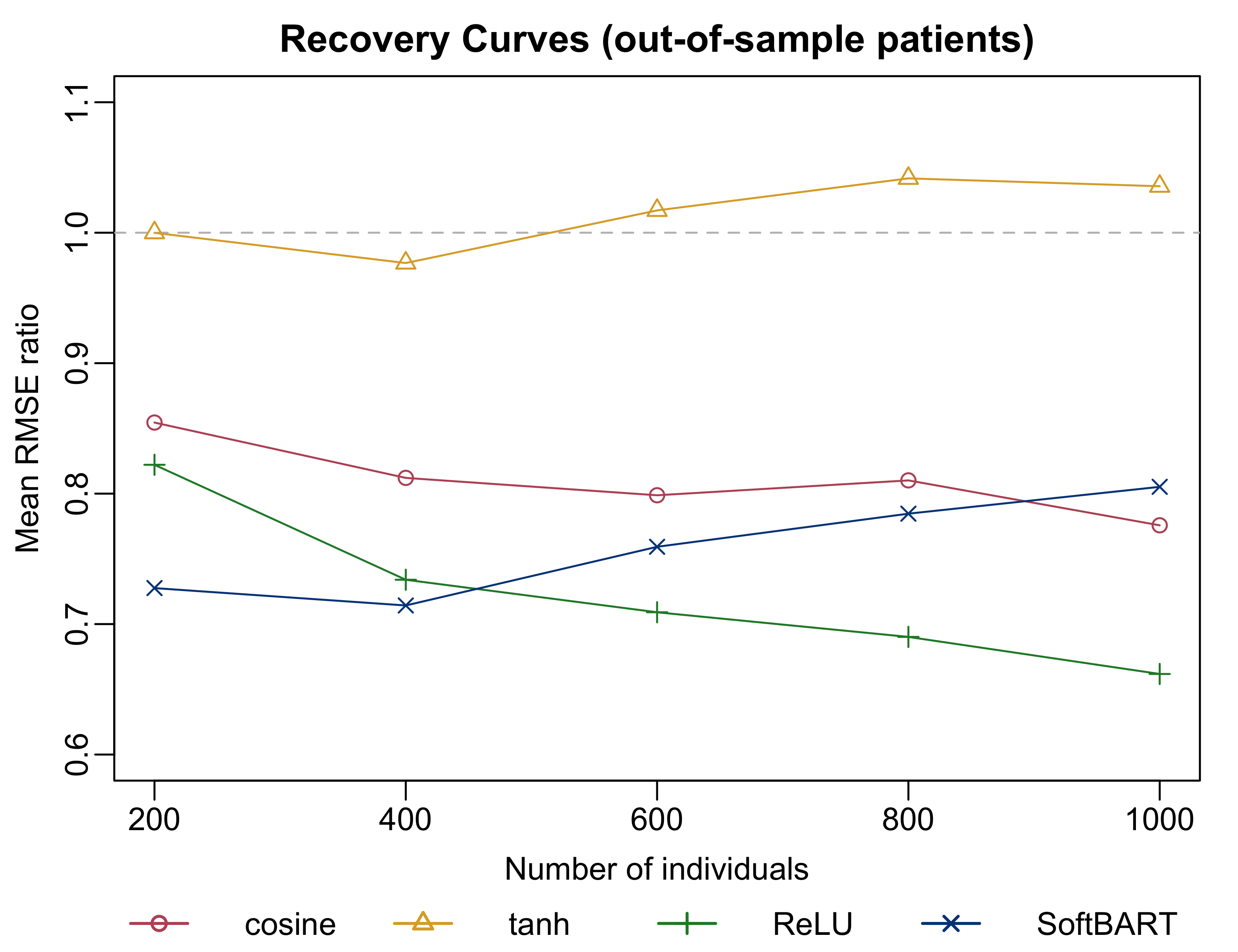}
    \caption{}
    \label{fig:rc_out_rmse}
    \end{subfigure}
    \begin{subfigure}{0.32\textwidth}
    \centering
    \includegraphics[width = \textwidth]{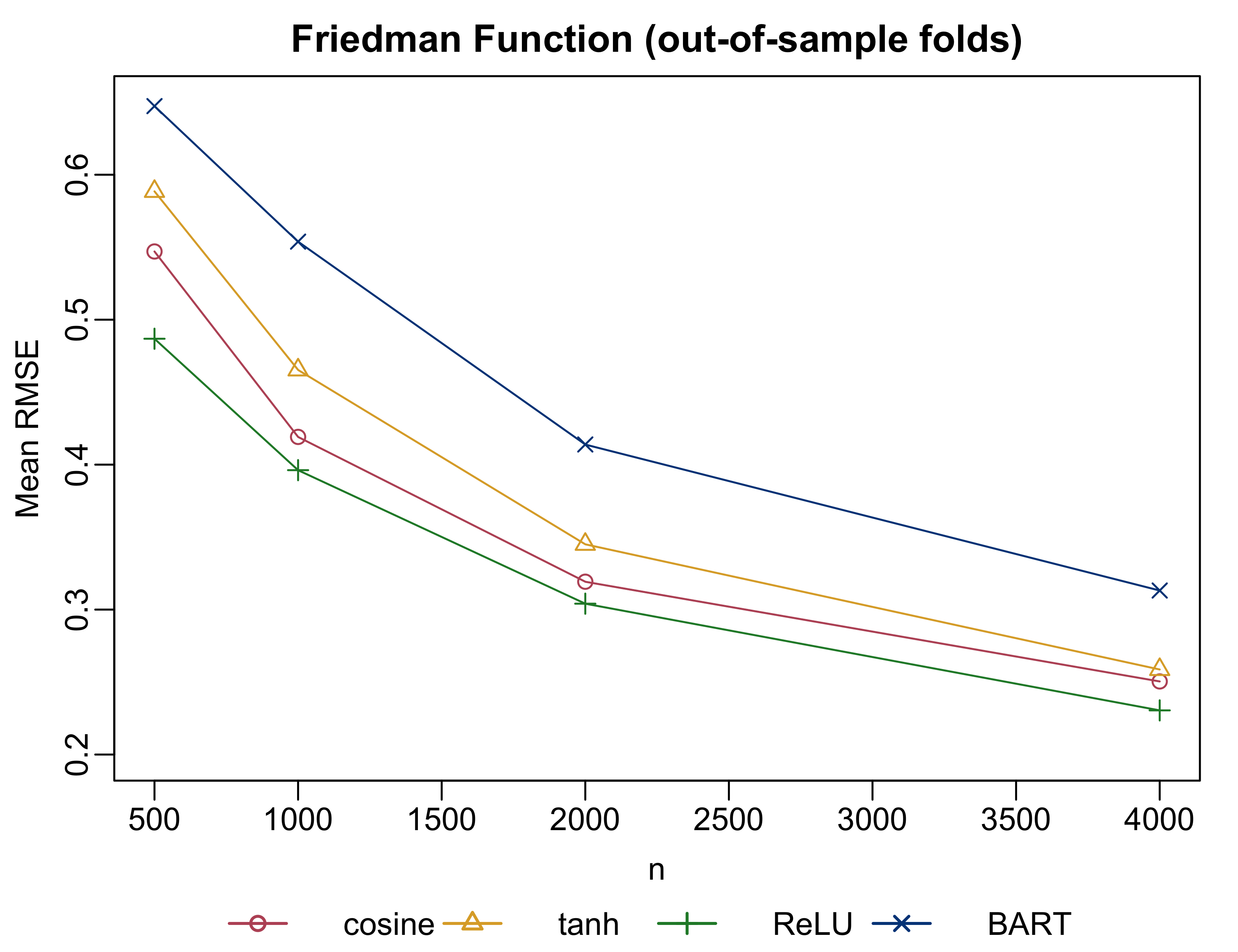}
    \caption{}
    \label{fig:friedman_lineplot}
    \end{subfigure}
    \caption{Prediction error for recovery curve simulation study on (a) in-sample patients and (b) out-of-sample-patients and (c) Friedman function simulation study on out-of-sample folds. Panels (a) and (b) depict mean RMSE relative to BART, while panel (c) depicts actual RMSE.}
    \label{fig:sim_study_results}%
\end{figure}%
\fi%
}
\switchref{\Cref{fig:rc_in_rmse,fig:rc_out_rmse}}{\Cref{fig:experiments}(A) and \Cref{fig:experiments}(B)} show how each method's in-sample patient and out-of-sample patient errors compare to the original BART model's as functions of the number of training patients $n.$
For predicting recovery curves for in-sample patients, each ridgeBART model and SoftBART outperformed the original BART model.
However, as we increased the number of training individuals $n,$ the gap between BART, SoftBART, and ridgeBART with $\tanh$ activation diminished. 
In contrast, ridgeBART with cosine and ReLU activations continued to improve relative to the original BART model, achieving 20-40\% reductions in RMSE. 
We additionally observed that ridgeBART with cosine and ReLU activation substantially out-performed BART when forecasting recovery curves for new patients and were competitive with SoftBART.

\subsection{Generalized Smoothing}
\label{sec:gen_smooth}

Although motivated by targeted smoothing, ridgeBART can also be used for generic nonparametric regression.
That is, given observations $y \sim \normaldist{f(\bx)}{\sigma^{2}},$ we can run ridgeBART with $\bz = \bx$ to approximate $f$ with piecewise-smooth functions. 
For each $n \in \{500, 1000, 2000, 4000\},$ we generated a dataset $(\bx_{1}, y_{1}), \ldots, (\bx_{n}, y_{n})$ where each $\bx_{i}$ was drawn uniformly from $[0,1]^{5}$ and $y_{i} \sim \normaldist{f(\bx_{i})}{1}$ where $f(\bx) = \sin(\pi x_1 x_2) + 20 (x_3 - 0.5)^2 + 10 x_4 + 5 x_5$ is the Friedman function frequently used to benchmark BART extensions. 

For each dataset, we compared ridgeBART's, BART's, and GP-BART's out-of-sample predictive performance using 25-fold cross-validation.
As in our targeted smoothing experiments, we attempted to simulate 10 Markov chains for 2000 iterations each for each method. 
We were only able to fit a GP-BART model to the data with $n = 500;$ running 10 GP-BART chains exceeded the 72-hour time limit set by our high-throughput cluster for $n > 500$.
In comparison, for $n = 1000,$ simulating a single ridgeBART chain took 158 seconds (cosine), 151 seconds (tanh), 77 seconds (ReLU) and simulating a single BART chain took 140 seconds; see \suppref{\Cref{sec:timing}}{Section B2 of the Supplementary Materials} for further timing comparisons.
\switchref{\Cref{fig:friedman_lineplot}}{\Cref{fig:experiments}(C)} shows how the average out-of-sample RMSE decreases as a function of $n.$
Each ridgeBART model substantially outperformed the original BART model for all values of $n.$

\subsection{Shot Charts}
\label{sec:shot_charts}

Returning to our motivating example, recall that our goal is to estimate the field goal percentage (FG\%) of every NBA player at every location on the court.
We obtained a dataset of 196,095 shots from the 2023-2024 season, including the playoffs, using the \textbf{hoopR} package \citep{hoopR} consisting of a binary indicator $y$ of whether a shot was made ($y = 1$) or missed ($y = 0$); the horizontal and vertical coordinates of the shot ($\bz$); and a vector of covariates $\bx$ about the player who attempted the shot.
In addition to the categorical player identity, $\bx$ included their position (guard, forward, or center), height, and weight.
By including these covariates into our model, we can ``borrow strength'' position and physical characteristics.

Using 20-fold cross-validation, we compared the predictive quality of ridgeBART models fit with different activation functions to three alternative models.
First, we fit two generalized additive models (GAM) to these data using the \textbf{mgcv} package \citep{mgcv}. 
The first pooled together data from all players, estimating FG\% as a function of location alone.
The second model fits separate player-specific GAMs.
Interestingly, the completely pooled model performed substantially better out-of-sample than the un-pooled, player-specific models.  
We additionally compare the ridgeBART and GAM fits to that produced by BART. 

Averaging across the folds, ridgeBART's out-of-sample log-losses were 0.640 (cosine), 0.642 (tanh), and 0.641 (ReLU), all of which are smaller than those for BART (0.643), the completely pooled GAM (0.645), and the player-specific GAM (0.687); see \suppref{\Cref{fig:shot_chart_cv2}}{Figure B4 of the Supplementary Materials} for a fold-by-fold comparison.
Unlike our synthetic data experiments, for which the ReLU activation was the best, running ridgeBART with cosine activation yielded the most accurate predictions for these data.
Simulating a single MCMC chain for 2000 iterations took ridgeBART 42.8 minutes (cosine), 42.5 minutes (tanh), 17.0 minutes (ReLU) and simulating a single chain for BART took 39.3 minutes.
{
\if\jrssb1%
\begin{figure}[ht]
\centering
\includegraphics[width = \textwidth]{figures/fig4}
\caption{Shot charts estimated for Stephen Curry (A), Ben Simmons (B), Giannis Antetokounmpo (C), and LeBron James (D) using ridgeBART with cosine activiation}
\label{fig:shot_chart}
\end{figure}
\else%
\begin{figure}[ht]
\centering
\begin{subfigure}[b]{0.23\textwidth}
\centering
\includegraphics[width = \textwidth]{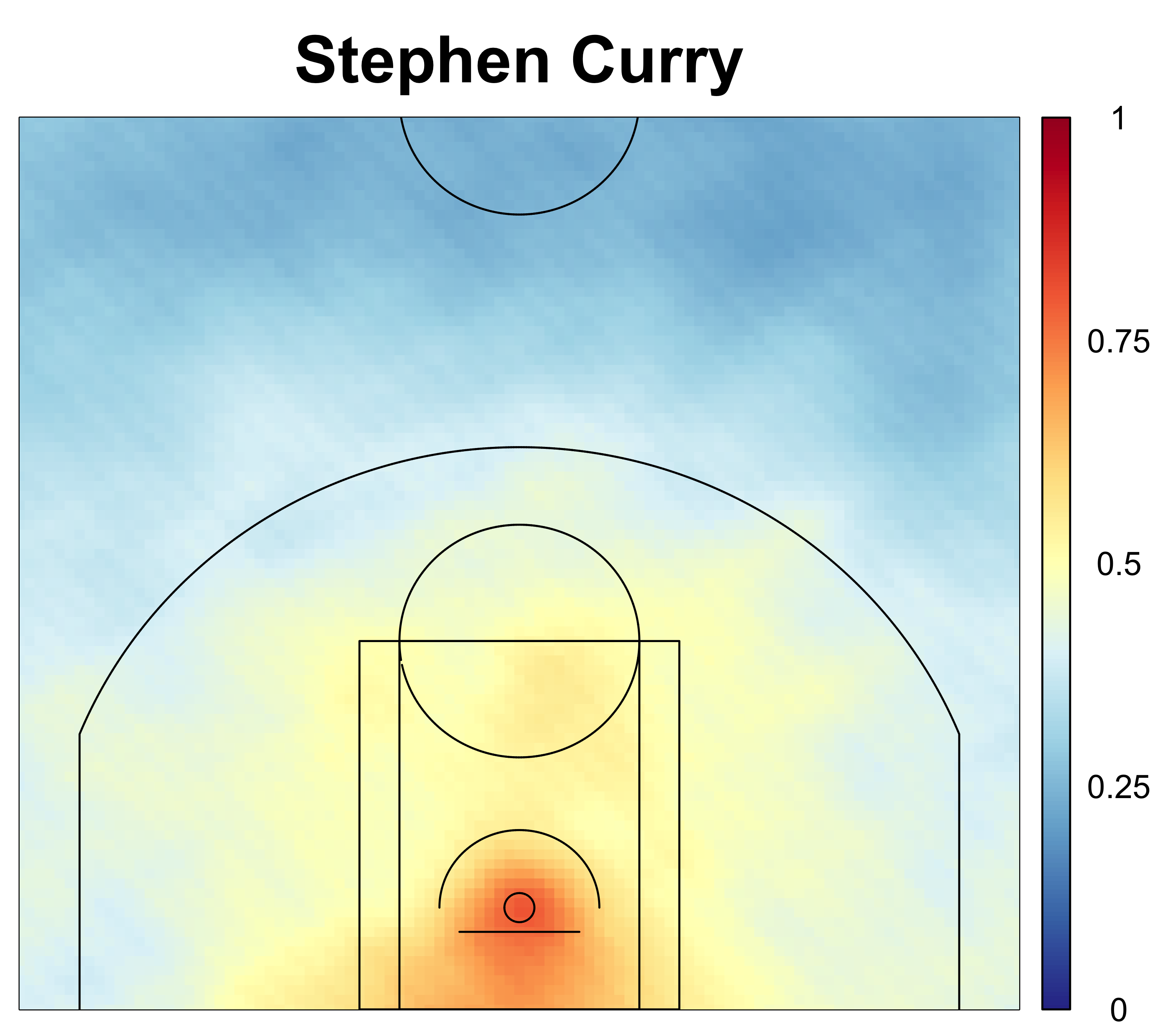}
\caption{}
\label{fig:chef}
\end{subfigure}
\begin{subfigure}[b]{0.23\textwidth}
\centering
\includegraphics[width = \textwidth]{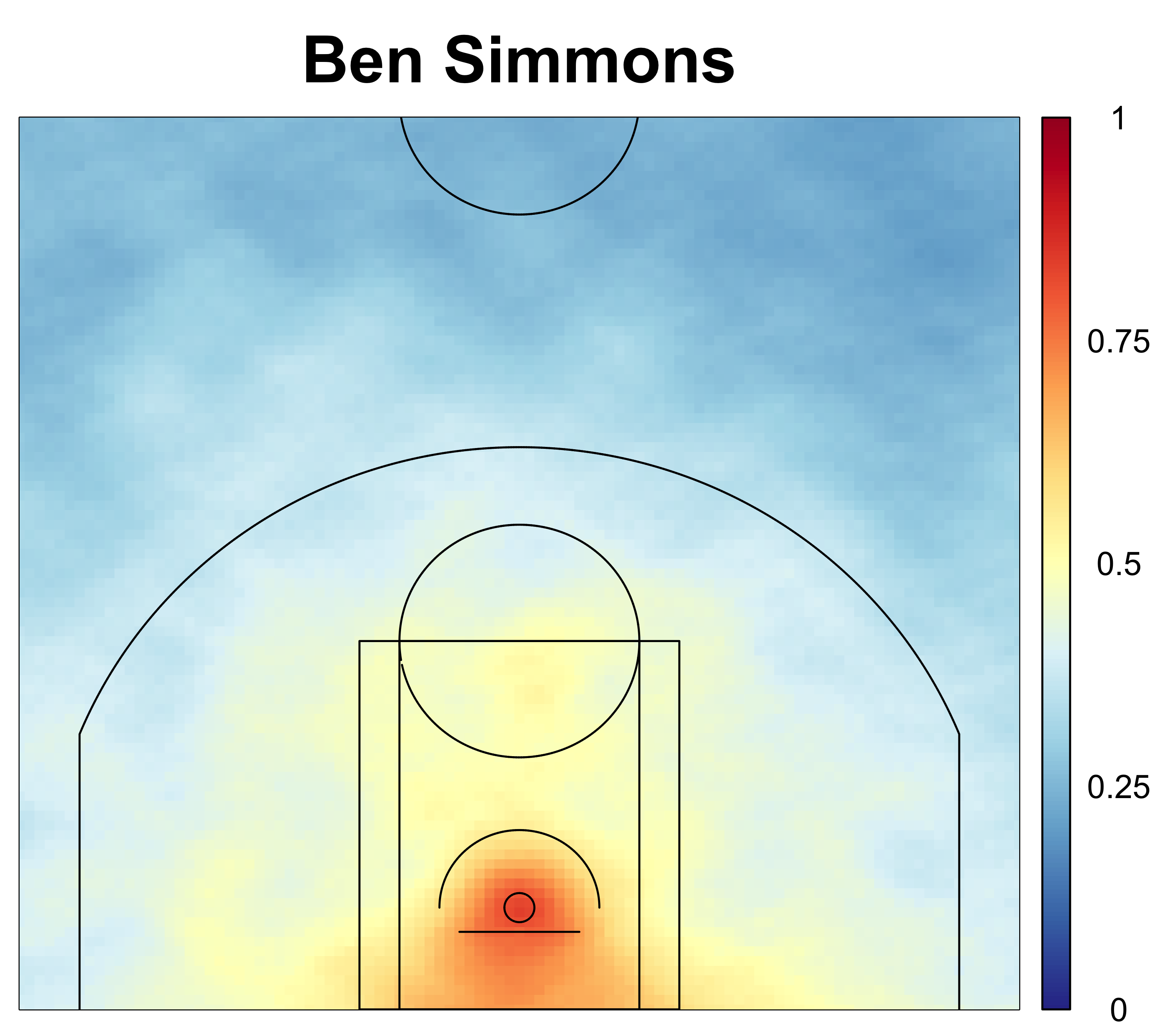}
\caption{}
\label{fig:simmons}
\end{subfigure}
\begin{subfigure}[b]{0.23\textwidth}
\centering
\includegraphics[width = \textwidth]{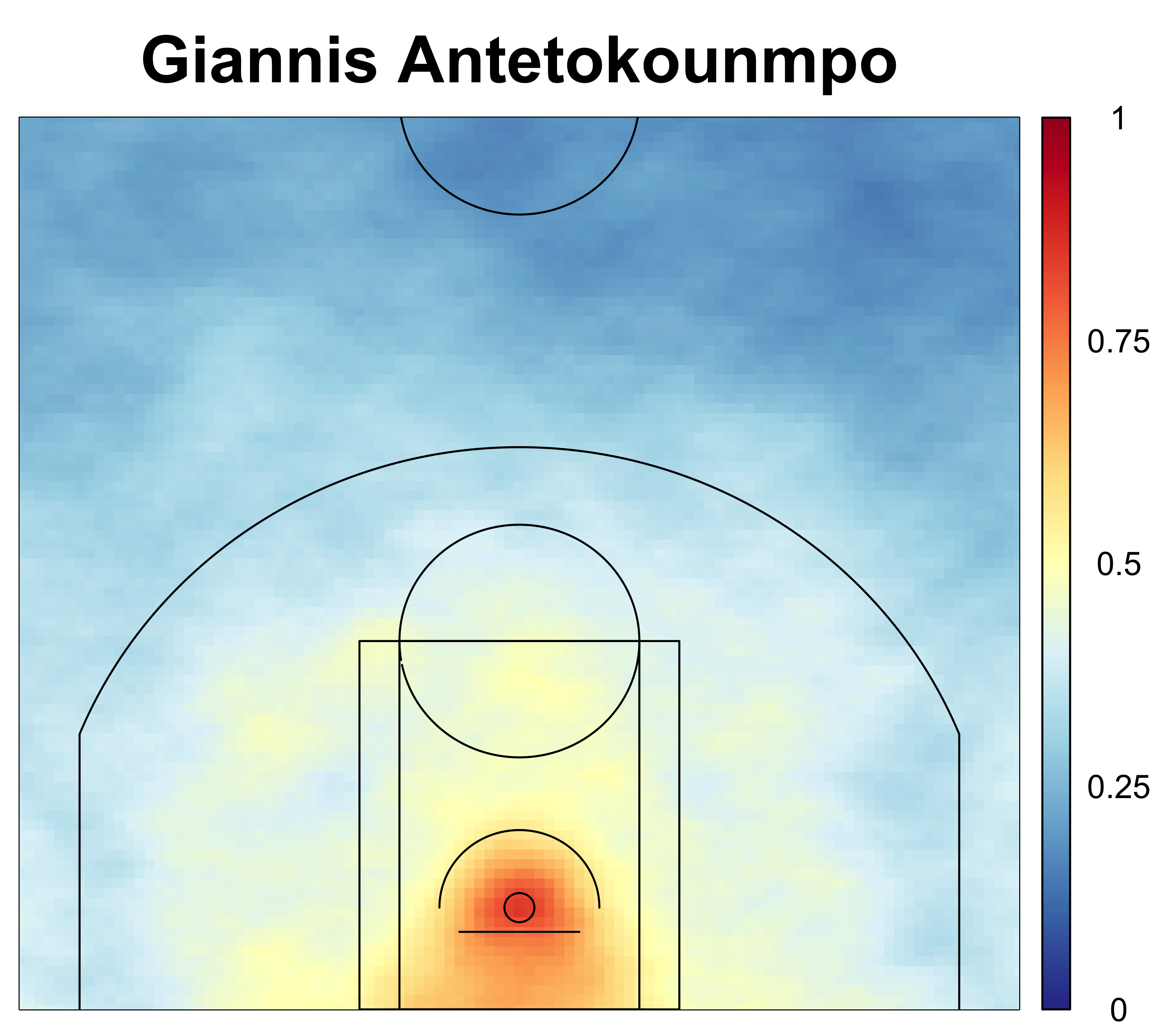}
\caption{}
\label{fig:simmons}
\end{subfigure}
\begin{subfigure}[b]{0.23\textwidth}
\centering
\includegraphics[width = \textwidth]{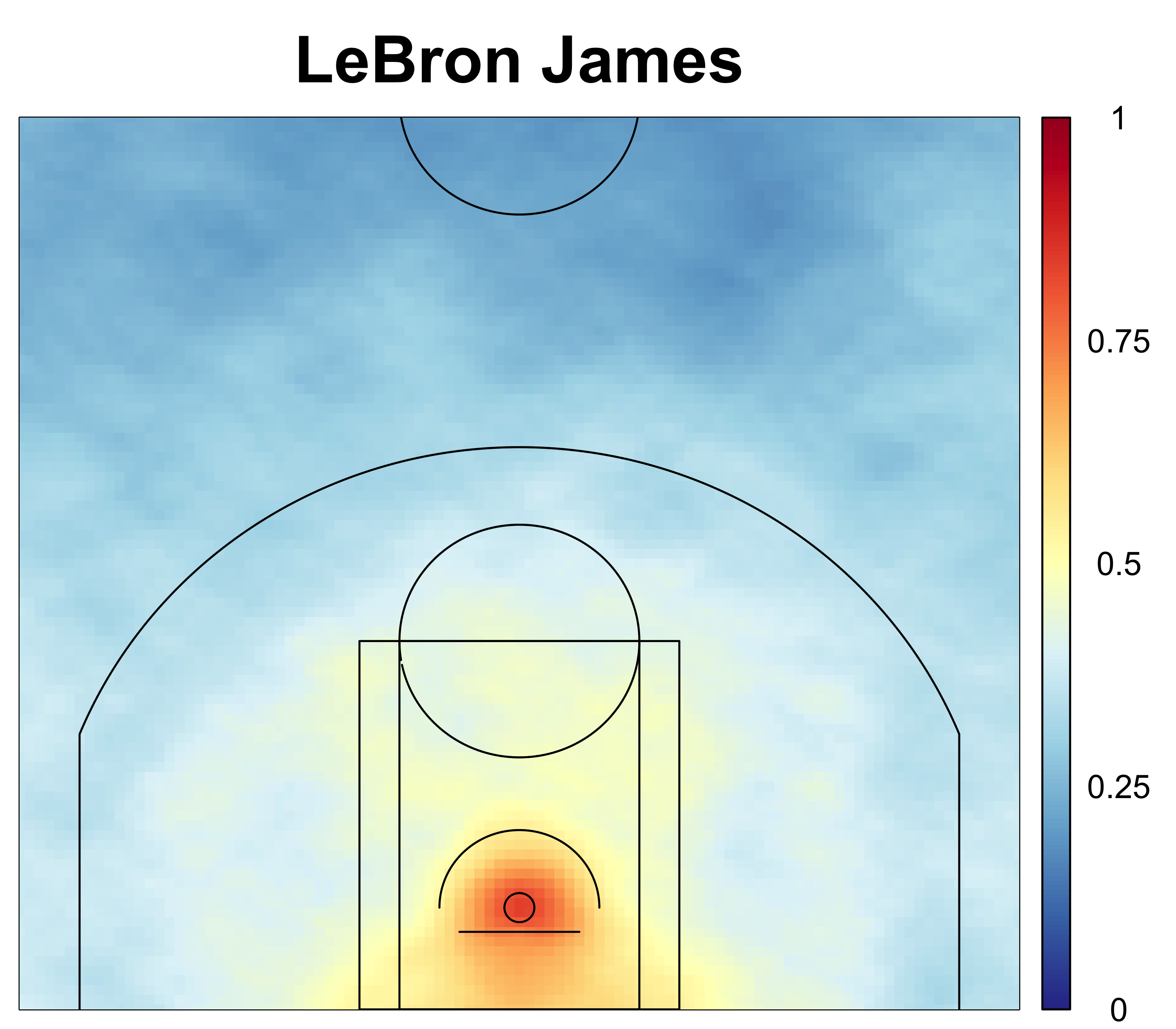}
\caption{}
\label{fig:simmons}
\end{subfigure}
\caption{Shot charts estimated for Stephen Curry (a), Ben Simmons (b), Giannis Antetokounmpo (c), and LeBron James (d) using ridgeBART with cosine activiation}
\label{fig:shot_chart}
\end{figure}%
\fi%
}
\Cref{fig:shot_chart} compares the FG\% heat maps estimated by ridgeBART with cosine activation for Stephen Curry, Ben Simmons, Giannis Antetokounmpo, and LeBron James.
These players are known to score in different ways.
Simmons and Antetokounmpo are known for being physical and scoring at close-range, while Curry is instead known for being a long-range shooter.
James, in contrast, is considered to be a good all-around scorer.
The comparisons between these players are illuminated by the ridgeBART shot charts, all of which do not have the same pathological sharp discontinuities as the BART shot charts.
Clearly, we can see that Stephen Curry is a significantly better shooter at all areas of the court compared to James, Antetokounmpo, and Simmons.
James and Antetokounmpo have very similar looking shot charts, where the main difference is that Antetokounmpo is slightly better in the immediate vicinity of the basket.
These results suggest that, despite their differences in play-style, James and Antetokounmpo are equally-skilled at shooting the basketball.

\section{Discussion}
\label{sec:discussion}
We have introduced ridgeBART, which expands the flexibility of the original BART model by replacing its constant jump outputs with linear combinations of ridge functions.
In our experiments, ridgeBART outperformed the original BART model and GP-based extensions in terms of predictive accuracy and computation time on both targeted smoothing and general regression problems.
Furthermore, we showed that, with minor modifications to the decision tree prior, the ridgeBART posterior consistently estimates functions in a piecewise anisotropic H\"older space.

There are several potential avenues for future development of ridgeBART.
First, the decision trees used by ridgeBART can be made even more flexible by incorporating oblique decision rules, which partition the continuous covariates along arbitrary hyperplanes.
Second, while ridgeBART effectively models smooth functions, it does not perform variable selection and remains sensitive to the curse of dimensionality. If the true function depends only on a subset $d \ll p$ of covariates, ridgeBART does not automatically achieve the improved convergence rate $n^{-\nicefrac{\alpha}{2 \alpha+d}}$ obtained by variable selection approaches in regression trees \citep{Rockova2019,Jeong2023}.
Compared to ensembles of piecewise constant trees, variable selection with ridgeBART is more delicate, as a covariate may exert influence through the decision rules and through the functional output in each leaf.
One potential approach to encourage variable selection would be to combine \citet{Linero2018}'s sparse Dirichlet priors on the decision rules with a spike-and-slab prior on the $\omega$'s within each leaf.

\bibliographystyle{apalike}
\bibliography{rffBART_refs}

\newpage
\appendix

\renewcommand{\thefigure}{\thesection\arabic{figure}}
\renewcommand{\thetable}{\thesection\arabic{table}}
\renewcommand{\theequation}{\thesection\arabic{equation}}
\renewcommand{\thelemma}{\thesection\arabic{lemma}}
\renewcommand{\thetheorem}{\thesection\arabic{theorem}}

\setcounter{figure}{0}
\setcounter{equation}{0}
\setcounter{table}{0}
\section{Hyperparameter sensitivity analysis}
\label{sec:hyperparameters}
Recall from \suppref{\Cref{sec:ridgeBART_prior}}{Section 3.1 of the main text} that the ridgeBART prior depends on several hyperparameters.
We performed a sensitivity analysis using the Friedman function data generating process described in \suppref{\Cref{sec:gen_smooth}}{Section 5.2 of the main text} with $n = 1000.$
We specifically assessed sensitivity to the number of trees $M$; the number of basis elements in each leaf $D$; and the hyperparameters for the scale of $\omega_\ell$, $\nu$ and $\lambda$.

\textbf{Sensitivity to $M$ and $D$}. Fixing $\nu = 3$ and $\lambda = 0.788,$ we ran ridgeBART with each combination of $M \in \{10, 50, 100\}$ and $D \in \{1, 5, 10\}.$
\Cref{fig:basis_sensitivity} shows boxplots of the out-of-sample RMSE across 25 cross-validiation folds.
We find that $M = 50$ and $D = 1$ provide the best results on average and take significantly less time to compute than similarly-performing settings.

\begin{figure}[ht]
\centering
\begin{subfigure}[b]{0.32\textwidth}
\centering
\includegraphics[width = \textwidth]{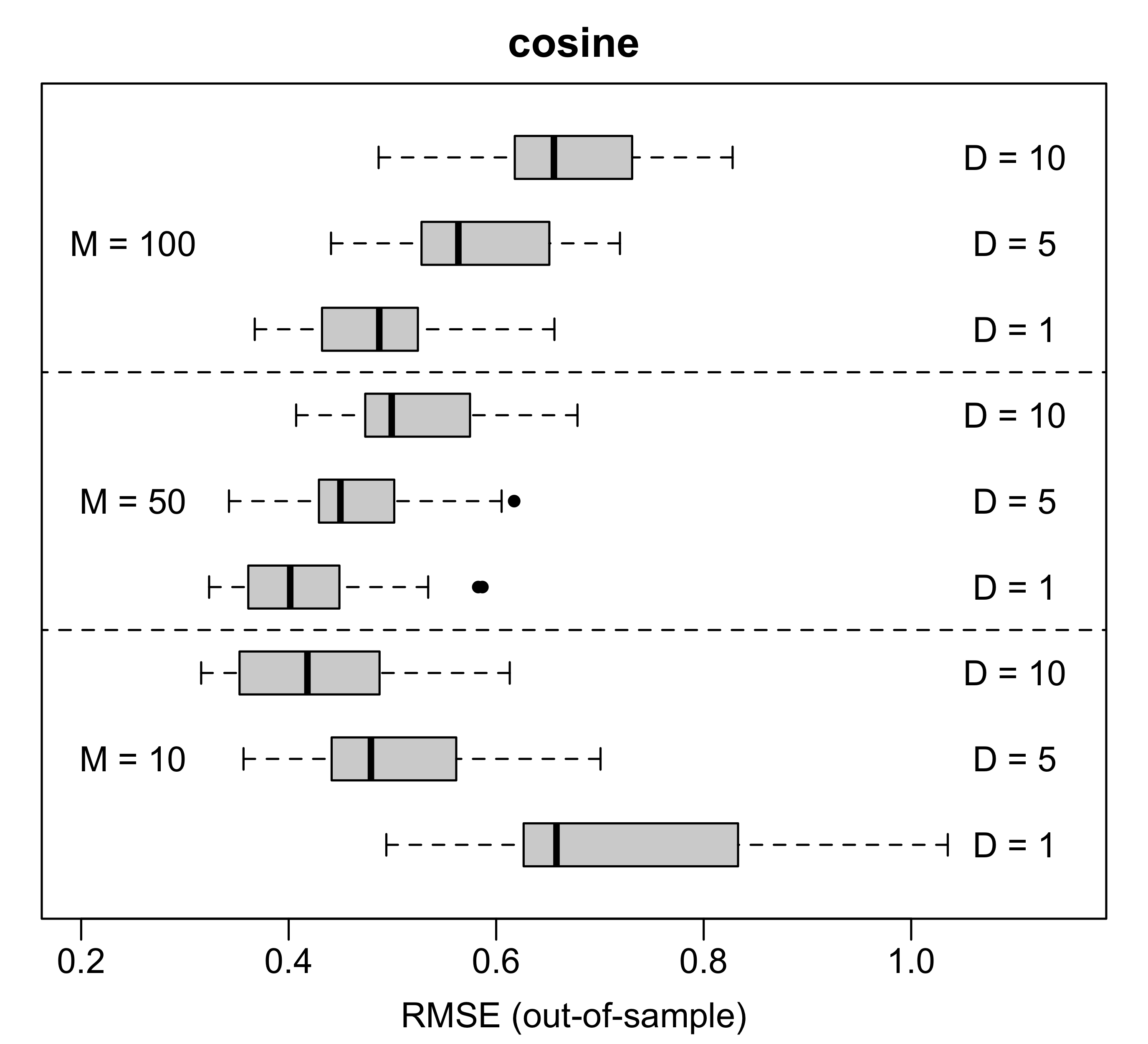}
\caption{}
\label{fig:hp_basis_cos}
\end{subfigure}
\begin{subfigure}[b]{0.32\textwidth}
\centering
\includegraphics[width = \textwidth]{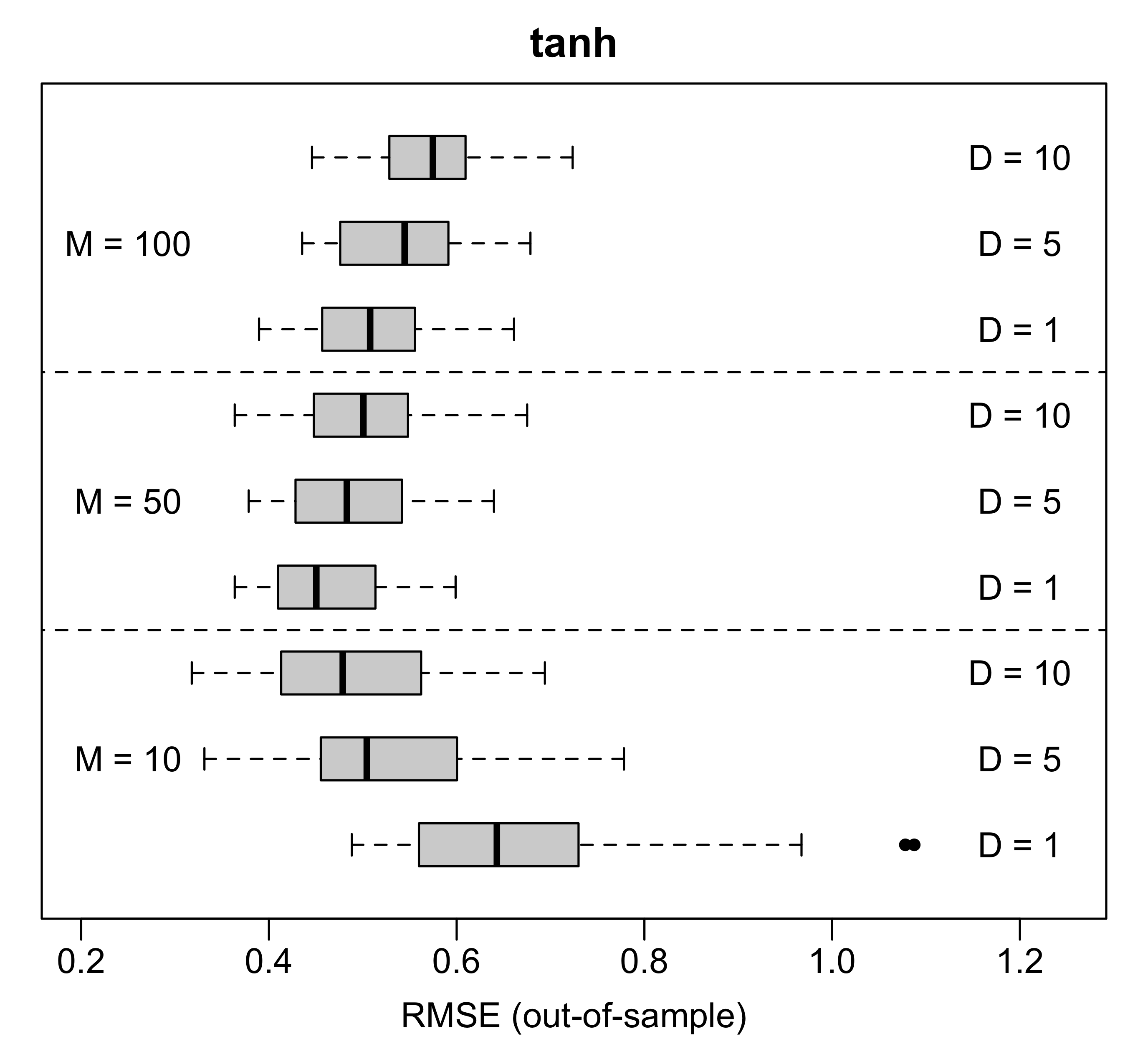}
\caption{}
\label{fig:hp_basis_tanh}
\end{subfigure}
\begin{subfigure}[b]{0.32\textwidth}
\centering
\includegraphics[width = \textwidth]{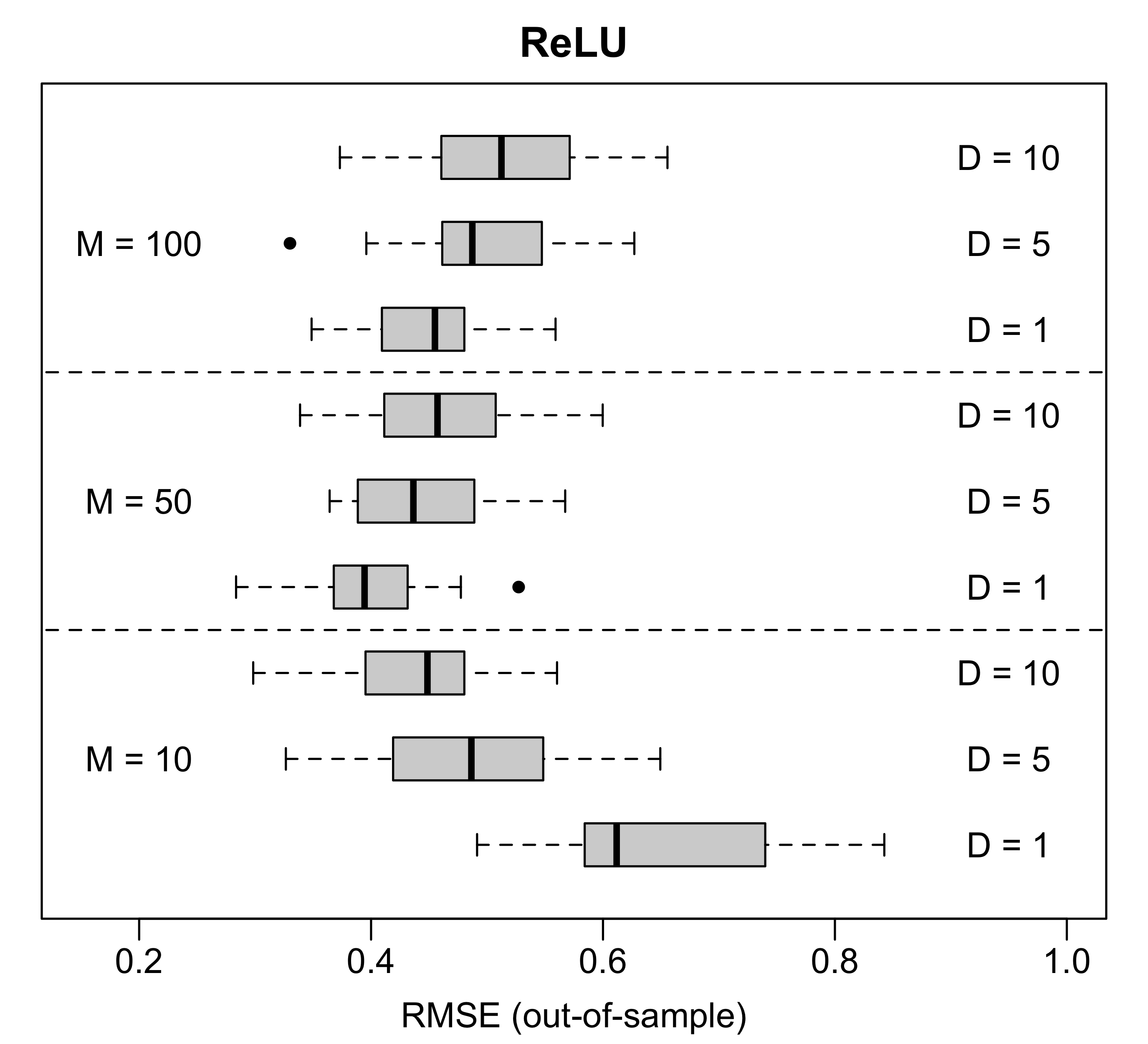}
\caption{}
\label{fig:hp_basis_relu}
\end{subfigure}
\caption{Out-of-sample RMSE across 25 train-test splits for combinations of $M$ and $D$.}
\label{fig:basis_sensitivity}
\end{figure}

%

\textbf{Sensitivity to $\nu$ and $\lambda$.} Fixing $M = 50$ and $D = 1,$ to assess sensitivity to the $\omega_\ell$ scale prior hyperparameters, we first fix $\nu = 3$ and ran ridgeBART with each combination of $\lambda$ such that $\P(\rho < q) = p$ where $p \in \{0.25, 0.5, 0.75\}$ and $q \in \{0.5, 1, 2\}.$
\Cref{fig:scale_sensitivity} shows boxplots of the RMSE on 25 out-of-sample test folds.
For these data, our recommended default values $\P(\rho < 1) = 0.5$ perform well across all activation functions. 

\begin{figure}[ht]
\centering
\begin{subfigure}[b]{0.32\textwidth}
\centering
\includegraphics[width = \textwidth]{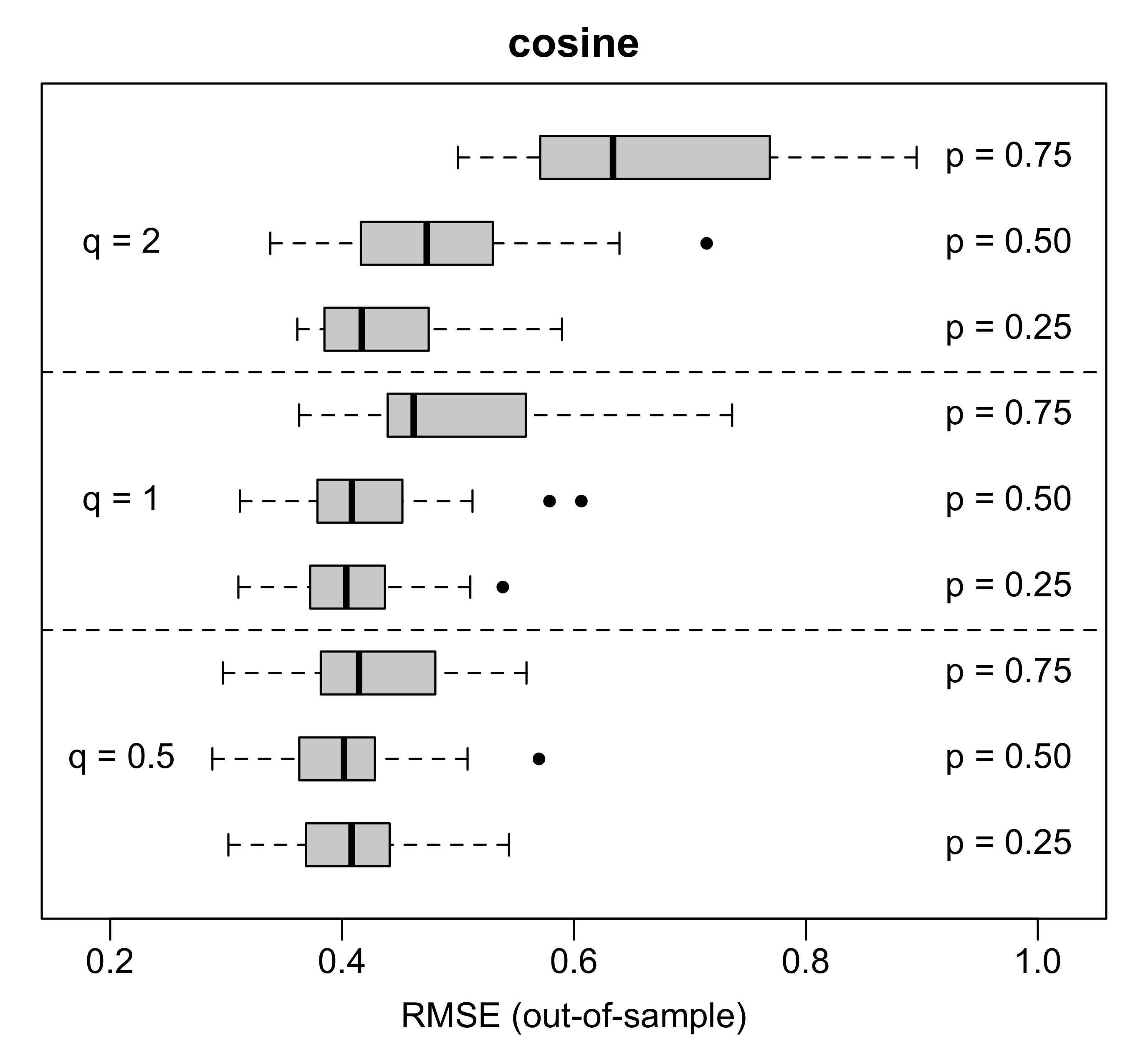}
\caption{}
\label{fig:hp_scale_cos}
\end{subfigure}
\begin{subfigure}[b]{0.32\textwidth}
\centering
\includegraphics[width = \textwidth]{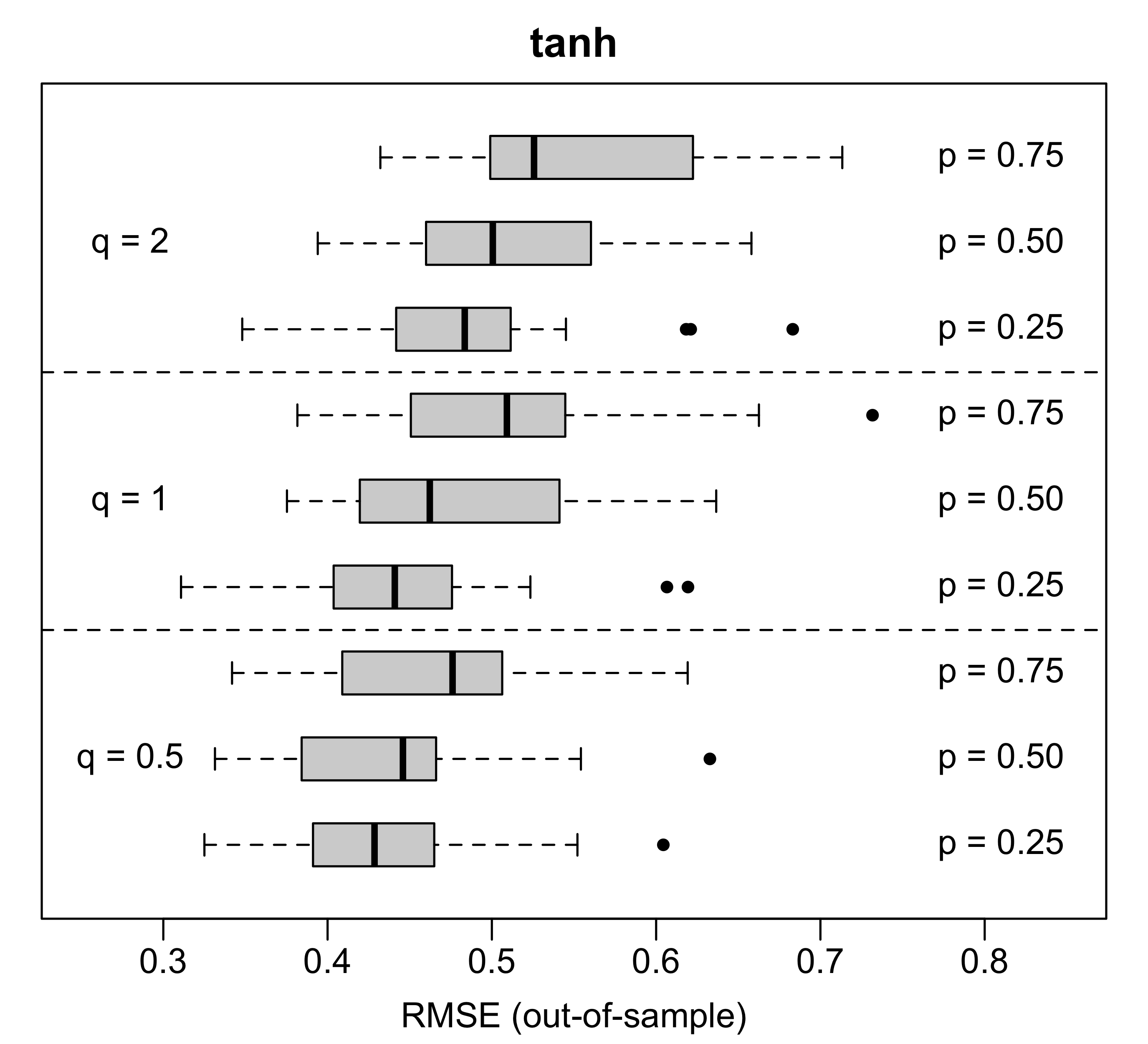}
\caption{}
\label{fig:hp_scale_tanh}
\end{subfigure}
\begin{subfigure}[b]{0.32\textwidth}
\centering
\includegraphics[width = \textwidth]{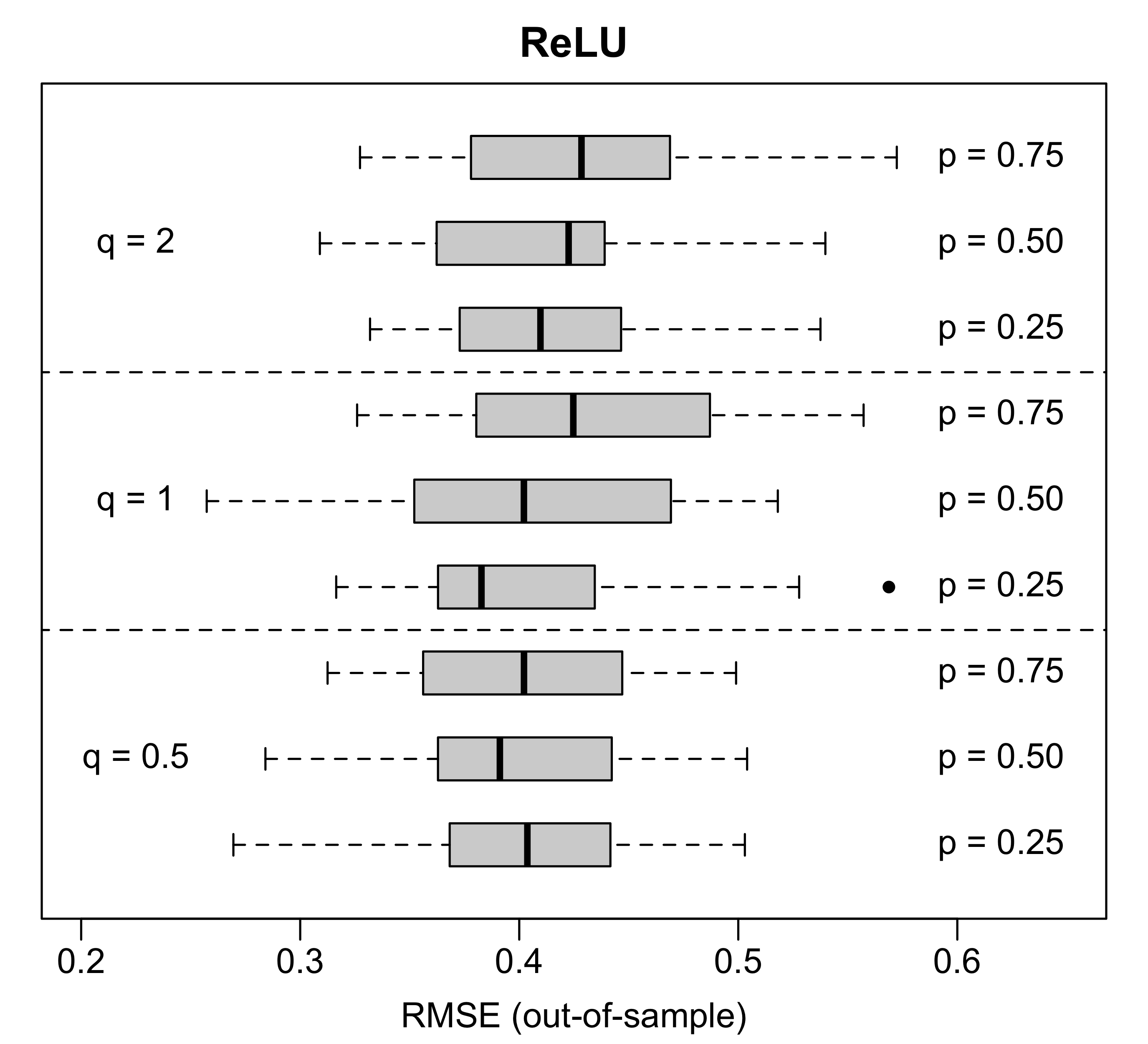}
\caption{}
\label{fig:hp_scale_relu}
\end{subfigure}
\caption{Out-of-sample RMSE on 25 train-test splits for combinations of $p$ and $q$.}
\label{fig:scale_sensitivity}
\end{figure}

\setcounter{figure}{0}
\setcounter{equation}{0}
\setcounter{table}{0}
\setcounter{theorem}{0}
\setcounter{lemma}{0}
\section{Additional experimental details and results}
\label{sec:additional_experiments}
\subsection{Recovery curve simulation}
\label{sec:recovery_curve_details}

In \suppref{\Cref{sec:ts}}{Section 5.1 of the main text}, we generated data to mimic a patients' recovery following a major injury or surgery.
For each patient $i = 1, \ldots, n,$ we generated $n_{i}$ noisy observations of the function $f(\bx_{i}, z_{it}) = (1 - A(\bx_{i}))(1 - B(\bx_{i}) e^{-z_{it} \times C(\bx_{i})}),$ where $\bx_{i}$ is a vector of $p = 6$ patient covariates and $z_{it} \in [0,24]$ is the observation time.
We drew the number of observations $n_{i} \sim 1 + \poisdist{3},$ so that each patient had at least one observation.
Conditional on $n_{i},$ we set the $z_{it}$'s to be uniform perturbations of a random subset of $n_{i}$ follow-up times in $\{1, 2, 4, 6, 8, 12, 16, 20, 24\},$ which correspond to common follow-up times after major surgery \citep{Cobian2024}. 
\Cref{fig:rc_example} shows twenty randomly chosen recovery curves and \Cref{fig:rc_obs} shows an example dataset consisting of $n = 600$ patients.

\begin{figure}[ht]
    \centering
    \begin{subfigure}{0.49\textwidth}
    \centering
    \includegraphics[width = \textwidth]{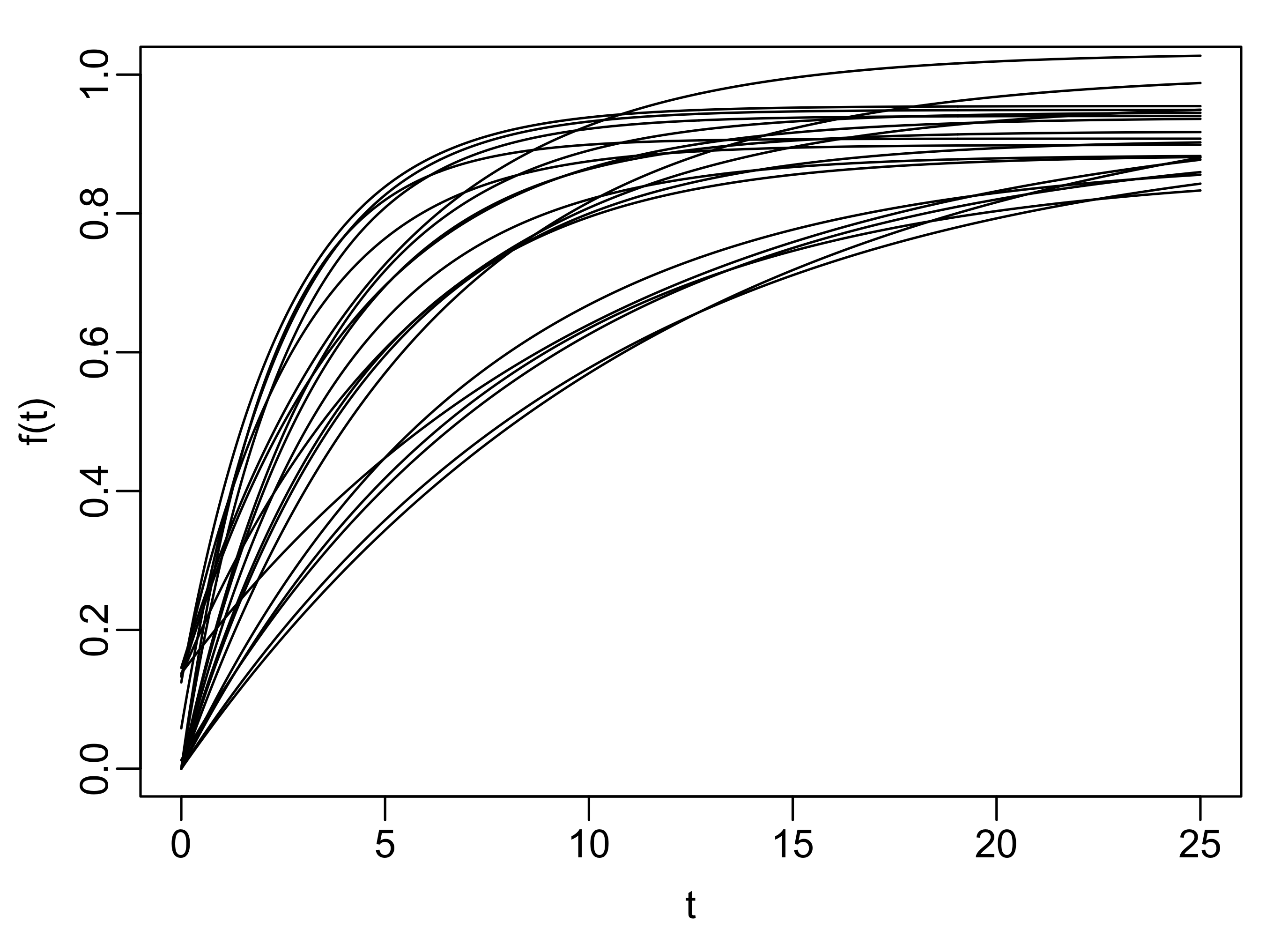}
    \caption{}
    \label{fig:rc_example}
    \end{subfigure}
    \begin{subfigure}{0.49\textwidth}
    \centering
    \includegraphics[width = \textwidth]{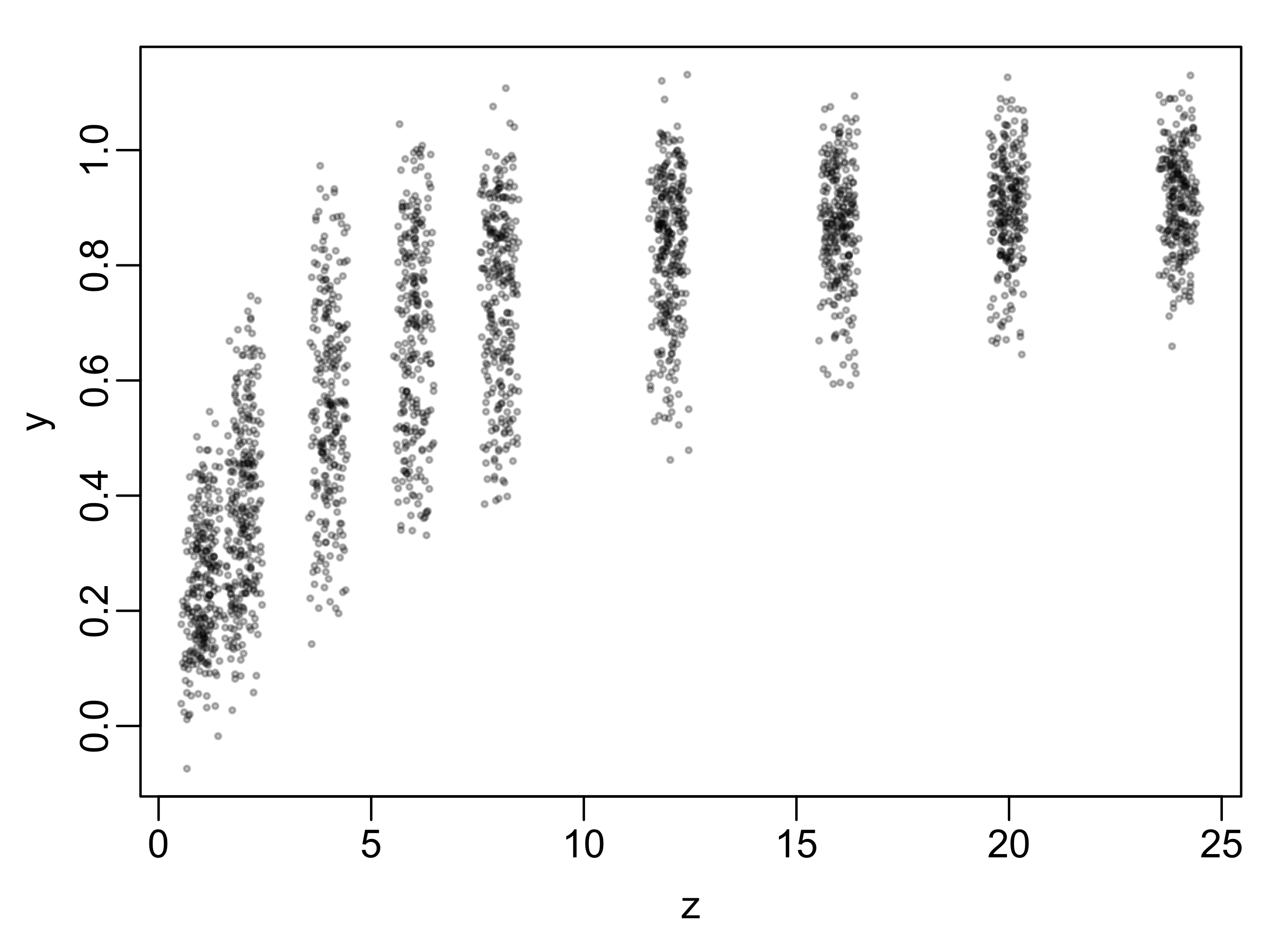}
    \caption{}
    \label{fig:rc_obs}
    \end{subfigure}
    \caption{Examples of (a) 20 randomly sampled recovery curves and (b) observations from a synthetic recovery curve dataset with 600 individuals.}
    \label{fig:rc_data}
\end{figure}

\subsection{Timing comparisons}
\label{sec:timing}
\Cref{fig:rc_time} shows the relative computation times for models fit in the recovery curve simulation study.
\Cref{fig:friedman_time} shows the number of seconds it took to fit each model in the Friedman function experiment.

\begin{figure}[ht]
    \centering
    \begin{subfigure}{0.49\textwidth}
    \centering
    \includegraphics[width = \textwidth]{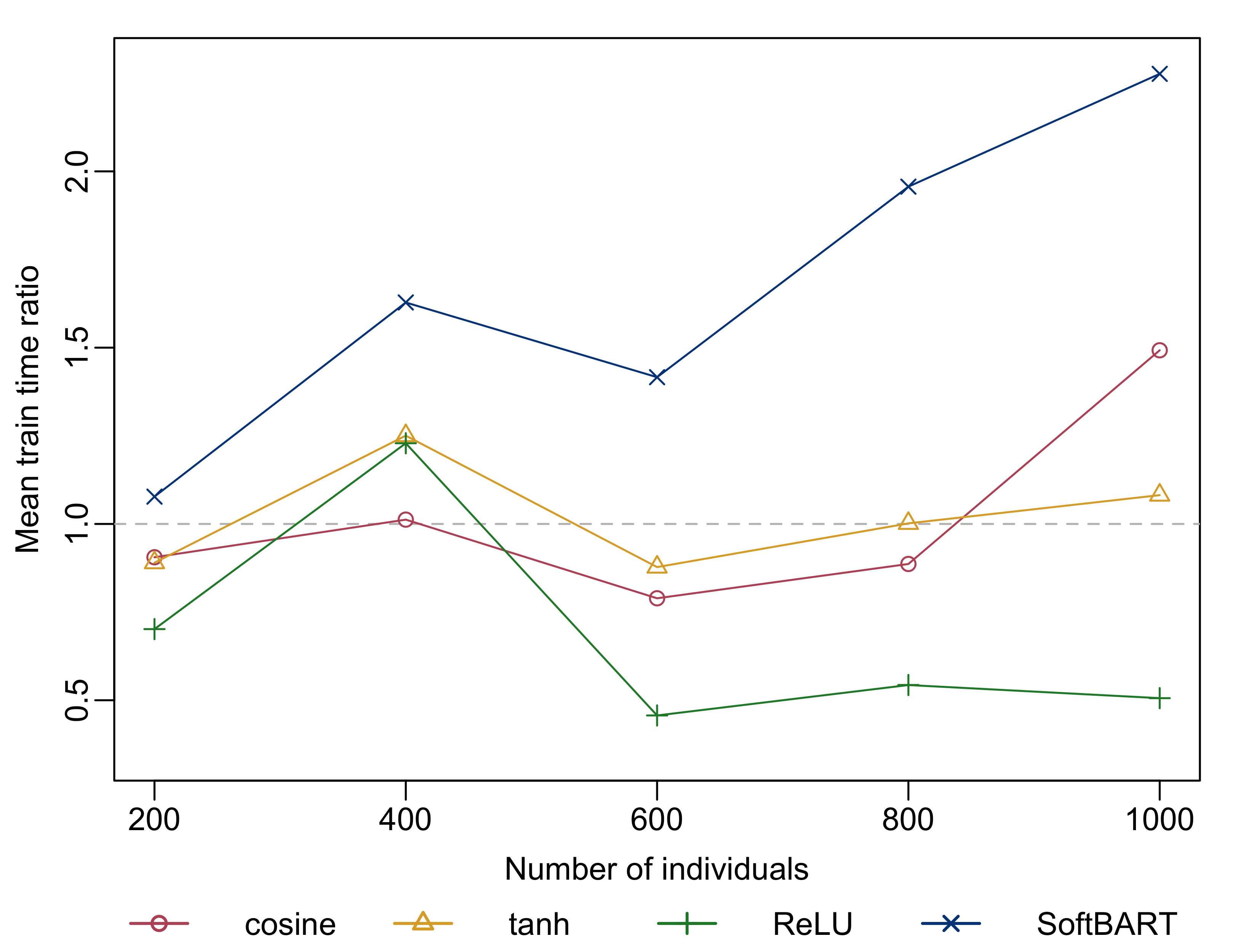}
    \caption{In-sample recovery curves}
    \label{fig:rc_time_in}
    \end{subfigure}
    \begin{subfigure}{0.49\textwidth}
    \centering
    \includegraphics[width = \textwidth]{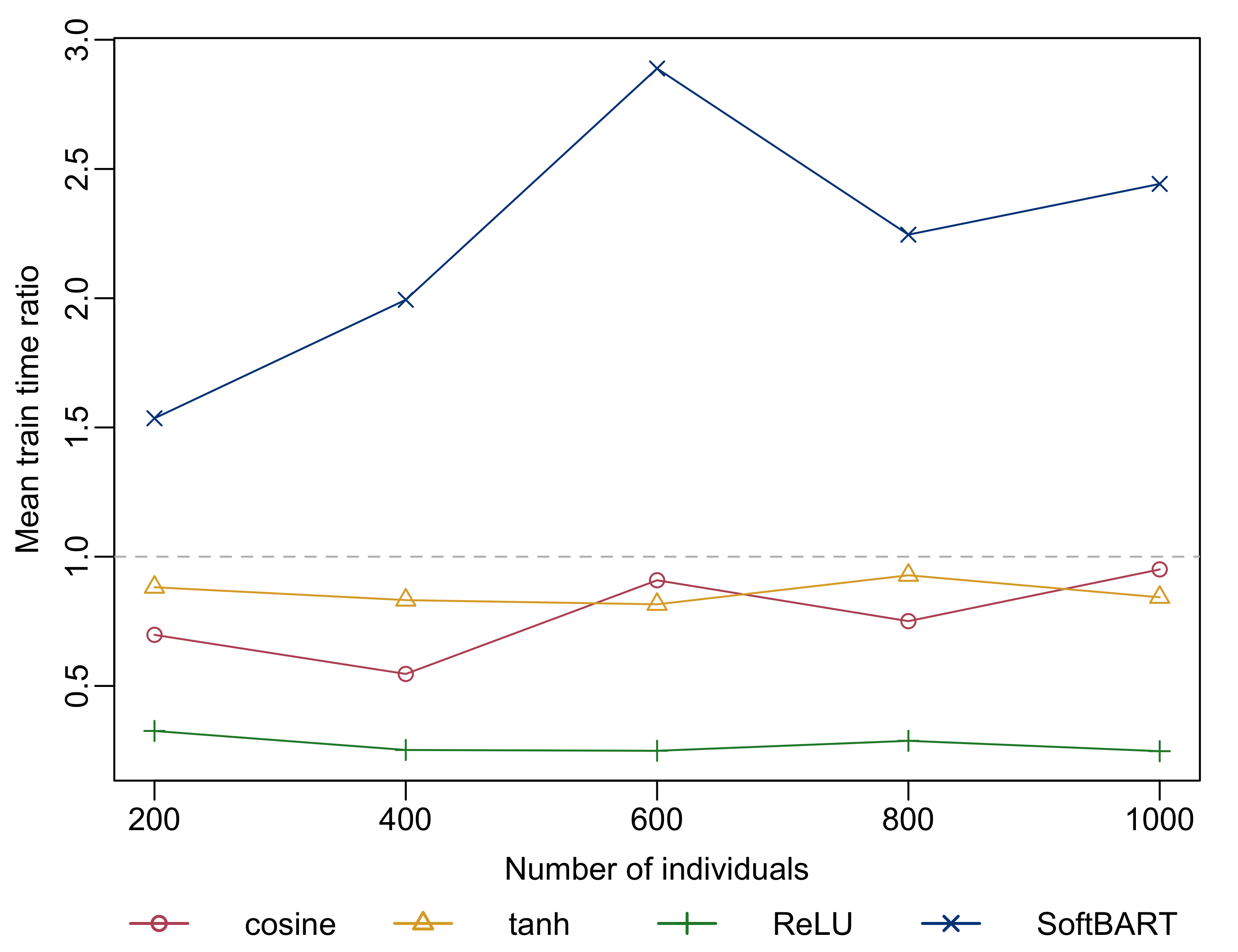}
    \caption{Out-of-sample recovery curves}
    \label{fig:rc_time_out}
    \end{subfigure}
    \caption{Average computation times for ridgeBART and SoftBART relative to BART in the recovery curve simulation study.}
    \label{fig:rc_time}
\end{figure}

\begin{figure}[ht]
    \centering
    \begin{subfigure}{0.49\textwidth}
    \centering
    \includegraphics[width = \textwidth]{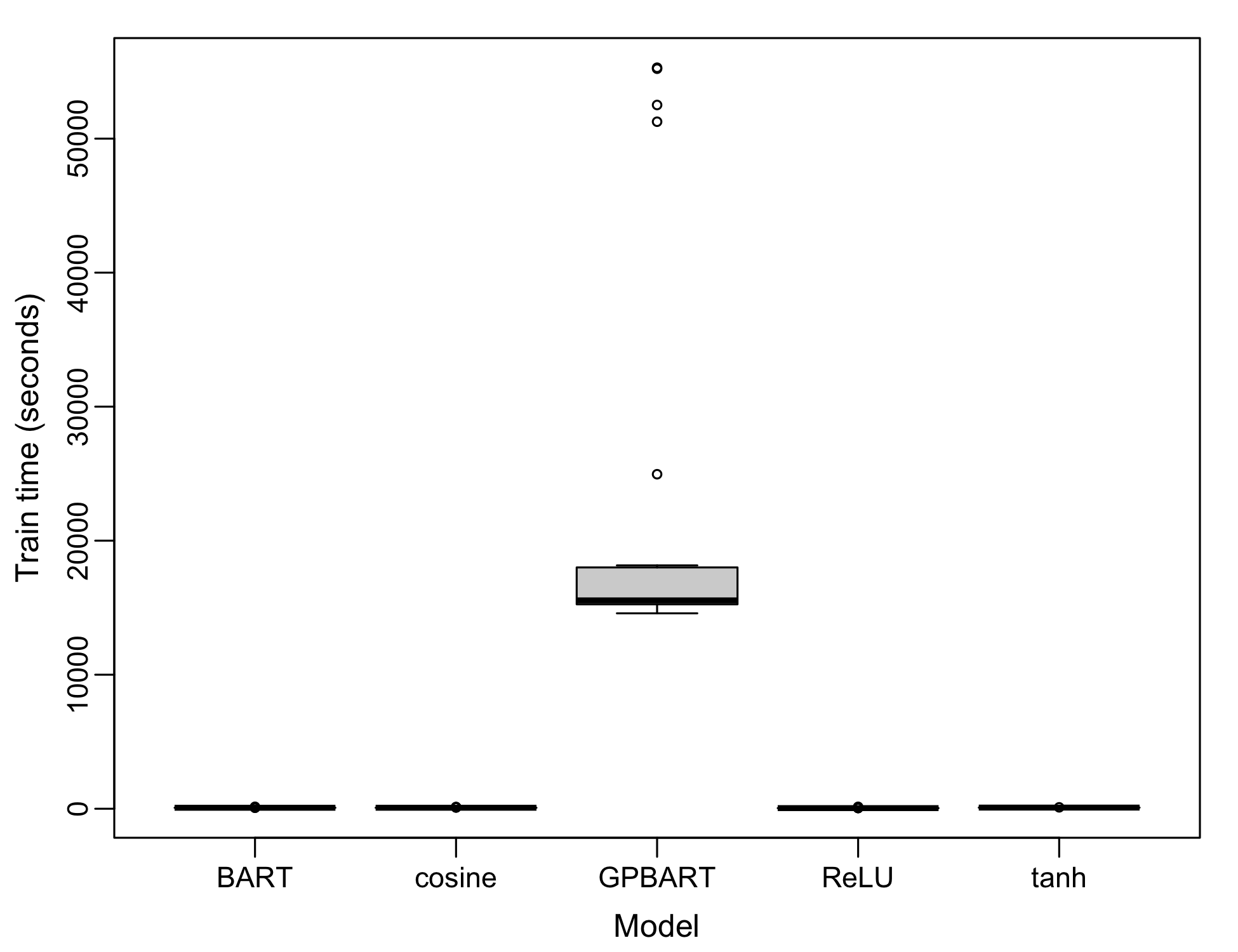}
    \caption{Computation time $n = 500$}
    \label{fig:friedman_box_time}
    \end{subfigure}
    \begin{subfigure}{0.49\textwidth}
    \centering
    \includegraphics[width = \textwidth]{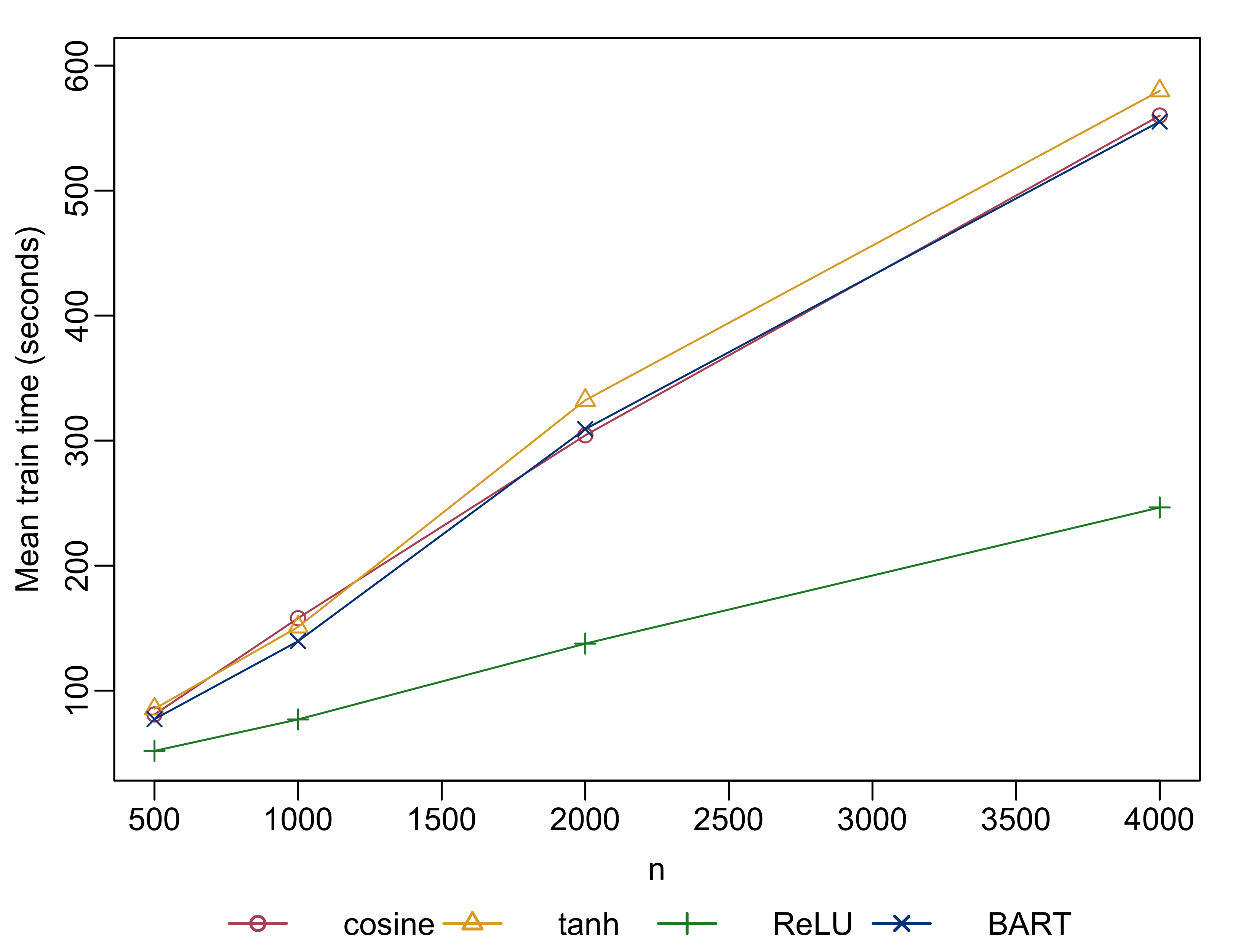}
    \caption{Computation time for multiple $n$}
    \label{fig:friedman_line_time}
    \end{subfigure}
    \caption{Time in seconds to simulate one MCMC chains of 2000 iterations each for models used in the Friedman function simulation study.}
    \label{fig:friedman_time}
\end{figure}

\subsection{NBA Shot Chart}
\label{sec:shot_chart_details}

\Cref{fig:shot_chart_cv2} compares the out-of-sample log-losses for each ridgeBART model, the completely pooled and player-specific GAMs, and BART.
We found that across folds, ridgeBART outperformed the competitors.

\begin{figure}[H]
    \centering
    \includegraphics[width = .65\textwidth]{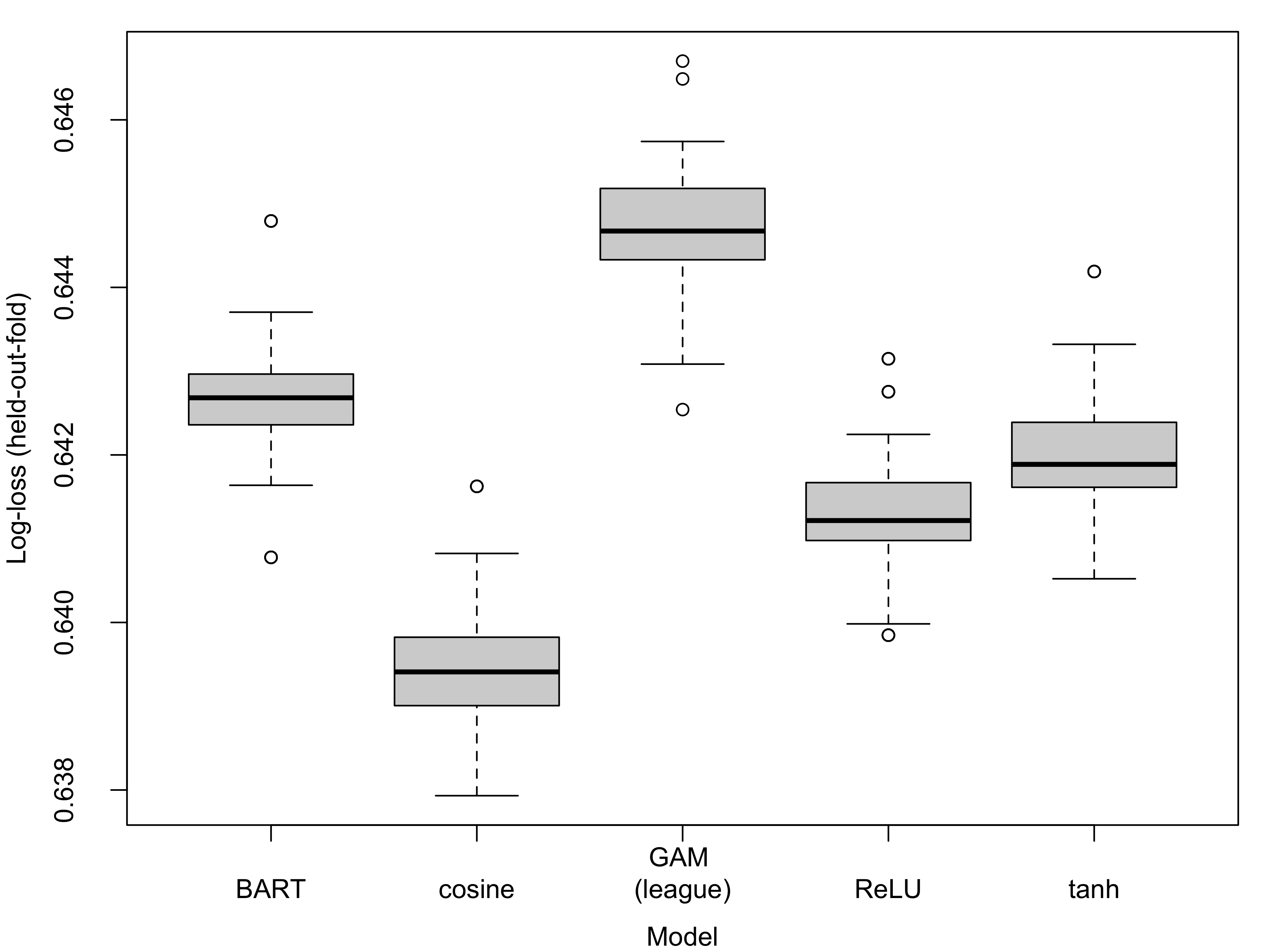}
    \caption{Performance of BART, ridgeBART, and a league-wide GAM on 20-fold cross validation of shot data from the 2023 NBA season.}
    \label{fig:shot_chart_cv2}
\end{figure}


\setcounter{figure}{0}
\setcounter{equation}{0}
\setcounter{table}{0}
\setcounter{theorem}{0}
\setcounter{lemma}{0}
\section{Proofs}
\label{sec:proofs}

\subsection{Ridge approximations anisotropic H\"older class functions}
\label{sec:globalapprox}

To rigorously prove the approximation rate of a H\"older function within each box $\Psi_r \subseteq [0,1]^p$ using ridge functions, we follow a structured approach leveraging classical approximation theory, particularly anisotropic Jackson-type inequalities (see \citet{Devore-Lorentz}) and known results on ridge function approximation. 
We thus state the local ridge approximation lemma.

\begin{lemma}\label{lem:localridge}
Let $f_0 \in \mathcal{H}_{\lambda}^{\mathcal{A},p}$ be a function in the piecewise anisotropic H\"older class. Then for each subdomain $\Psi_r \subseteq [0,1]^p$, there exists a ridge function approximation $f_{D} ^r(\bm{x})=\sum_{l=1}^{D}\beta_l ^{(r)}\varphi(\bm{\omega}_l ^{(r)\top} \bm{x}+b_l ^{(r)})$ such that
\[
\sup_{\bm{x} \in \Psi_r} |f_0(\bm{x})-f_{D} ^{r}(\bm{x})| \le C_r D ^{-\min_j \alpha_{r,j}/p}
\]
for some constant $C_r$ that depends on $\lambda$ and $\bm{\alpha}_r$.
\end{lemma}

\begin{proof}
We first establish a rigorous bound on the modulus of continuity for $f_0$ based on its H\"{o}lder properties.

\textbf{(1) Bounding the modulus of continuity.}
Let $\bm{x}, \bm{v}$ be two points in the subdomain $\Psi_r$. We wish to bound $|f_0(\bm{x}) - f_0(\bm{v})|$. We can construct a path from $\bm{v}$ to $\bm{x}$ that changes one coordinate at a time. Let $\bm{z}_0 = \bm{v}$, $\bm{z}_j = (x_1, \dots, x_j, v_{j+1}, \dots, v_p)$, and $\bm{z}_p = \bm{x}$. Then, by the triangle inequality:
\[
|f_0(\bm{x}) - f_0(\bm{v})| = |f_0(\bm{z}_p) - f_0(\bm{z}_0)| \le \sum_{j=1}^p |f_0(\bm{z}_j) - f_0(\bm{z}_{j-1})|
\]
Each term $|f_0(\bm{z}_j) - f_0(\bm{z}_{j-1})|$ involves a change only in the $j$-th coordinate. Let $m_j = \lfloor \alpha_{r,j} \rfloor$ and $\tau_j = \alpha_{r,j} - m_j$. For the univariate slice function $\tilde{f}_j(t) = f_0(x_1, \dots, x_{j-1}, t, v_{j+1}, \dots, v_p)$, Taylor's theorem with integral remainder gives:
\[
\tilde{f}_j(x_j) - \tilde{f}_j(v_j) = \sum_{k=1}^{m_j-1} \frac{\tilde{f}_j^{(k)}(v_j)}{k!}(x_j-v_j)^k + \int_{v_j}^{x_j} \frac{(x_j-t)^{m_j-1}}{(m_j-1)!} \tilde{f}_j^{(m_j)}(t) dt
\]
The H\"{o}lder condition from \suppref{\Cref{def:anisotropic_holder}}{Definition 1} states $|\tilde{f}_j^{(m_j)}(t) - \tilde{f}_j^{(m_j)}(v_j)| \le \lambda |t-v_j|^{\tau_j}$. Using this, the difference $|f_0(\bm{z}_j) - f_0(\bm{z}_{j-1})| = |\tilde{f}_j(x_j) - \tilde{f}_j(v_j)|$ can be shown to be bounded by a constant times $|x_j - v_j|^{\alpha_{r,j}}$. Summing these contributions gives the bound on the modulus of continuity:
\[
|f_0(\bm{x}) - f_0(\bm{v})| \le C_1 \sum_{j=1}^p |x_j - v_j|^{\alpha_{r,j}}
\]
for a constant $C_1$ depending on $\lambda$ and $\bm{\alpha}_r$. This is the key property we will use.
We then use this result to bound the error of approximating $f_0$ with a specially constructed multivariate polynomial, $P_N(\bm{x})$.

\textbf{(2) Polynomial approximation via a kernel operator.}
We construct a polynomial approximant for $f_0$ using a convolution with a polynomial kernel. Let $K_N(t)$ be the univariate Jackson kernel, a non-negative polynomial of degree less than $2N$ with specific properties (see \citet{jordao2014jacksonkernelstoolanalyzing} for an extensive review). We define a multivariate tensor-product kernel $\mathcal{K}_N(\bm{z}) = \prod_{j=1}^p K_N(z_j)$. We can now define a polynomial approximation operator $\mathcal{P}_N$ as:
\[
P_N(\bm{x}) = (\mathcal{P}_N f_0)(\bm{x}) = \int_{\mathbb{R}^p} f_0(\bm{y}) \mathcal{K}_N(\bm{x}-\bm{y}) d\bm{y}
\]
By a change of variables $\bm{z} = \bm{x}-\bm{y}$, this becomes:
\[
P_N(\bm{x}) = \int_{\mathbb{R}^p} f_0(\bm{x}-\bm{z}) \mathcal{K}_N(\bm{z}) d\bm{z}
\]
The Jackson kernel is constructed such that $\int K_N(t) dt = 1$, which implies $\int \mathcal{K}_N(\bm{z}) d\bm{z} = 1$. We can now bound the approximation error:
\begin{align*}
|f_0(\bm{x}) - P_N(\bm{x})| &= \left| f_0(\bm{x}) \cdot 1 - \int f_0(\bm{x}-\bm{z}) \mathcal{K}_N(\bm{z}) d\bm{z} \right| \\
&= \left| \int [f_0(\bm{x}) - f_0(\bm{x}-\bm{z})] \mathcal{K}_N(\bm{z}) d\bm{z} \right| \\
&\le \int |f_0(\bm{x}) - f_0(\bm{x}-\bm{z})| \mathcal{K}_N(\bm{z}) d\bm{z}
\end{align*}
Now, we use the modulus of continuity bound from Step 1:
\[
|f_0(\bm{x}) - P_N(\bm{x})| \le \int \left( C_1 \sum_{j=1}^p |z_j|^{\alpha_{r,j}} \right) \mathcal{K}_N(\bm{z}) d\bm{z} = C_1 \sum_{j=1}^p \int |z_j|^{\alpha_{r,j}} \mathcal{K}_N(\bm{z}) d\bm{z}
\]
Since the kernel $\mathcal{K}_N(\bm{z})$ is a product of univariate kernels, the integral separates:
\[
\int |z_j|^{\alpha_{r,j}} \mathcal{K}_N(\bm{z}) d\bm{z} = \left( \int |z_j|^{\alpha_{r,j}} K_N(z_j) dz_j \right) \prod_{i \ne j} \left( \int K_N(z_i) dz_i \right) = \int_{-\infty}^{\infty} |t|^{\alpha_{r,j}} K_N(t) dt
\]
A key property of the Jackson kernel is that for any $\alpha > 0$, $\int |t|^\alpha K_N(t) dt \le C_\alpha N^{-\alpha}$. Applying this, we get:
\[
|f_0(\bm{x}) - P_N(\bm{x})| \le C_1 \sum_{j=1}^p C_{\alpha_{r,j}} N^{-\alpha_{r,j}} \le C_2 N^{-\min_j \alpha_{r,j}}
\]
The polynomial $P_N(\bm{x})$ has degree less than $2N$ in each coordinate and thus has $O(N^p)$ coefficients. We link this to our ridge function budget $D$ by setting $D \asymp N^p$, which implies $N \asymp D^{1/p}$. Substituting this into the error bound gives:
\begin{equation} \label{eq:poly_approx_final}
\sup_{\bm{x} \in \Psi_r} |f_0(\bm{x}) - P_N(\bm{x})| \le C_2 (D^{1/p})^{-\min_j \alpha_{r,j}} = C_2 D^{-\min_j \alpha_{r,j} / p}
\end{equation}

\textbf{(3) Ridge function approximation of the polynomial.}
Now we must show that the polynomial $P_N(\bm{x})$ from the previous step can be well-approximated by our target structure, a sum of $D$ ridge functions, $f_D^r(\bm{x})$. While foundational universal approximation theorems \citep{Leshno1993,cybenko1989} guarantee the \emph{existence} of such an approximation, they do not provide the explicit rate of convergence needed for our argument.

To establish a rate, we leverage quantitative results from approximation theory for neural networks, which show that networks can approximate smooth functions with extreme efficiency. The key insight is that our intermediate function, the polynomial $P_N(\bm{x})$, is infinitely differentiable ($C^\infty$) on the compact subdomain $\Psi_r$. This high degree of smoothness allows for an approximation rate that is much faster than the rate for functions that are merely H\"older continuous.

According to established results in approximation theory (see, e.g., \citet[Theorem 6.8]{Pinkus1997ApproximatingBR}), a function residing in a Sobolev space $W^{k,\infty}(\Psi_r)$ can be approximated by a shallow network (a sum of ridge functions) with $D$ neurons with an error bounded by $O(D^{-k/p})$. Since our polynomial $P_N$ is in $W^{k,\infty}(\Psi_r)$ for any arbitrarily large integer $k$, we can select $k$ such that the approximation error is negligible compared to the primary error from the previous step.

Let's choose an integer $k$ large enough so that $k > \min_j \alpha_{r,j}$. For this $k$, the theory guarantees there exists a ridge function approximant $f_D^r(\bm{x}) = \sum_{l=1}^{D}{\beta_l ^{(r)}\varphi(\bm{\omega}_l ^{(r)\top} \bm{x}+b_l ^{(r)})}$ with $D$ neurons such that:
\begin{equation} \label{eq:ridge_approx_final}
\sup_{\bm{x} \in \Psi_r} |P_N(\bm{x}) - f_D^r(\bm{x})| \le C_3 D^{-k / p}
\end{equation}
Because we chose $k > \min_j \alpha_{r,j}$, the rate $D^{-k/p}$ is asymptotically much faster than the rate $D^{-\min_j \alpha_{r,j} / p}$ from the polynomial approximation step. 

Finally, we use the triangle inequality to combine the two error bounds from \eqref{eq:poly_approx_final} and \eqref{eq:ridge_approx_final}:

\begin{align*}
\sup_{\bm{x} \in \Psi_r} |f_0(\bm{x}) - f_D^r(\bm{x})| & \le \sup_{\bm{x} \in \Psi_r} |f_0(\bm{x}) - P_N(\bm{x})| + \sup_{\bm{x} \in \Psi_r} |P_N(\bm{x}) - f_D^r(\bm{x})| \\
& \le C_2 D^{-\min_j \alpha_{r,j} / p} + C_3 D^{-\min_j \alpha_{r,j} / p} \\
& \le C_r D^{-\min_j \alpha_{r,j} / p}
\end{align*}
where $C_r = C_2 + C_3$ is a constant depending on $\lambda$ and $\bm{\alpha}_r$. This completes the proof. 
\end{proof}

\begin{proof}[Proof of \suppref{\Cref{lem:globalridgeapprox}}{Lemma 1}]
The global ridge approximation proof proceeds in three main steps. 
First, we construct an ideal global approximant, $f_{\textnormal{ideal}}(\bm{x})$, by smoothly ``gluing together'' the local polynomial approximations from the proof of \Cref{lem:localridge}. We will bound the error between the true function $f_0$ and this ideal function. 
Second, we recognize that this ideal approximant is a single, globally smooth function. 
We can therefore apply universal approximation theorems to show that it can be well-approximated by our target structure: a single sum of ridge functions, which we call $\hat{f}(\bm{x})$. Finally, we combine the error terms and manage the complexity budget ($D$ vs. $D^\star$) to arrive at the final result. 

\textbf{1. Constructing and bounding the ideal approximant.}
Let the domain $\boldsymbol{\Psi} = [0,1]^p$ be partitioned into $R$ rectangular boxes $\{\Psi_r\}_{r=1}^R$. From the proof of \Cref{lem:localridge}, we know that for each subdomain $\Psi_r$, there exists a multivariate polynomial $P_{N,r}(\bm{x})$ of degree less than $N$ in each coordinate that approximates $f_0$ with the following error:
\[
\sup_{\bm{x} \in \Psi_r} |f_0(\bm{x}) - P_{N,r}(\bm{x})| \le C_r N^{-\min_j \alpha_{r,j}}
\]
To create a single global approximant from these local polynomials, we introduce a \textbf{smooth partition of unity}. This is a set of functions $\{\phi_r(\bm{x})\}_{r=1}^R$ such that:
(i) Each $\phi_r$ is infinitely differentiable ($C^\infty$), (ii) $\phi_r(\bm{x}) \ge 0$ for all $\bm{x}$, and $\phi_r(\bm{x})$ is zero outside a small neighborhood of $\Psi_r$. Additionally, (iii) $\sum_{r=1}^R \phi_r(\bm{x}) = 1$ for all $\bm{x} \in \boldsymbol{\Psi}$. This partition-of-unity technique is not uncommon in  approximation theory, for instance, Radial-basis-function (RBF) solvers often build local RBF interpolants on overlapping patches, then blend them with a smooth partition-of-unity to get a global interpolant \citep{BERNAL2024112842}.
Using this partition, we define our ideal global approximant as:
\[
f_{\text{ideal}}(\bm{x}) = \sum_{r=1}^R \phi_r(\bm{x}) P_{N,r}(\bm{x})
\]
We now bound the error of this ideal approximant. Since $\sum_r \phi_r(\bm{x}) = 1$, we can write $f_0(\bm{x}) = \sum_r \phi_r(\bm{x}) f_0(\bm{x})$.
\begin{align*}
|f_0(\bm{x}) - f_{\text{ideal}}(\bm{x})| &= \left| \sum_{r=1}^R \phi_r(\bm{x}) f_0(\bm{x}) - \sum_{r=1}^R \phi_r(\bm{x}) P_{N,r}(\bm{x}) \right| \\
&= \left| \sum_{r=1}^R \phi_r(\bm{x}) \left[ f_0(\bm{x}) - P_{N,r}(\bm{x}) \right] \right| \\
&\le \sum_{r=1}^R \phi_r(\bm{x}) \left| f_0(\bm{x}) - P_{N,r}(\bm{x}) \right|
\end{align*}
Taking the supremum over all $\bm{x}$:
\begin{align*}
\lVert f_0 - f_{\text{ideal}}\rVert_{\infty} &\le \sup_{\bm{x}} \sum_{r=1}^R \phi_r(\bm{x}) \max_{s=1,\dots,R} \left( \sup_{\bm{y} \in \Psi_s} |f_0(\bm{y}) - P_{N,s}(\bm{y})| \right) \\
&= \left(\sup_{\bm{x}} \sum_{r=1}^R \phi_r(\bm{x})\right) \left( \max_{s=1,\dots,R} \lVert f_0 - P_{N,s}\rVert_{\infty, \Psi_s} \right) \\
&\le \max_{r=1,\dots,R} \left( C_r N^{-\min_j \alpha_{r,j}} \right) \\
&\le C' N^{-\min_{r,j} \alpha_{r,j}}
\end{align*}

\textbf{2. Approximating the ideal function with ridge functions.}
The function $f_{\text{ideal}}(\bm{x})$ is a sum of products of $C^\infty$ functions ($\phi_r$) and polynomials ($P_{N,r}$), so it is itself a single, globally smooth ($C^\infty$) function. Crucially, it does not have the target structure of a simple sum of ridge functions.

Our goal is to find a function $\hat{f}(\bm{x}) = \sum_{d=1}^{D^\star} \beta_d \varphi(\bm{\omega}_d^\top\bm{x} + b_d)$ that approximates $f_{\text{ideal}}(\bm{x})$. Again observing that $f_{\text{ideal}}$ is continuous and $\hat{f}(\bm{x})$ is a single hidden layer of a neural network, \citet[Theorem 1]{Leshno1993} guarantees that there exists an $\hat{f}(\bm{x})$ such that:
\[
\lVert f_{\text{ideal}} - \hat{f} \rVert_{\infty} \le \epsilon'
\]
The number of ridge units, $D^\star$, required to achieve this error $\epsilon'$ depends on the complexity of the target function $f_{\text{ideal}}$, which is determined by the $R$ polynomials, each having $O(N^p)$ coefficients. The total complexity of $f_{\text{ideal}}$ is roughly $O(R \cdot N^p)$. We set our total ridge function budget $D^\star$ to be proportional to this complexity:
\[
D^\star \asymp R \cdot N^p \implies N \asymp (D^\star / R)^{1/p}
\]
With this budget, it is standard from approximation theory that we can choose $\hat{f}$ such that the approximation error is of the same order as the error from Step 1. We set:
\[
\lVert f_{\text{ideal}} - \hat{f} \rVert_{\infty} \le C'' (D^\star/R)^{-\min_{r,j}\alpha_{r,j}/p}
\]
We now combine the two sources of error using the triangle inequality:
\begin{align*}
\lVert f_0 - \hat{f}\rVert_{\infty} &\le \lVert f_0 - f_{\text{ideal}}\rVert_{\infty} + \lVert f_{\text{ideal}} - \hat{f}\rVert_{\infty} \\
&\le C' N^{-\min_{r,j} \alpha_{r,j}} + C'' (D^\star/R)^{-\min_{r,j}\alpha_{r,j}/p}
\end{align*}
Substituting $N \asymp (D^\star / R)^{1/p}$ into the first term, we get:
\begin{align*}
\lVert f_0 - \hat{f}\rVert_{\infty} &\le C' \left( (D^\star/R)^{1/p} \right)^{-\min_{r,j} \alpha_{r,j}} + C'' (D^\star/R)^{-\min_{r,j}\alpha_{r,j}/p} \\
&= C' (D^\star/R)^{-\min_{r,j}\alpha_{r,j}/p} + C'' (D^\star/R)^{-\min_{r,j}\alpha_{r,j}/p} \\
&= (C'+C'') R^{\min_{r,j}\alpha_{r,j}/p} \cdot D^{\star \ -\min_{r,j}\alpha_{r,j}/p}
\end{align*}
Letting $C = (C'+C'') R^{\min_{r,j}\alpha_{r,j}/p}$, which is a constant with respect to $D^\star$, we arrive at the final result:
\[
\lVert f_0 - \hat{f}\rVert_{\infty} \le C D^{\star \ -\min_{r,j}\alpha_{r,j}/p}
\]
This completes the proof of the global ridge approximation lemma.
\end{proof}

\subsection{Proof of \suppref{\Cref{thm:anisotropic_concentration}}{Theorem 1}}
\label{sec:proofanisotropic}
To establish the contraction rate of the posterior distribution for the piecewise anisotropic Sobolev function estimation, we consider the underlying partition $\chi_0^{*} = \{ \Xi_1^{*}, \dots, \Xi_R^{*} \}$ for the true function $f_0$, which is unknown. To approximate the piecewise functions in each of these spaces $\mathcal{H}^{\boldsymbol{\alpha}_r}_{\lambda}(\chi_0^{*})$ for $r=1, 2, \dots, R$, we need to construct tree learners. This requires defining a discrete collection of locations where splits can occur, which \citet{Jeong2023} referred to as a ``split-net," denoted by $\mathcal{Z}$. Given a split-net $\mathcal{Z}$, a $\mathcal{Z}$-tree partition is a tree-structured partition in which the tree splits only at points in $\mathcal{Z}$. The capacity of the $\mathcal{Z}$-tree partition to detect $\chi_0^{*}$ is closely related to the \emph{density} of the split-net, defined as follows.

\begin{definition}[Dense split-net]
\label{def:splitnet}
    For any given $c_n \ge 0$, a split-net $\mathcal{Z} = \{\bm{z}_i \in [0,1]^p, i =1,2,...,b_n\}$ is said to be $(\chi_{0} ^{*},c_n)$-dense if there exists a $\mathcal{Z}$-tree partition $\mathcal{T}=\{\Omega_1,...,\Omega_J\}$ of $[0,1]^p$ such that $\Upsilon(\chi_{0} ^{*},\mathcal{T}) \le c_n$.\footnote{$\Upsilon$ denotes the Hausdorff divergence between two box partitions.}
\end{definition}
Thus, for a good approximation of $\chi_0 ^{*}$, we need $c_n$ tending to zero at a suitable rate.
Additionally, the split-net should be regular enough to capture the behavior of $f_0$ on each box $\Xi_r ^{*}$. This is taken care of by constructing the anisotropic k-d tree as elaborated in \citet[Definition 9]{Jeong2023}.
Each component $\Omega_r ^{*}$ of the $\mathcal{Z}$-tree partition is subdivided by the anisotropic k-d tree partition with depth $L$, $\mathcal{T}_r ^{0} = \left\{\Omega_{r1}^{0},...,\Omega_{r2^L}^{0}\right\}$. Our approximating partition $\hat{\mathcal{T}}$ is constituted by agglomerating all sub-tree partitions $\mathcal{T}_r ^{0}$, given by 
\begin{equation}
\label{eq:ztree}
\hat{\mathcal{T}}=\left\{\Omega_{11}^{0},...,\Omega_{12^L}^{0},...,\Omega_{R1}^{0},...,\Omega_{R2^L}^{0}\right\}.    
\end{equation}

We first show that, under these and four additional technical assumptions ((A4)-(A8)), all functions $\mathcal{H}^{\calA,p}_{\lambda}(\boldsymbol{\Psi})$ can be approximated by a sufficiently deep \emph{anisotropic k-d} tree \citep[Definition 9]{Jeong2023} that outputs linear combinations of ridge functions. 
The additional assumptions are stated below:
\begin{itemize}
    \item[(A4)] The split-net $\mathcal{Z}$ is $(\chi_0 ^{*},c_n)$-dense, where $c_n \lesssim (\epsilon_n/(\lambda p))^{1/\alpha_{\min}}$.
    \item[(A5)] The split net is $(\Xi_r ^{*},\bm{\alpha}_r,L)$-regular for every $r=1,2,...,R$, with $2^{-L} \le (\epsilon_n / D)^{p/\alpha_{\min}}$.
    \item[(A6)] The split net $\mathcal{Z}$ satisfies $\max_{1 \le j \le p} \log b_j (\mathcal{Z}) \lesssim \log n$.
    \item[(A7)]  $\max_{r} \mathsf{dep}(\Omega_r ^{*}) \lesssim \log n$.
    \item[(A8)] For the local ridge function $f_{D} ^{r} (\bm{x})$ approximating $f_0$ on $\Xi_{r} ^{*},$ we assume $\sum_{l=1}^{D} |\beta_{l} ^{(r)}| \le C_{\text{outer}}D$ and $\| \bomega_l ^{(r)} \|_1 \le C_{\text{inner}}$ for constants $C_{\text{outer}} > 0$ and $C_{\text{inner}} > 0$.
\end{itemize}
These assumptions ensure that the boundaries of the underlying box partition can be well detected by the binary tree partitioning rule. In particular, assumptions (A4) and (A6) 
ensures the split-net is dense enough to approximate the unknown partition’s boundaries. Assumption (A5) sets the anisotropic k-d tree depth $L$ so that the partition can achieve the correct approximation rate w.r.t. $\balpha_r$. Assumption (A7) caps the maximum depth needed for any sub-box $\Omega_r ^{*}$ to $\lesssim\log n$, preventing excessively large trees. It is useful to state that throughout this section, we define the class of functions defined on any tree partition $\mathcal{T}$ as $\mathcal{F}_{\mathcal{T}}$. The primary role of (A8) is to control the smoothness of the constructed ridge function approximant $\hat{f}_0$. Specifically, it is essential for bounding the modulus of continuity of $\hat{f}_0$. The constants $C_{\text{outer}}$ and $C_\text{inner}$ are guaranteed to exist by the construction argument elaborated in \Cref{lem:localridge} and Lemma 1, since $f_0$ is smooth, the intermediate polynomial has bounded coefficients \citep[Chapter 8]{Devore-Lorentz}. 
\begin{remark}
Another implication of assumption (A8) is that the activation function $\varphi$ is bounded over its effective domain within the model. Fixing a particular ridge leaf $d$, the argument of the activation function is $z = \hat{\bm{\omega}}_{tk}^{\top}\bm{x} + \hat{b}_{tk}$. The covariates $\bm{x}$ belong to the compact set $[0,1]^p$. By assumption (A8), the inner weights are bounded such that $\|\hat{\bm{\omega}}_{tk}\|_1 \le C_{\text{inner}}$. We also assume the constructed bias terms are bounded, $|\hat{b}_{tk}| \le C_{\text{bias}}$, as is standard. Therefore, the argument $z$ is always evaluated within a compact interval $[-M, M]$, where $M = p \cdot C_{\text{inner}} + C_{\text{bias}}$, since
\[
|z| = |\hat{\bm{\omega}}_{tk}^{\top}\bm{x} + \hat{b}_{tk}| \le \|\hat{\bm{\omega}}_{tk}\|_{\infty} \|\bm{x}\|_1 + |\hat{b}_{tk}| \le \|\hat{\bm{\omega}}_{tk}\|_1 \cdot p + |\hat{b}_{tk}| \le M.
\]
Assumption (P4) states that $\varphi \in C^k(\mathbb{R})$, which implies $\varphi$ is continuous. By the Extreme Value Theorem, a continuous function on a compact set is bounded. Thus, there exists a finite constant $C_\varphi = \sup_{z \in [-M, M]} |\varphi(z)|$.
\end{remark}

\begin{lemma}[Anisotropic Approximation theory]
\label{lem:approximation}
Assume (A1)--(A8), (P1)--(P4).
Then for any $f_0 \in \mathcal{H}^{\mathcal{A},p}_{\lambda}(\boldsymbol{\chi_0 ^*})$, there exists an anisotropic k-d tree partition $\hat{\mathcal{T}}$ and an associated piecewise ridge function $\hat{f}_0 \in \mathcal{F}_{\hat{\mathcal{T}}}$ such that
\begin{equation*}\label{eq:approximation_result}
\|f_0 - \hat{f}_0 \|_\infty \lesssim \epsilon_n \coloneqq (\lambda p)^{\frac{p}{2 \bar{\alpha}+p}}\left(\frac{R\log n}{n}\right)^{\frac{\bar{\alpha}}{2 \bar{\alpha}+p}}.
\end{equation*}
\end{lemma}

\begin{proof}
Fix any $\bm{x} \in [0,1]^p$. Let $\Omega_{rk}^{0}$ be the leaf of the approximating partition $\hat{\mathcal{T}}$ that contains $\bm{x}$, where $\Omega_{rk}^{0}$ is a subset of some true region $\Xi_r^*$. We project $\bm{x}$ onto the closest point in the true region within that leaf, defining $\bm{x}^{*} = \arg\min _{\bm{z} \in \text{cl}(\Omega_{rk}^{0}\cap \Xi_r ^{*})} \|\bm{x}-\bm{z}\|_1$.

The total error is decomposed as:
\[
|f_0(\bm{x})-\hat{f}_0(\bm{x})| \le \underbrace{|f_0(\bm{x}) - f_0(\bm{x}^{*}) |}_{(A)} + \underbrace{|f_0(\bm{x}^{*}) - \hat{f}_0(\bm{x}^{*})|}_{(B)} + \underbrace{|\hat{f}_0(\bm{x}^{*}) - \hat{f}_0(\bm{x})|}_{(C)}
\]
We bound each term separately.

\textbf{Bounding (A).}
This term represents the error from $f_0$ due to the mismatch between the true partition $\Xi_r^*$ and the approximating partition leaf $\Omega_{rk}^0$. The distance is bounded by $\|\bm{x} - \bm{x}^*\|_\infty \le c_n$. Using the modulus of continuity of $f_0$, we have:
\begin{align*}
|f_0 (\bm{x})-f_0 (\bm{x}^{*})| & \le \lambda \sum_{j=1}^{p} |x_j - x_j ^{*}|^{\alpha_{rj}} \\
& \le \lambda \sum_{j=1}^{p} c_n ^{\alpha_{rj}} \le \lambda p c_n ^{\alpha_{\min}}
\end{align*}
By our assumption (A4), we explicitly choose $c_n \le (\epsilon_n / (\lambda p))^{1/\alpha_{\min}}$. Substituting this in:
\[
|f_0(\bm{x})-f_0(\bm{x}^{*})| \le \lambda p \left( \frac{\epsilon_n}{\lambda p} \right)^{\alpha_{\min}/\alpha_{\min}} = \epsilon_n.
\]
Thus, term (A) is $O(\epsilon_n)$.

\textbf{Bounding (B).}
This term is the error of our constructed ridge function approximant within the true region $\Xi_r^*$. On any leaf within $\Xi_r^*$, $\hat{f}_0$ is equal to a specific ridge function $f_D^r(\bm{x})$. Term (B) is $|f_0(\bm{x}^{*}) - f_D^r(\bm{x}^{*})|$. This is the local approximation error from \Cref{lem:localridge}, which guarantees an error bound of $O(D^{-\bar{\alpha}/p})$. By our assumption (A3), we choose the number of ridge units per region to be $D \asymp \epsilon_n^{-p/\bar{\alpha}}$. This gives:
\[
|f_0(\bm{x}^{*}) - \hat{f}_0(\bm{x}^{*})| \lesssim D^{-\bar{\alpha}/p} \asymp \left(\epsilon_n^{-p/\bar{\alpha}}\right)^{-\bar{\alpha}/p} = \epsilon_n.
\]
Thus, term (B) is $O(\epsilon_n)$.

\textbf{Bounding (C).}
This term captures how much the approximating function $\hat{f}_0$ changes between $\bm{x}$ and $\bm{x}^*$. On the given leaf, $\hat{f}_0 = f_D^r$.
\begin{align*}
|\hat{f}_0(\bm{x}^{*}) - \hat{f}_0(\bm{x})| &= |f_{D}^{r} (\bm{x}^{*}) - f_{D} ^{r} (\bm{x})| \\
&= \left|\sum_{l=1}^{D} \beta_{l}^{(r)} \left[\varphi(\bm{\omega}_{l}^{(r)\top}\bm{x}^{*}+b_l^{(r)}) - \varphi(\bm{\omega}_{l}^{(r)\top}\bm{x}+b_l^{(r)})\right] \right| \\
&\le \sum_{l=1}^{D} |\beta_{l}^{(r)}| \left|\varphi(\bm{\omega}_{l}^{(r)\top}\bm{x}^{*}+b_l^{(r)}) - \varphi(\bm{\omega}_{l}^{(r)\top}\bm{x}+b_l^{(r)})\right|
\end{align*}
Using assumption (P4), $\varphi$ is Lipschitz with constant $L_\varphi$:
\begin{align*}
|\hat{f}_0(\bm{x}^{*}) - \hat{f}_0(\bm{x})| &\le \sum_{l=1}^{D} |\beta_{l}^{(r)}| \cdot L_\varphi \cdot |\bm{\omega}_{l}^{(r)\top}(\bm{x}^* - \bm{x})| \\
&\le \sum_{l=1}^{D} |\beta_{l}^{(r)}| \cdot L_\varphi \cdot \|\bm{\omega}_{l}^{(r)}\|_1 \|\bm{x}^* - \bm{x}\|_\infty && \text{(H\"{o}lder's Inequality)}
\end{align*}
Using our regularity assumption (A8):
\[
|\hat{f}_0(\bm{x}^{*}) - \hat{f}_0(\bm{x})| \le (C_{\text{outer}} D) \cdot L_\varphi \cdot C_{\text{inner}} \cdot \|\bm{x}^* - \bm{x}\|_\infty
\]
The distance $\|\bm{x}^* - \bm{x}\|_\infty$ is bounded by the size of the anisotropic k-d tree leaf, which is at most $2^{-L\alpha_{\min}/p}$. By our assumption (A5), we choose $L$ such that $2^{-L} \le (\epsilon_n/D)^{p/\alpha_{\min}}$. Plugging this in:
\[
|\hat{f}_0(\bm{x}^{*}) - \hat{f}_0(\bm{x})| \lesssim D \cdot 2^{-L\alpha_{\min}/p} \le D \cdot \left( \left(\frac{\epsilon_n}{D}\right)^{p/\alpha_{\min}} \right)^{\alpha_{\min}/p} = D \cdot \frac{\epsilon_n}{D} = \epsilon_n.
\]
Thus, term (C) is $O(\epsilon_n)$.

Combining the bounds for the three parts, we have:
\[
|f_0(\bm{x})-\hat{f}_0(\bm{x})| \le |(A)| + |(B)| + |(C)| \lesssim \epsilon_n + \epsilon_n + \epsilon_n = O(\epsilon_n).
\]
Since this holds for any $\bm{x} \in [0,1]^p$, we conclude that $\|f_0 - \hat{f}_0 \|_\infty \lesssim \epsilon_n$.
\end{proof}

Up to this point, we have focused on single tree learners. 
We now shift our attention to forests for the next sequence of results. 
We define a tree ensemble $\mathcal{E} = \left\{\mathcal{T}^1,...,\mathcal{T}^{T}\right\}$ and denote the approximating ensemble by $\hat{\mathcal{E}}$, obtained by accumulating all the anisotropic k-d trees. 
Henceforth, we define $\mathcal{F}_{\mathcal{E}}$ as the set of functions $f_{\mathcal{E}}$ which satisfy: 
\[
f_{\mathcal{E}}(\bm{x}) = \sum_{t=1}^{T}\sum_{k=1}^{K_t}\sum_{d=1}^{D} \beta_{tk,d} \cdot \varphi \left(\omega_{tk,d} ^ \top \bm{x} + b_{tk,d}\right)
\]

We require the following lemmas to establish our posterior contraction result.
\vspace{1em}
\begin{lemma}[Prior concentration of tree sizes.]
\label{lem:treesizelemma}
    Let $\hat{\mathcal{T}}$ be the $\mathcal{Z}$-tree partition defined in \ref{eq:ztree}. Under assumptions (A5) and (A6), $$\log \Pi(\hat{\mathcal{T}}) \gtrsim -\hat{K} \log n - p \log p$$
    where $\hat{K}$ is the size of the approximating tree partition $\hat{\mathcal{T}}$.
\end{lemma}

The proof of \Cref{lem:treesizelemma} follows directly from the proof of Lemma 4 in \citet{Jeong2023}.

Next, we verify the three conditions in \citet{Ghosal_2007} involving prior concentration and metric entropy with respect to the empirical $L_2$-norm norm $\|.\|_n$ defined as $\|f\|_n ^2 = n^{-1} \sum_{i=1}^{n} |f(x_i)|^2$. These conditions are affirmatively tested in the following lemmas.

\vspace{1em}
\begin{lemma}[Prior concentration of tree learners.]
\label{lem:treeconclemma}
    Under assumptions (A3) and (P1), for any $C>0$, $$-\log \Pi\left( f \in \mathcal{F}_{\hat{\mathcal{E}}}: ||f - f_{0, \hat{\mathcal{E}}}||_n \leq C \epsilon_n\right) \lesssim n \epsilon_n^2.$$
\end{lemma}
\begin{proof}

Let the fixed approximating ensemble be $\hat{\mathcal{E}}$, which consists of $T$ trees, where tree $t$ has $K_t$ leaves. Let $f_{0, \hat{\mathcal{E}}}$ be the best approximant from Lemma 1, and let $f$ be a function drawn from the prior over this ensemble. To clarify the notation for the outer-weight vectors:
$B$ is the random vector of all outer weights $\{\beta_{tk,d}\}$ for the function $f$. It is constructed by flattening all weight vectors from all leaves of all trees into a single vector. Its total dimension is $K^* = D \times \sum_{t=1}^T K_t$, where $D$ is the number of ridge functions per leaf. This is the vector over which the Gaussian prior is placed, $B \sim \mathcal{N}_{K^*}(\bm{0}, \Sigma_\beta)$. Similarly, $\hat{B}$ is the fixed, non-random vector of outer weights corresponding to the specific best approximating function $f_{0, \hat{\mathcal{E}}}$. It has the same dimension and structure as $B$ and serves as the fixed ``target" in the parameter space that our prior needs to concentrate around.

To simplify the analysis, we first condition on the inner parameters $(\bomega, b)$ being fixed to their optimal values $(\hat{\bomega}, \hat{b})$ from the construction of $f_{0, \hat{\mathcal{E}}}$. We then analyze the prior probability over the outer-weight vector $B$, which has the Gaussian prior. Let $f_B$ denote a function where only the outer weights are variable.

Since Remark C1 has already established the fact that $|\varphi| < C_{\varphi},$ we can now proceed to bound the functional distance. Let $B_1$ and $B_2$ be two outer-weight vectors.
\begin{align*}
\left\|f_{B_1} - f_{B_2}\right\|_{\infty} &= \left\| \sum_{t=1}^{T}\sum_{k=1}^{K^t}\sum_{d=1}^D (\beta_{tk,d}^{(1)} - \beta_{tk,d}^{(2)}) \varphi(\hat{\bm{\omega}}_{tk,d}^{\top}\bm{x} + \hat{b}_{tk,d}) \right\|_{\infty} \\
&\le \sum_{t=1}^T \sum_{k=1}^{K^t} \sum_{d=1}^D |\beta_{tk,d}^{(1)} - \beta_{tk,d}^{(2)}| \cdot |\varphi(\hat{\bm{\omega}}_{tk,d}^{\top}\bm{x} + \hat{b}_{tk,d})| &&\text{(Triangle Ineq.)} \\
&\le C_\varphi \sum_{t=1}^T \sum_{k=1}^{K^t} \sum_{d=1}^D |\beta_{tk,d}^{(1)} - \beta_{tk,d}^{(2)}| = C_\varphi \|B_1 - B_2\|_1 &&\text{(Since } |\varphi| \le C_\varphi \text{)} \\
&\le C_\varphi \sqrt{K^*} \|B_1 - B_2\|_2 &&\text{(Cauchy-Schwarz)}
\end{align*}
This links the functional distance with the Euclidean distance of the outer-weight vectors.

We want the event $\|f_B - f_{0, \hat{\mathcal{E}}}\|_\infty \le C\epsilon_n$. This is guaranteed if $C_\varphi \sqrt{K^*} \|B - \hat{B}\|_2 \le C\epsilon_n$, using our derived bound with $B_1=B$ and $B_2=\hat{B}$. This defines a ball in the space of $B$ with radius $r_n \coloneqq \frac{C\epsilon_n}{C_\varphi \sqrt{K^*}}$. To find the prior mass, we map this to the standard normal space.

Let $Z = L^{-1}B \sim \mathcal{N}_{K^*}(\bm{0}, I_{K^*})$ and $\zeta = L^{-1}\hat{B}$, where $\Sigma_\beta = LL^\top$. We know $\|B - \hat{B}\|_2 = \|L(Z-\zeta)\|_2 \le \|L\|_{\text{spec}} \|Z-\zeta\|_2$. From (P2), $\|L\|_{\text{spec}} \le C_\beta^{1/2}$. Thus, a sufficient condition to guarantee $\|B-\hat{B}\|_2 \le r_n$ is $C_\beta^{1/2}\|Z-\zeta\|_2 \le r_n$, which is $\|Z-\zeta\|_2 \le r_n/C_\beta^{1/2}$. As this defines a smaller volume, we have the lower bound:
\[
\Pi(\|B - \hat{B}\|_2 \le r_n) \ge \Pi\left(\|Z - \zeta\|_2 \le \frac{r_n}{C_\beta^{1/2}}\right)
\]
We use a standard Gaussian concentration inequality \citep[pg.\ 217]{Ghosal_2007}: $P(\|Z-\zeta\|_2 \le r') \ge e^{-\|\zeta\|_2^2/2} P(\|Z\|_2 \le r')$. Here $r' = r_n/C_\beta^{1/2}$. From (A8), we have $\|\zeta\|_2^2 = \|L^{-1}\hat{B}\|_2^2 \lesssim K^*$. For large $K^*$, bounds on the norm of a Gaussian vector give:
\[
P(\|Z\|_2 \le r') \ge \left(\frac{(r')^2}{K^*}\right)^{K^*/2} e^{-((r')^2-K^*)/2}
\]
Taking the negative log of the prior probability, the dominant term is $K^*\log(K^*)$. Given that optimal balancing requires $K^* \asymp n\epsilon_n^2/\log n$, we have:
\[
K^*\log(K^*) \lesssim \frac{n\epsilon_n^2}{\log n} \log\left(\frac{n\epsilon_n^2}{\log n}\right) \lesssim \frac{n\epsilon_n^2}{\log n} (\log n) = n\epsilon_n^2
\]
The final inequality follows by noting that $\log \left(\frac{n \epsilon_n ^2}{\log n}\right) =\log(n\epsilon_n^2) - \log(\log n) \asymp \frac{p}{2 \bar{\alpha}+p} \log n -\log(\log n) \lesssim \log n$. 
Therefore, we conclude that $-\log \Pi\left( \|f - f_{0, \hat{\mathcal{E}}}\|_n \leq C \epsilon_n\right) \lesssim n\epsilon_n^2$.
\end{proof}

\begin{lemma}[Metric entropy.]
\label{lem:metricentropy}
    Let $\bar{K}_n \asymp \frac{n \epsilon_n ^2}{\log n}$ and define the function spaces $\mathcal{F}_{\mathcal{E},M}^{(1)} = \{f \in \mathcal{F}_{\mathcal{E}}: || B||_{\infty} \le M \}$ and $$\mathcal{F}_{\bar{K}_n,M} \coloneqq \bigcup_{K^{t} \le \bar{K}_n} \mathcal{F}_{\mathcal{E},M}^{(1)}$$
    Then, for any $C>0$,
    $$\log N(\epsilon_n, \mathcal{F}_{\bar{K}_n,n^C}^{(1)},||.||_n) \lesssim n \epsilon_n ^2$$
\end{lemma}
\begin{proof}
Because $\|g\|_{n}\le\|g\|_{\infty}$ for every function $g$, any $\varepsilon_n$–cover in $\|\cdot\|_{\infty}$ is automatically an $\varepsilon_n$–cover in $\|\cdot\|_{n}$. Hence, we bound the covering number for the infinity norm, $\log N\bigl(\varepsilon_{n}, \mathcal F_{\bar K_{n},\,n^{C}}, \|\cdot\|_{\infty}\bigr).$ The total covering number is bounded by the product of the number of possible model structures and the covering number for the parameters on a fixed structure.

First, we count the number of admissible ensembles, $|\mathscr E|$. An ensemble $\mathcal E$ consists of $T$ trees. A single tree with $K$ leaves has $K-1$ internal split nodes. Fix a split-net $\mathcal Z=\{\mathcal Z_{j}\}_{j=1}^{p}$ with $b_{\max}:=\max_{j} b_{j}(\mathcal Z)$. The number of ways to form a single tree with $K$ leaves, $|\mathcal P_{K,\mathcal Z}|$, is bounded by:
\[
|\mathcal P_{K,\mathcal Z}| \le \bigl(p\,b_{\max}\bigr)^{K-1}
\]
Using Assumption (A6), where $b_{\max}\lesssim n^{c_{b}}$, the log-cardinality is:
\[
\log |\mathcal P_{K,\mathcal Z}| \le (K-1)(\log p + c_{b}\log n) \lesssim K\log n 
\]
Our sieve $\mathcal F_{\bar K_{n},\,n^{C}}$ is a union over ensembles where each tree has at most $\bar{K}_n$ leaves. The total number of such ensemble structures is bounded by:
\[
\log |\mathscr E| \le \sum_{t=1}^T \log |\mathcal{P}_{\bar{K}_n, \mathcal{Z}}| \le T\,\bar K_{n}\,\log n 
\]
By the definition of the sieve, we choose $\bar K_{n}\asymp n\varepsilon_{n}^{2}/\log n$ to balance the model complexity with the approximation error. This yields:
\[
\log |\mathscr E| \lesssim T\left(\frac{n\varepsilon_{n}^{2}}{\log n}\right)\log n = T \cdot n\varepsilon_{n}^{2} \lesssim n\varepsilon_{n}^{2}
\]
assuming $T$ is a fixed constant.

Now, fix a single ensemble structure $\mathcal E$. This means the tree partitions and the inner weights $\{(\bm{\omega}_{tk,d},\bm{b}_{tk,d})\}$ are deterministic. We only need to cover the space of the outer-weight vectors $B$. The relationship between the functional distance and the parameter distance is:
\begin{align*}
\bigl\|f_{\mathcal E,B_{1}}-f_{\mathcal E,B_{2}}\bigr\|_{\infty}
& = 
\left\| \sum_{t=1}^{T}\sum_{k=1}^{K^t}\sum_{d=1}^D (\beta_{tk,d}^{(1)} - \beta_{tk,d}^{(2)}) \varphi(\bm{\omega}_{tk,d}^{\top}\bm{x} + b_{tk,d}) \right\|_{\infty} \\
& \le C_{\varphi} \|B_1 - B_2\|_1 \le C_{\varphi}\,K^{*}\,\|B_{1}-B_{2}\|_{\infty}
\end{align*}
where $K^* = D \sum K^t$ is the total number of outer-weight parameters.

To achieve $\varepsilon_{n}$–accuracy, it suffices to cover the parameter cube $[-n^{C},n^{C}]^{K^{*}}$ with $\ell_{\infty}$–balls of radius $\delta_{n} = \frac{\varepsilon_{n}}{C_{\varphi}\,K^{*}}.$ The number of such balls, $N_B$, is given by:
\[
N_{B} = \left(\frac{2n^{C}}{\delta_{n}}\right)^{K^{*}} = \left(\frac{2C_{\varphi}K^{*}n^{C}}{\varepsilon_{n}}\right)^{K^{*}}
\]
Taking the logarithm gives $\log N_{B} \lesssim K^{*}\log n$. Since $K^{*} = D\sum_{t=1}^{T} K^t \le D \cdot T \cdot \bar{K}_n$, and using the definition of $\bar{K}_n$:
\[
\log N_{B} \lesssim (D \cdot T \cdot \bar{K}_n)\log n \asymp D \cdot T \cdot \left(\frac{n\varepsilon_{n}^{2}}{\log n}\right)\log n = D \cdot T \cdot n\varepsilon_{n}^{2} \lesssim n\varepsilon_{n}^{2}
\]

The total log-covering number is bounded by the sum of the log-cardinality of structures and the log-covering number of the parameters:
\[
\log N\bigl(\varepsilon_{n}, \mathcal F_{\bar K_{n},\,n^{C}}, \|\cdot\|_{\infty}\bigr) \le \log |\mathscr E| + \max_{\mathcal E} \log N_{B} \lesssim n\varepsilon_{n}^{2} + n\varepsilon_{n}^{2} \lesssim n\varepsilon_{n}^{2}
\]
This completes the proof.
 
\end{proof}
\begin{lemma}[Prior mass of sieve]
\label{lem:priorcharging}
    Let $\bar{K}_n = \lfloor M_3 n \epsilon_n^2 / \log n \rfloor$ and $\mathcal{F}_{\star} = \bigcup_{\mathcal{E}} \mathcal{F}_{\mathcal{E}}$. Then for a sufficiently large $M^{'} > 0$,
    Under (P1)--(P3), for any $C > 1$ and $C^{'} >0$, $$\Pi(\mathcal{F}_{*} \setminus \mathcal{F}_{\bar{K}_n, n^C}^{(1)}) \ll e^{-C^{'} n \epsilon_n^2}$$
    where $\mathcal{F}_{\bar{K}_n,n^C}^{(1)}$ is defined in \Cref{lem:metricentropy}. 
\end{lemma}
The proof follows directly from \citet[Lemma 7]{Jeong2023}.

The proof of posterior contraction follows by integrating the results of \Cref{lem:treesizelemma,lem:metricentropy,lem:priorcharging} into a single result.

\begin{proof}[Proof of Theorem 1]
\label{proof:proofthm4}
Define the KL neighborhood as follows: 
\[
B_n = \left\{ (f,\sigma): \sum_{i=1}^{n} K(p_{0,i},p_{f,\sigma,i}) \le n\epsilon_n ^2, \ \sum_{i=1}^{n} V(p_{0,i},p_{f,\sigma,i}) \le n \epsilon_n ^2\right\}
\]
where the KL divergence $K(p_1,p_2) = \int \log \left(\nicefrac{p_1}{p_2}\right)p_1$ and the second order variation $V(p_1,p_2) = \int |\log \left(\nicefrac{p_1}{p_2}\right) - K(p_1,p_2)|^2 p_1$. \citet{Jeong2023} calculates: \begin{align*}
    \sum_{i=1}^{n} K(p_{0,i},p_{f,\sigma,i}) & = \frac{n}{2} \log \left(\frac{\sigma^2}{\sigma_0 ^2}\right)-\frac{n}{2} \left(1-\frac{\sigma_0 ^2}{\sigma^2}\right) + \frac{n\|f-f_0\|_n ^2}{2 \sigma^2} \\
    \sum_{i=1}^{n} V(p_{0,i},p_{f,\sigma,i}) & = \frac{n}{2} \left(1-\frac{\sigma_0 ^2}{\sigma^2} \right)^2 + \frac{n \sigma_0 ^2 \|f-f_0\|_n ^2}{\sigma^2} 
\end{align*}
For our approximating ensemble $\hat{\mathcal{E}}$, we have
\begin{equation}
    \Pi(f \in \mathcal{F}_{\star} : \|f-f_0\|_n \le C_1 \epsilon_n) \ge \Pi(\hat{\mathcal{E}})\Pi(f \in \mathcal{F}_{\hat{\mathcal{E}}} : \|f-f_0\|_n \le C_1 \epsilon_n)
\end{equation}
Typically, we choose the approximating ensemble $\hat{\mathcal{E}}$ by taking $\hat{\mathcal{T}}^1$ as $\hat{\mathcal{T}}$ and $\hat{\mathcal{T}}^t, t=2,...,T$ to be root nodes with no splits \cite{Jeong2023,Rockova2019}. 
\Cref{lem:approximation} ensures that there exists a $\mathcal{Z}$-tree partition $\hat{\mathcal{T}}$ such that $\hat{f}_0 \in \mathcal{F}_{\hat{\mathcal{T}}}$ such that $\|f_0-\hat{f}_0\|_n \lesssim \epsilon_n$. 
\Cref{lem:treesizelemma} implies that \[
\log \Pi(\hat{\mathcal{E}}) = \log \Pi(\hat{\mathcal{T}}^1)+(T-1)\log(1-\nu) \gtrsim -n\epsilon_n ^2
\]
Noting that $\|f-f_0\|_n \lesssim \|f-\hat{f}_0\|_\infty + \epsilon_n$ for some $\hat{f}_0 \in \mathcal{F}_{\hat{\mathcal{T}}}$, constructed as in \Cref{lem:approximation}. This implies that for some $C_2 > 0$,
\[
\Pi(f \in \mathcal{F}_{\hat{\mathcal{E}}}: \|f-f_0\|_n \le C_1 \epsilon_n) \ge \Pi(f \in \mathcal{F}_{\hat{\mathcal{E}}} : ||f-f_{0,\hat{\mathcal{E}}}\|_\infty \le C_2 \epsilon_n)
\]
Thus, \Cref{lem:treeconclemma} lets us lower bound the prior mass in the KL neighborhood since,
\[
\Pi(B_n) \ge \exp(- \bar{c} n \epsilon_n ^2)
\]
for a constant $\bar{c} > 0$.  To verify the metric entropy conditions, for a given $\mathcal{E}$ and $M>0$, we define the function spaces as in \citet{Jeong2023}, $\mathcal{F}_{\mathcal{E},M} ^{(1)} = \{f \in \mathcal{F}_{\mathcal{E}}: \|B\|_\infty \le M\}$ and $\mathcal{F}_{\mathcal{E},M} ^{(2)} = \{f \in \mathcal{F}_{\mathcal{E}}: \|B\|_\infty > M\}$, where $B$ denotes the vector of outer weights. For $\bar{K}_n \asymp \nicefrac{n \epsilon_n ^2}{\log n}$, we define:
\begin{equation}
    \mathcal{F}_{\bar{K}_n,M} ^{(l)} \coloneqq \bigcup_{\mathcal{E} \in \mathscr{E},K^t \le \bar{K}_n} \mathcal{F}_{\mathcal{E},M} ^{(l)}, \ l=1,2. 
\end{equation}
We construct the sieve $\Theta_n = \mathcal{F}_{\bar{K}_n,M} ^{(1)} \times (n^{-M_2},\exp(M_2 n \epsilon_n ^2))$ for large $M_1,M_2 > 0$. 
\Cref{lem:metricentropy} verifies the metric entropy condition. 
Consequently, (P3) and tail bounds of the inverse gamma distributions imply that $\Pi(\sigma^2 \notin (n^{-2M_2},\exp(2M_2 n\epsilon_n^2)) \exp(\bar{c}' n \epsilon_n^2) \rightarrow 0$. 
Finally, we can conclude that $\Pi(\mathcal{F}_{\star} \setminus \mathcal{F}_{\bar{K}_n, n^C}^{(1)}) \exp(\bar{c}'n \epsilon_n ^2) \rightarrow 0$ by using \Cref{lem:priorcharging}, verifying all three conditions ensuring posterior contraction.
\end{proof}

\end{document}